\documentclass[10pt]{article}
\usepackage[utf8x]{inputenc}

\usepackage[margin=3cm]{geometry}
\usepackage{amsmath,amssymb}
\usepackage[mathscr]{euscript}
\usepackage[usenames,dvipsnames]{color}
\usepackage{pstricks}
\usepackage{graphicx}
\usepackage[all]{xy}
\usepackage{tabularx}

\newtheorem{prop}{Proposition}[section]

\newtheorem{remark}{Remark}[section]
\newtheorem{defi}{Definition}[section]

\newtheorem{lemm}{Lemma}[section]
\newtheorem{thm}{Theorem}[section]
\newtheorem{coro}{Corollary}[section]
\newtheorem{ex}{Example}[section]

\newcommand {\R}{\mathbb{R}}
\newcommand {\Z}{\mathbb{Z}}
\newcommand {\C}{\mathbb{C}}
\newcommand {\N}{\mathbb{N}}

\newcommand {\calD}{{\mathcal D}}
\newcommand {\calm}{{\mathcal M}}
\newcommand {\calo}{{\mathcal O}}
\newcommand {\calr}{{\mathcal R}}

\newcommand {\bfd}{\mathbf{d}}
\newcommand{\delete}[1]{}
\newcommand {\greenf}{\mathfrak{G}}
\newcommand {\gmd} {meromorphic germ of distributions }
\newcommand {\gmds} {meromorphic germs of distributions }

\newcommand {\zb}[1]{\textcolor{blue}{ #1}}

\newcommand{\bbox}{\normalsize {}%
        \nolinebreak \hfill $\blacksquare$ \medbreak \par}
\newenvironment{proof}{\noindent\emph{Proof} ---}{\bbox\vspace{0,15cm}}

\date{}

\begin{document}
\title{Renormalization of Feynman amplitudes on manifolds by spectral zeta regularization and blow-ups.}

\author{Nguyen Viet Dang\footnote{
Institut Camille Jordan (U.M.R. CNRS 5208), Universit\'e Claude Bernard Lyon 1, 
B\^atiment Braconnier, 43, boulevard du 11 novembre 1918, 
69622 Villeurbanne Cedex}
\\
Bin Zhang\footnote{
School of
Mathematics, Yangtze
Center of
Mathematics, Sichuan
University, Chengdu, 610064, P. R. China}
}

\maketitle

\begin{abstract}
Our goal in this paper is to present
a generalization of the spectral zeta regularization
for general Feynman amplitudes on Riemannian manifolds.
Our method uses complex powers of elliptic operators
but involves several complex parameters
in the spirit of the \emph{analytic renormalization} by Speer, 
to build mathematical foundations for the
renormalization of perturbative interacting
quantum field theories.
Our main result 
shows that spectrally regularized
Feynman amplitudes admit an analytic continuation as
meromorphic germs with linear poles 
in the sense of the works of
Guo--Paycha and the second author. 
We also give
an explicit determination 
of the affine hyperplanes
supporting the poles. 
Our proof relies
on suitable resolution 
of singularities of products 
of heat kernels to make them smooth.

As an application of the analytic continuation result, we use a universal
projection from meromorphic germs with linear poles on holomorphic germs 
to construct renormalization maps 
which subtract singularities of Feynman amplitudes of Euclidean fields. 
Our renormalization maps are shown to 
satisfy consistency conditions previously 
introduced in the work of Nikolov--Todorov--Stora 
in the case of flat space--times.
\end{abstract}

\tableofcontents

\section{Introduction.}
%
\paragraph{Zeta regularization.}
Let $M$ be a smooth, compact, connected manifold 
without boundary and $P$ be a symmetric, positive, elliptic pseudodifferential
operator on $M$. 
Later on, we will 
specialize to
Schr\"odinger 
operators of the form $P=-\Delta_g+V$ 
where $-\Delta_g$ is a Laplace
operator and
$V$ is a smooth nonnegative
potential. But the present discussion 
applies to any symmetric, positive, elliptic pseudodifferential
operator $P$. 
Then $P$ admits a discrete spectral resolution~\cite[Lemma 1.6.3 p.~51]{Gilkey} which means there is
an increasing sequence
of eigenvalues $$\sigma (P)=
\{0 \leqslant\lambda_0\leqslant \lambda_1\leqslant \lambda_2 \leqslant \dots \leqslant \lambda_n\rightarrow +\infty \}$$
and corresponding $L^2$-basis of eigenfunctions
$(e_\lambda)_{\lambda\in \sigma(P)}$ so that
$P e_\lambda=\lambda e_\lambda$.
In his seminal work~\cite{Seeley},
Seeley constructed the
complex powers $(P^{-s})_{s\in \mathbb{C}}$ of $P$ as a holomorphic
family
of linear continuous operators acting on suitable scales of
Sobolev spaces on the manifold $M$. 
In particular for $Re(s)\geqslant 0$, $P^{-s}$ is bounded in $L^2(M)$.
Now let us consider the \emph{spectral zeta function} $\zeta_{P}(s)$ which is defined as
the trace $TR\left(P^{-s}\right)$ and coincides with
the series~:
\begin{equation}
\boxed{\zeta_{P}(s)=TR\left(P^{-s}\right)=\sum_{\lambda \in\sigma(P)\setminus \{0\}} \lambda^{-s}.}
\end{equation}
By Weyl's law on the growth of eigenvalues of $P$~\cite[Lemma 1.12.6 p.~113]{Gilkey},
the operator $P^{-s}$ is trace class
and the series $\zeta_{P}(s)=\sum_{\lambda>0} \lambda^{-s}$
converges as a holomorphic function in $s$ on the
half--plane
$Re(s)>\frac{\dim(M)}{\deg(P)}$.
Then Seeley showed that $\zeta_P(s)$
admits an \textbf{analytic continuation}
on the complex plane as a \textbf{meromorphic function}
~\cite[Thm 1.12.2 p.~108]{Gilkey} with simple poles. 
In case $P$ is a \textbf{differential operator}, 
$\zeta_P(s)$ is holomorphic at $s=0$.
This result shows one of the first instances of the power of zeta regularization, 
where we can regularize the divergent series $\sum_{\lambda\in \sigma(P)} 1$
and obtain the value $\zeta_P(0)$ of the spectral zeta function $\zeta_P$ at $s=0$. 
More importantly,
the residues of $\zeta_{P}(s)$ at its poles can be expressed as multiple of
integrals over $M$ of \textbf{local invariants} of the operator $P$~\cite[p.~299-303]{Atiyah-73} and are intimately
related to the heat invariants of $P$~\cite[Thm 1.12.2 p.~108]{Gilkey}.

\paragraph{From zeta regularization to regularized traces.}
In the same spirit, zeta regularization techniques were also used in 
global analysis to 
construct regularized traces for certain algebras of pseudodifferential operators. 
The above result of Seeley on the analytic continuation of $TR(P^{-s})$ has been generalized to
\emph{canonical traces} on pseudodifferential operators by Kontsevich-Vishik~\cite{kontsevichdeterminants}, 
to
study anomalies of regularized zeta determinants with related works
by 
Lesch~\cite{lesch1999noncommutative} among many authors. Then general 
types of tracial anomalies were discussed
in~\cite{MelroseNistor,paychacardonaducourtioux,PaychaScottchernweil}, sometimes in relation with
quantum field theory, 
and finally a general notion of trace for 
\emph{holomorphic families of pseudodifferential operators} appears in
the work of Paycha--Scott~\cite{paycha2007laurent}.
An important object underlying all these constructions is the
notion of noncommutative residue
for any pseudodifferential operator $A$.
This noncommutative residue can be defined by 
zeta regularization using 
complex powers of elliptic operators as follows. 
Choose any symmetric, positive, elliptic differential operator $P$,
then the noncommutative residue of $A$ is defined as the residue at $s=0$
of the meromorphic continuation of the trace $TR\left(AP^{-s}\right)$,
and is given by a local formula in the symbol of $A$.
In his seminal works, Wodzicki~\cite{wodzicki1987,wodzicki1984} proved that up to constant,
this residue is the unique trace on the algebra of pseudodifferential
operators. It plays a central role in global analysis and noncommutative geometry.
We refer the reader to the monographs~\cite{paycha2012regularised, scott2010traces} for
further details on these topics.

\paragraph{Zeta regularization for partition functions.}

Already in the simple case of spectral zeta
functions of the Laplace--Beltrami operator, these regularization methods turn out to be
extremely useful to study Euclidean quantum fields on Riemannian manifolds. 
In the mathematical physics literature,
zeta regularization was
first applied to quantum field theory on curved spaces by Hawking~\cite{hawking1977zeta} to give a
definition of the partition function of Euclidean QFT. 
It can also be used to
give a mathematical model of the
Casimir effect~\cite{fulling2007vacuum}. For topological quantum field theories,
following the seminal work of Ray--Singer~\cite{raysinger1971}
on analytic torsion,
it was soon realized by Schwarz that one can  
define and calculate the partition function of some
abelian BF theories~\cite{schwarz1978partition}
using zeta regularized determinants.
Formally, for some flat bundle $(E,\nabla)$ over some smooth
compact manifold $M$ of dimension $d$, his formula for the partition 
function of the BF theory reads
$$ \int_{(A,B)\in \Omega^k(M,E)\times \Omega^{n-k-1}(M,E)} \exp\left(-\int_M B\wedge d^\nabla A \right)=\prod_{k=0}^d  \det\left(\Delta^{(k)} \right)^{(-1)^kk/2}  $$
where $d^\nabla$ is the twisted differential acting on $\Omega^\bullet(M,E)$ 
and the right hand side is the Ray--Singer analytic torsion
of the flat bundle $(E,\nabla)\mapsto M$ which is a topological invariant~\cite[(10) p.~9]{mnev2014lecture}.  
Then Witten generalized the above
work of Schwarz by showing that the perturbative partition
function of Chern--Simons theory involved
the Ray--Singer analytic torsion and also the eta invariant of Atiyah--Patodi--Singer. Since
the formula looks quite complicated, 
we refer the reader to~\cite[(12) p.~9]{mnev2014lecture}.
But the important point is that 
the formula involves zeta 
regularized
determinants.
The main idea underlying the above results
is that partition functions are \textbf{formally} expressed as
functional integrals on some space of fields, 
these partition functions are
then identified with 
regularized determinants 
of elliptic operators. 
For instance, in the case of the Dirichlet action 
functional $S(\varphi)=\frac{1}{2}\int_M\varphi(-\Delta_g)\varphi dv(x)$
where $-\Delta_g$ is the Laplace--Beltrami operator and $dv$
the Riemannian volume, the partition function $Z$ reads~:
$$Z=\int d\varphi \exp\left(-\frac{1}{2}\int_M \varphi(-\Delta_g)\varphi dv\right)=\det(-\Delta_g)^{-\frac{1}{2}} $$
where $\det(-\Delta_g)$ may be defined as
$\exp(-\zeta^\prime(0))$ where $\zeta$ is the regularized zeta
function of the elliptic operator $(-\Delta_g)$ appearing in the definition of the partition function.
 
 For applications in mathematical physics and in the present work,
a particular role will be played
by complex powers of generalized Laplacians (more generally elliptic, positive, self--adjoint operators of order $2$) and their relation with the
heat kernel asymptotics. 
These methods based on the local asymptotic expansion of the heat kernel
are crucial in the local index theory~\cite{Berline-04} and 
are also used in the works~\cite{BarMoroianu,Baereinstein} 
to give a purely spectral definition
of the Einstein--Hilbert
action functional following~\cite{kalau1995gravity, ackermann1996note, kastler1995dirac}.
 
 Another interesting physical property of zeta regularization is its natural covariance which
is why it was used in the first place by Hawking. Indeed,
for any diffeomorphism
$\Phi:M\mapsto M$, 
the spectrum $\sigma\left(-\Delta_{\Phi^*g}\right)$ of 
$-\Delta_{\Phi^*g}$ on $(M,\Phi^*g)$ coincides with the spectrum
of $-\Delta_g$ on $(M,g)$ and $\sigma(-\Delta_g)$ is thus an \textbf{invariant of the Riemannian structure} $(M,g)$\footnote{The space of Riemannian structure if the set of pairs $(M,g)$ quotiented by isometries}.
Therefore zeta regularization is a coordinate independent regularization scheme
which depends only on the spectral properties of the Laplacian
which in turns is entirely specified by the Riemannian structure $(M,g)$.

\paragraph{Renormalization in quantum field theory.}

The present paper is written for analysts and
does not require any background in physics or
quantum field theory.
We present our results in a
purely mathematical form.
However, we felt that
for readers with some interest in QFT, it would be
preferable to
present some
physical motivations and the uninterested reader can skip the present paragraph.
QFT is a general framework
aimed at describing the fundamental forces and particles.
In QFT, we are given some graphs called
\emph{Feynman graphs} which
pictorially represent complicated interaction
process between various particles and we associate to 
every graph
$G$ some number $c_G$, called \emph{Feynman amplitude}, which is often given by some divergent
integral
as soon as the graph $G$ contains any loops.
The issue is that the above zeta regularization methods can only be used to
renormalize one loop graphs
as discussed in~\cite[1.4 p.~10]{bytsenko2003analytic}. 
For interacting QFT's, it is not enough to 
regularize only the partition function and one loop graphs,
one must renormalize
amplitudes whose corresponding graphs contain
an arbitrary number of loops.
For instance in quantum electrodynamics (QED) which is the QFT
describing the interaction of light and matter,
the computation of the probability amplitude
of some scattering
process for two incoming and two outgoing electrons
is represented by the following Feynman diagram~:
\begin{center}
\includegraphics[scale=0.4]{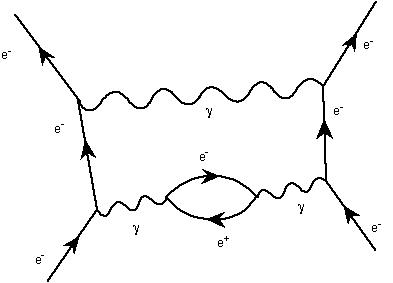}
\end{center}
where 
the electrons are denoted by $e^-$,
positrons by $e^+$ and photons
by $\gamma$. 
The corresponding Feynman amplitude
is given by some product of electron propagators, represented
by the straight lines, and photon
propagators represented by the 
wiggly lines. These propagators are distributions on $\R^4\times \R^4$
valued in $4\times 4$ matrices. 

For the sake of simplicity, we limit ourselves to
scalar theories in the present paper. In these theories, unlike in gauge theories, 
there is only one scalar valued propagator which is denoted by $\greenf$ in the sequel. 
The topology of the Feynman graphs that we encounter is dictated
by the interaction of the theory. For instance in the massless $\phi^4$ theory, the only Feynman graphs
that we encounter
have vertices of degree $4$. 
Our goal is to use spectral zeta regularization to renormalize multiple
loop amplitudes for Euclidean QFT on Riemannian manifolds 
with the aim to relate them to geometric
invariants of Riemannian 
manifolds which is
the subject of future work of the authors.
Our starting point is the work of Eugene Speer
on analytic renormalization in QFT~\cite{speer1968, speer1976, speerseminar} who found an
alternative formulation of the usual BPHZ renormalization algorithm,
based on analytic regularization with several complex parameters.
The analytic structure 
of the regularized amplitude in these
variables 
encodes 
rich algebraic structure so that a renormalized
amplitude
may be defined by the application of a universal
projector, independent of the graph in question, to the regularized amplitude.
Indeed, we will show that
regularized amplitudes are
\textbf{meromorphic germs with linear poles} and
in subsection \ref{ss:renormproj},
we will describe  
a straightforward way of subtracting the divergent part of the regularized amplitudes while 
keeping only the holomorphic part.
Then renormalization will
be reformulated in definition
\ref{d:renormmapszeta}
as the 
evaluation at some poles
of the holomorphic part
of the regularized
amplitude.
This projection is a useful
substitute to the BPHZ algorithm and the
method pioneered by Connes--Kreimer
based on Hopf algebras and Birkhoff factorizations.
In our work, a common point with
the BPHZ algorithm and Speer's work, is that we rely on Hepp sectors and 
resolution of singularities arguments.

Let us show how the idea of \emph{analytic renormalization} works in an example
on flat space.
On Euclidean space $\R^4$,
the Green function of the Laplace operator reads
$ \greenf(x,y)=C Q^{-1}(x-y)$ where $C$ is some constant
and $Q$ is the quadratic form $Q(v)=\sum_{i=1}^4 v_i^2$.
On configuration space $(\mathbb{R}^4)^6$,
the Feynman rules assign to the graph
\begin{center}
\includegraphics[scale=0.3]{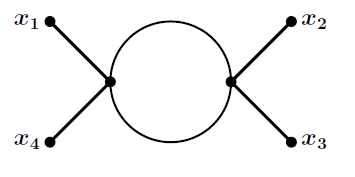}
\end{center}
the corresponding amplitude
$$T(x_1,x_2,x_3,x_4)=\int_{(y_1,y_2)\in (\mathbb{R}^4)^2} 
\greenf(x_1,y_1)
\greenf(x_2,y_1)
\greenf^2(y_1,y_2)\greenf(y_2,x_3)\greenf(y_2,x_4)d^4y_1d^4y_2 ,$$
which
is given by some formal
product of Green functions.
To get rid
of the infrared divergence due to the fact that we integrate in some
infinite volume $(\mathbb{R}^{4})^2$, one may either introduce
a sharp cut--off by replacing $\mathbb{R}^4$ by a finite box, or
we may as well insert some 
smooth compactly supported cut--off
function $g\in C^\infty_c(\R^4)$ 
for each variable $(y_i)_{i\in \{1,2\}}$ corresponding to the internal vertices of the Feynman graph as
follows~:
$$\boxed{T(x_1,x_2,x_3,x_4)=\int_{(\mathbb{R}^4)^2} \greenf(x_1,y_1)
\greenf(x_2,y_1)
\greenf^2(y_1,y_2)\greenf(y_2,x_3)\greenf(y_2,x_4)g(y_1)g(y_2)d^4y_1d^4y_2 .}$$
In fact, it is natural to view the full amplitude $\greenf(x_1,y_1)
\greenf(x_2,y_1)
\greenf^2(y_1,y_2)\greenf(y_2,x_3)\greenf(y_2,x_4)$ 
as a \textbf{distribution} in 
$\mathcal{D}^\prime\left(\left(\R^4\right)^6\right)$, so we may think that 
we insert some smooth compactly supported cut--off
function $g(y_1)g(y_2)$ on $(\R ^4)^2$ so that 
$T(x_1,x_2,x_3,x_4)$ 
is well--defined as the
\textbf{pushforward} of the product
$\greenf(x_1,y_1)
\greenf(x_2,y_1)
\greenf^2(y_1,y_2)\greenf(y_2,x_3)\greenf(y_2,x_4)g(y_1)g(y_2)d^4y_1d^4y_2$
along the fibers of
the projection $(\R^4)^6\mapsto (\R^4)^4$.

In terms of the quadratic function $Q$, the above
amplitude reads~:
\begin{eqnarray*}
\int_{(y_1,y_2)\in (\mathbb{R}^4)^2} Q^{-1}(x_1,y_1)
Q^{-1}(x_2,y_1)
Q^{-1}(y_1,y_2)Q^{-1}(y_1,y_2)\\
\times Q^{-1}(y_2,x_3)Q^{-1}(y_2,x_4)g(y_1)g(y_2)d^4y_1d^4y_2 .
\end{eqnarray*}
Now for each $Q^{-1}$ in factor in the amplitude, we shall introduce a complex power $s$ as follows~:
\begin{eqnarray*}
T(\mathbf{s})=\int_{(y_1,y_2)\in (\mathbb{R}^4)^2} 
Q^{-s_1}(x_1,y_1)
Q^{-s_2}(x_2,y_1)
Q^{-s_3}(y_1,y_2)Q^{-s_4}(y_1,y_2)\\
\times Q^{-s_5}(y_2,x_3)Q^{-s_6}(y_2,x_4)g(y_1)g(y_2)d^4y_1d^4y_2 
\end{eqnarray*}
where the new amplitude depends 
on $\mathbf{s}=(s_1,\dots,s_6)\in \mathbb{C}^6$.
For $Re(s_i)_{1\leqslant i\leqslant 6}$ large enough, one can easily see that
the amplitude defining $T$ is integrable. The main result of Speer
is the fact that $T(\mathbf{s})$ admits an analytic continuation
in $\mathbf{s}\in \mathbb{C}^6$ as a meromorphic function with linear poles.
Then he shows that $T(\mathbf{s})$ decomposes as the sum of a singular part and a holomorphic part
at $\mathbf{s}=(1,\dots,1)\in \C^6$ and renormalization consists in subtracting the singular part
and evaluating at $(1,\dots,1)\in \C^6$.

The main goal of the present paper is to combine the
methods from zeta regularization
to present a generalization of analytic renormalization
to general Riemannian manifolds. 
Then we will show that
the renormalization defined satisfies the consistency axioms of
Nikolov--Todorov--Stora in~\cite{NST}
inspired by the
seminal works of Epstein--Glaser~\cite{Epstein}.

\section{Main results.}

In the present section, we introduce the 
main objects of study
and state the main results
of our work. We define first Feynman amplitudes, next we explain how to
implement a zeta regularization with several complex parameters then 
we state the first main analytic continuation Theorem and we finally
give a simplified
version of our second main theorem
concerning applications of the analytic continuation 
result to renormalization in QFT.

\subsection{Feynman amplitudes }

We work on a compact, connected Riemannian manifold $(M,g)$ without boundary, 
the Laplace--Beltrami operator
is denoted by $\Delta_g$ and $C^\infty_{\geqslant 0}(M)$ denotes smooth, nonnegative functions
on $M$.
For a potential $V\in C^\infty_{\geqslant 0}(M)$, 
it is well--known
that the Schr\"odinger operator $P=-\Delta_g+V$ is
a second order, symmetric, positive, elliptic differential operator which defines a unique
unbounded, self--adjoint operator acting
on $L^2(M)$~\cite[p.~34-35]{taylor2013partial}.
We now generalize the Feynman rules to this case. That is, to every graph we associate a formal product
of Green kernels of the
operator $P$. Since on a general manifold, there is no Fourier transform,
our Feynman rules
are just the Riemannian versions of the Euclidean Feynman rules
in \textbf{position space}
of~\cite[definition 2.1]{ceyhan2012feynman} (see also~\cite{Costello-11}).

\begin{defi}
[\textbf{Feynman rules}]
 Let $\greenf (x,y)$ denote the Green kernel of the operator $P$,
then for a graph $G$ with the set of vertices $V(G)$
and the set of edges $E(G)$, if for any edge $e\in E(G)$, the vertices incident to $e$ are $i(e)$ and $j(e)$
and
$G$ has no self--loops,
then the Feynman amplitude associated to $G$ is defined as
\begin{equation}
\boxed{t_G=\prod_{e\in E(G)}\greenf (x_{i(e)},x_{j(e)})}
\end{equation}
as a $C^\infty$ function on $M^{V(G)}\setminus\{\text{all diagonals}\}$.
\end{defi}

\begin{remark}
Since the Green kernel
$\greenf$ is symmetric in its variables,
$t_G$ is well-defined.
The graphs are not allowed to have self--loops since the 
Green function
$\greenf$
is not well--defined on the diagonal hence cannot be evaluated at coinciding
points.
The above Feynman rules correspond to a perturbative
Euclidean QFT
where the Lagrangian is already Wick renormalized
which explains why self--loops (also called tadpoles in the
physics literature) are excluded.
\end{remark}

\subsection{Multiple spectral zeta regularization.}

The operator $P^{-s}$ is defined as a spectral function of the
operator $P$ in a very simple way following~\cite[equation (1.12.13) p.~107]{Gilkey}~:
\begin{defi}[Complex powers]
For $Re(s)\geqslant 0$, for every $u\in L^2(M)$, decompose
$u$ in the orthonormal basis $(e_\lambda)_\lambda$ of $L^2(M)$
given by the eigenfunctions of $P$.
Then
$$P^{-s}u=\sum_{\lambda\in\sigma(P)\setminus\{0\}}\lambda^{-s}\langle u,e_\lambda \rangle e_\lambda $$
where the sum on the right hand side converges absolutely in $L^2(M)$ since the eigenvalue $\lambda_i
\underset{i\rightarrow+\infty}{\longrightarrow} +\infty$
hence the sequence $(\lambda_i^{-s})_i$ remains bounded.
\end{defi}
The Schwartz kernel of $P^{-s}$ is then by definition~:
\begin{equation}
\boxed{\greenf ^s(x,y)=\sum_{\lambda\in \sigma(P)\setminus\{0\}} \lambda^{-s}e_\lambda(x) e_\lambda(y)}
\end{equation}
where we abusively denoted by $\greenf ^s(x,y)$ an actual \emph{distribution}
$\greenf^s\in \mathcal{D}^\prime(M\times M)$ and the series on the r.h.s. converges in
$\mathcal{D}^\prime(M\times M)$. We will later see that
$\greenf ^s$ is actually a function on $M\times M$ for $Re(s)>\frac{d}{2}$ where $d=\dim(M)$.
We shall generalize this regularization to the case $M=\mathbb{R}^d$ with flat Euclidean metric $g$
and $P=-\Delta_g+m^2, m\in \R_{\geqslant 0}$.
Our definition of $\greenf^s$ in the flat case is similar to the compact case
since we define $\greenf^s$ with complex powers of the Laplace operator~:
\begin{defi}[Complex powers for flat space.]
\label{d:complexflat}
If $M=\mathbb{R}^d$, $g$ is a constant quadratic form and $m\in \R_{\geqslant 0}$ is a mass, then
we set~:
$$ \boxed{\greenf^s(x,y)=\frac{1}{(2\pi)^d} \int_{\mathbb{R}^d} \frac{e^{i\langle (x-y),\xi \rangle}}{(g^{\mu\nu}\xi_\mu\xi_\nu +m^2)^s} d^d\xi=  \frac{1}{\Gamma(s)} \int_0^{+\infty} \frac{1}{(4\pi t)^{\frac{d}{2}}} e^{-\frac{\langle x-y,x-y \rangle_g}{4t}}e^{-tm^2} t^{s-1}dt. }$$
\end{defi}

It is immediate from the above formulas that $\greenf^s$ is the Schwartz kernel of
$(-\Delta_g+m^2)^{-s}$ and that when $s=1$, we recover the Green function of the
operator $(-\Delta_g+m^2)$.

\begin{defi}
[\textbf{Regularized Feynman rules}] Under the above assumptions,
we denote by $P^{-s}$ the complex powers of $P$ and
by $  \greenf  ^s(x,y)\in\mathcal{D}^\prime(M\times M)$ the corresponding Schwartz kernel. Then for 
a graph $G$ with vertex set $V(G)$ and edge set $E(G)$, the regularized Feynman amplitude reads
\begin{equation}
\boxed{t_G(s)=\prod_{e\in E(G)}\greenf ^{s_e}(x_{i(e)},x_{j(e)})}
\end{equation}
which is in $C^\infty(M^{V(G)}\setminus \{\text{all diagonals}\})$.
\end{defi}
\begin{remark}
We will see later in Lemma \ref{keylemma} that
$\greenf^s$ is actually in $C^k(M\times M)$ for $Re(s)$ large enough hence it follows
that the above Feynman rules also make sense for graphs $G$ with self--loops when
$Re(s)$ is large enough
which was not true for $s=1$ since $\greenf$ would be a distribution singular on
the diagonal.
\end{remark}

Let us state our first main Theorem~:
\\
\\
\fbox{
\begin{minipage}{0.94\textwidth}
\begin{thm}
\label{mainthmintro}
Let $(M,g)$ be a smooth, compact, connected Riemannian manifold without boundary of dimension $d$, $dv(x)$ the Riemannian volume and $P=-\Delta_g+V$, $V\in C^\infty_{\geqslant 0}(M)$ 
or $M=\mathbb{R}^d$ with a constant metric $g$ and $P=-\Delta_g+m^2, m\in \R_{\geqslant 0}$.
Then for every graph 
$G$, on the configuration space $(M^{V(G)},g_{V(G)})$ endowed with the product metric $g_{V(G)}$, 
and product volume form $dv_{M^{V(G)}}$,
for any test function $\varphi\in C^\infty(M^{V(G)})$,
\begin{equation}
s\mapsto \int_{M^{V(G)}} t_G(s)\varphi dv_{M^{V(G)}}
\end{equation}
can be analytically continued near $(s_{e}=1)_{e \in E(G)}$ as
a \textbf{meromorphic germ with possible linear poles}
on the hyperplanes of equation $\sum_{e\in G^\prime}s_e-\vert E(G^\prime)\vert=0$
where $G^\prime$ is a subgraph of $G$ such that
$2\vert E(G^\prime)\vert-b_1(G^\prime)d\leqslant 0 $, 
$\vert E(G^\prime)\vert$ is the number of edges
in $G^\prime$ and $b_1(G^\prime)$ the first Betti number of $G^\prime$.
\end{thm}
\end{minipage}
}
\\
\\
To recover renormalized Feynman amplitudes, we follow the strategy
of~\cite[2.2]{dang2015complex}. We cannot evaluate
$t_G(s)$ at $(s_e=1)_{e\in E(G)}$ since it might belong to the polar set
of $t_G$. However, applying the machinery from~\cite{guopaychazhang2015}
allows us to subtract
the polar part of $t_G(s)$ at $(s_e=1)_{e\in E(G)}$
while keeping
a holomorphic part.
This is based on an extension of the framework
of~\cite{guopaychazhang2015} to distributions
valued in meromorphic germs with linear poles constructed
in paragraph \ref{ss:decompositiongermsdistrib}. 
Then to recover the renormalized Feynman 
amplitude, it suffices to
evaluate the holomorphic part at $(s_e=1)_{e\in E(G)}$. 
Following Speer~\cite[section 3]{speer1976}, analytic 
renormalization will
be reformulated in definition
\ref{d:renormmapszeta}
as the 
evaluation at some poles
of the holomorphic part
of the regularized
amplitude.
This idea was recently abstracted
in the works~\cite{clavier2017algebraic,guo2017counting} in a purely algebraic way where the composition 
of a
projection on the holomorphic part and the evaluation at $(s_e=1)_{e\in E(G)}$ is called
\textbf{evaluator}~\cite[1.3 p.~6]{guo2017counting}. 
The renormalization $\mathcal{R}(t_G)$ of some
amplitude $t_G$ is the
composition of the operations
summarized
in the following diagram~:
\begin{eqnarray*}
\boxed{t_G \overset{\small{\text{regularization}}}{\longrightarrow} t_G(s)  \overset{\small{\text{projection on holomorphic part}}}{\longrightarrow} \pi(t_G(s)) 
 \overset{\small{\text{evaluation at }s=s_0}}{\longrightarrow} \mathbf{ev}|_{s_0}\left(\pi(t_G(s)) \right)=\mathcal{R}(t_G),}
\end{eqnarray*}
where $s_0=(s_e=1)_{e\in E(G)}$.

In section~\ref{s:renormalization}, we apply
the above ideas to
the renormalization of
quantum field theories on Riemannian manifolds
and show the existence of a collection of renormalization maps $(\mathcal{R}_{M^I})_{I\subset \mathbb{N}}$ that roughly
assign to each graph $G$ a renormalized amplitude in $\mathcal{D}^\prime(M^{V(G)})$
such that
the renormalization maps satisfy the consistency axioms~\ref{d:Functionalequations} which
come from the work of Nikolov--Todorov--Stora~\cite{NST}.
Let us state a simplified version of our second main Theorem~\ref{t:renormtheorem}~:
\\
\\
\fbox{
\begin{minipage}{0.94\textwidth}
\begin{thm}
\label{t:renormsimplified}
Let $(M,g)$ be a smooth, compact, connected Riemannian manifold without boundary of dimension $d$, $dv(x)$ the Riemannian volume and $P=-\Delta_g+V$, $V\in C^\infty_{\geqslant 0}(M)$ 
or $M=\mathbb{R}^d$ with a constant metric $g$ and $P=-\Delta_g+m^2, m\in \R_{\geqslant 0}$.
Then for every
graph $G$, on the configuration space $M^{V(G)}$ endowed with the 
product volume form $dv_{M^{V(G)}}$~:
\begin{itemize}
\item there exists distributions
$\pi(t_G(s)),(1-\pi)(t_G(s))$ such that
for any test function $\varphi\in C^\infty(M^{V(G)})$, we have a unique decomposition
\begin{eqnarray*}
\int_{M^{V(G)}} t_G(s)\varphi dv_{M^{V(G)}}=\underset{\text{meromorphic germ}}{\underbrace{\int_{M^{V(G)}} (1-\pi)(t_G(s))\varphi dv_{M^{V(G)}}}}+ \int_{M^{V(G)}} \pi(t_G(s))\varphi dv_{M^{V(G)}}
\end{eqnarray*}
where $s\mapsto \int_{M^{V(G)}} \pi(t_G(s))\varphi dv_{M^{V(G)}}$ is a holomorphic
germ at $s_0=(s_e=1)_{e\in E(G)}\in \mathbb{C}^{\vert E(G)\vert}$.

\item If $\varphi\in C_c^\infty(M^{V(G)}\setminus\text{ all diagonals})$ then
\begin{equation}
\lim_{s\rightarrow s_0}\int_{M^{V(G)}} \pi(t_G(s))\varphi dv_{M^{V(G)}}=\int_{M^{V(G)}} t_G\varphi dv_{M^{V(G)}}
\end{equation}
which means
$\lim_{s\rightarrow s_0} \pi(t_G(s))$ is a distributional extension
of $t_G\in C^\infty(M^{V(G)}\setminus\text{ all diagonals}) $.
\end{itemize}
\end{thm}
\end{minipage}
}
\\
\\
Then the reader is referred to
Theorem~\ref{t:renormtheorem}
where we prove
many important properties
enjoyed by the renormalized amplitudes
$\lim_{s\rightarrow s_0} \pi(t_G(s))$. 
The most important being
a factorization equation
appearing in definition \ref{d:Functionalequations}
which translates in mathematical terms
the essential property of locality in
Euclidean QFT. 

\paragraph{Related works.}

In recent works of Hairer~\cite{HairerBPHZ} and Pottel~\cite{pottel2017bphz,pottel2017analytic}, the authors
give analytic treatments of the BPHZ algorithm.
Like in the present paper, they also start
from Feynman amplitudes in \textbf{position space}
but Hairer works on $\mathbb{R}^d$ with abstract kernels $K$
with specific singularity along diagonals,
whereas we work on Riemannian manifolds but we
limit our discussion to Green kernels of
Laplace type operators. He also uses Hepp sectors
to perform some kind of multiscale analysis
to analyze the divergences of the Feynman amplitudes.
It would be
interesting to compare the renormalization maps
defined in the present paper with the valuations in
Hairer's paper~\cite{HairerBPHZ} and definition \ref{d:Functionalequations}
with the consistency axioms of~\cite{HairerBPHZ}.

Our treatment of renormalization bears a strong inspiration from the seminal work of
Epstein--Glaser~\cite{Epstein} who were among the first to understand
the central role of causality (this is replaced in the current work by locality) in perturbative
renormalization. Their work was generalized by Brunetti--Fredenhagen~\cite{Brunetti2} to curved
space times while the crucial physical notions
of covariance of the renormalization were adressed in the works of Hollands--Wald~\cite{Hollands, Hollands2}.
A recent investigation of the Epstein--Glaser renormalization using
resolution of singularities
can be found in the thesis of Berghoff~\cite{berghoff,berghoff2015} clarifying some previous attempts~\cite{Bergbauerdip, Bergbauer2009}.
Our results seem to be more general 
since we work in the manifold case and we resolve singularities by hand 
instead of
using the compactifications of configuration 
space of Fulton--McPherson and de Concini--Procesi.

There is a famous interpretation of the BPHZ renormalization in terms of Hopf algebras
pioneered by Connes--Kreimer~\cite{Connes,CKI,CKII}. This approach using dimensional regularization
works essentially in momentum space and does not generalize in a straightforward way
to curved spaces. Motivated by problems from number theory,
Marcolli--Ceyhan~\cite{ceyhan2012feynman} managed to reformulate the Hopf--algebraic approach on configuration space.

Renormalization in the Riemannian setting was recently discussed 
in the book by Costello~\cite{Costello-11}, however
it seems
his proof of subtraction of counterterms contains some gaps
that were fixed by Albert who also extended Costello's work to manifolds with
boundary~\cite{albert2016heat, albertthesis}.

\paragraph{Perspectives.}

A natural extension of our results 
would be to prove
an analytic continuation result
for Feynman amplitudes made from 
Schwartz kernels of holomorphic families of pseudodifferential operators 
in the sense of Paycha--Scott~\cite{paycha2007laurent}
generalizing the
Schwartz kernels of complex powers of
Laplace operators. 
For the sake of simplicity,
we limited ourselves to complex powers of Laplace operators
because of their explicit relation with heat kernels and
leave it to another work for the investigation of the more general case.
Another interesting situation is when the manifold $M$ is noncompact
with specific asymptotic structure as in scattering theory. Probably in this case, we would need to use resolvents to define
complex powers.

It would also be very interesting to test our proof in the Lorentz case
with the \emph{Feynman propagator} instead of the Green's function of the Laplacian.
Then a
natural question would be what is the substitute in the Lorentz case
for the complex powers of the Laplace operators ?
The first author defined
complex regularization of Feynman propagators in some previous work~\cite{dang2015complex} 
on \textbf{analytic Lorentzian space--times}
under some very restrictive assumptions of \textbf{geodesic convexity}.
This was based on the Hadamard parametrix for the Feynman propagator.
From our point of view, it would be preferable to define
some
complex regularization scheme on
\textbf{smooth} Lorentzian space--times which are \textbf{not necessarily geodesically convex}. 
The scheme
should be manifestly covariant as spectral regularization on
Riemannian manifolds. Probably, this would be based on the recent
results of~\cite{bar1506index,gerard2016massive,gerard2016feynman,gerard2017hadamard,derezinski2016feynman,derezinski2017evolution,vasy2015quantum} on the analytic structure of
Feynman propagators.
Another
interesting direction is to
investigate if it is possible to
renormalize the amplitudes in Euclidean
theory then
perform a geometric Wick rotation as in
Gérard--Wrochna~\cite{gerard2017analytic}
to build renormalized amplitudes of the
corresponding
Lorentzian QFT.

\paragraph{Acknowledgements.}

We would like to thank
M. Wrochna, D. H\"afner, Y. Colin de Verdière, E. Herscovich,
for many interesting comments they made to one of us (NVD)
when we first presented this work in Grenoble.
Also big thanks to C. Guillarmou and S. Paycha for answering many
of our questions on
regularization using pseudodifferential powers.
The first author wants to thank Dang Nguyen Bac for many discussions 
related to blow--ups and algebraic geometry.
Finally, NVD acknowledges the ANR-16-CE40-0012-01 grant 
for financial support.

\section{Preliminaries.}

The goal of the present section
is to introduce
the language of meromorphic germs with linear
poles and give the main definitions since
meromorphic germs appear in the formulation of Theorem~\ref{mainthmintro}. 
We also introduce their
distributional counterpart
which we call \gmds which is the fundamental object 
needed for the proof of Theorem~\ref{mainthmintro}.
The \gmds are essentially distributions depending on some parameter
$s\in \mathbb{C}^p, p\in \mathbb{N}$, such that when they are paired with some test function $\varphi$, they
give meromorphic germs in $s$.

\subsection{Meromorphic functions with linear poles.}

In this paper, all
meromorphic functions of
several variables $s=(s_1,\dots,s_p)\in\mathbb{C}^p$
have singularities
along
unions of affine hyperplanes.
In fact, we will work with 
\textbf{meromorphic germs with linear poles}
in the terminology
of~\cite{guopaychazhang2015}.
We work in the space $\mathbb{R}^p$, and with the standard complex structure
on $\mathbb{C}^p=\mathbb{R}^p \otimes \mathbb{C}$, let $(\mathbb{R}^p)^*$ be the dual.
%
In the sequel, holomorphic functions on 
some domain $\Omega\subset \mathbb{C}^p$ and holomorphic 
germs at $s_0\in \mathbb{C}^p$ 
are denoted by $\mathcal{O}(\Omega)$ and 
$\mathcal{O}_{s_0}(\mathbb{C}^p)$
respectively.

\begin{defi}[meromorphic germs]
Let $s_0\in \mathbb{R}^p$,
then $f$ is a meromorphic germ with (real) \textbf{linear poles} at $s_0$ if there are vectors
$(L_i)_{1\leqslant i\leqslant m}$ in $(\R^p)^*$,
such that
\begin{equation}
(\prod_{i=1}^m L_i(.-s_0))f\in \mathcal{O}_{s_0}(\C^p).
\end{equation}
Meromorphic germs with linear poles at $s_0\in \mathbb{C}^p$
are denoted by 
$\mathcal{M}_{s_0}(\mathbb{C}^p)$.
\end{defi}

Geometrically such a meromorphic germ $f$ is singular
along some arrangement of affine hyperplanes
$\{s\in\mathbb{C}^p \text{ s.t. } L_i(s-s_0)=0 \}_{1\leqslant i \leqslant m},$
intersecting at the point $s_0$.

\subsection{Meromorphic germs of distributions.}
\label{ss:meromgermsdistrib}
In this paper, we deal with families of distributions $t(s)$ 
on some smooth second countable manifold $X$ without boundary, depending meromorphically on
some parameter $s$ and whose poles are linear.
We will also call them distributions valued in meromorphic germs with linear poles and will
denote the space of such families by 
$\mathcal{D}^\prime(X,\mathcal{M})$.
We devote this subsection to their proper definition.
Our plan is to give the definition gradually starting from holomorphic objects.
For a smooth manifold $X$ with given smooth density $dv$, 
we will use $\calD^\prime (X)$ to denote the space of distributions on $X$ and 
$\calD^\prime(X)$ is defined
in the present paper as the topological dual of $C_c^\infty(X)\otimes dv$ which is 
the space of smooth, compactly supported densities. But in many situations where the density
is explicitely given by the geometric problem, we may equivalently
think of distributions as the dual of $C_c^\infty(X)$.

\paragraph{Holomorphic families of distributions.}
Before we discuss \gmds, let us start smoothly by 
defining distributions depending holomorphically on some extra parameter.
\begin{defi}[holomorphic families]
Let $\Omega\subset \mathbb{C}^p$ be a complex domain, and $X$ be a smooth manifold, a holomorphic family of
distributions on $X$ parametrized by $\Omega$
is a family $(t(s))_{s\in \Omega}$ of distributions on $X$, such that
for every test function $\varphi\in C_c^{\infty}(X)$,
$s\mapsto \langle t(s),\varphi \rangle$ defines a holomorphic function on $\Omega$.
Such set of holomorphic families of distributions will be denoted by
$\mathcal{D}^\prime(X,\mathcal{O}(\Omega))$.
\end{defi}
Then we next introduce a variant of the above definition 
involving distributions whose distributional order is bounded by some
integer $m$. 
\begin{defi}[holomorphic families with bounded order]
Let $m$ be an integer and $X$
a smooth manifold, then a distribution $t$ is of order bounded above by $m$ on $X$ if $t$ defines a
continuous linear function on $C_c^{m}(X)$. For a complex domain $\Omega $, a holomorphic family of
distributions $(t(s))_{s\in\Omega}$ of order bounded above by $m$ on $X$,
is a family $(t(s))_{s\in \Omega}$ of distributions of order bounded above by $m$ such that
for every test function $\varphi\in C_c^{m}(X)$,
$s\mapsto \langle t(s),\varphi \rangle$ defines a holomorphic function on $\Omega$.
This set is denoted by $\mathcal{D}^{\prime,m}(X,\mathcal{O}(\Omega))$.  
\end{defi}

Once we defined holomorphic families of distributions where the complex parameter lives on some domain $\Omega$ containing 
some element $s_0$, it is natural to give a definition where
we want to forget the information about $\Omega$ and localize around $s_0$. 
We thus work at the level of 
holomorphic germs near $s_0$.
\begin{defi}[holomorphic germs]
A holomorphic germ at a point $s_0\in \mathbb{C}^p$ of distributions on $X$ 
is an equivalence class of holomorphic families of distributions on $X$ 
w.r.t. the natural equivalence relation: $(t(s))_{s\in\Omega _1}\sim (u(s))_{s\in\Omega _2}$ 
if there exists $\Omega _3\subset \Omega _1\cap \Omega _2$ 
such that $s_0\in \Omega _3$ and $t(s) =u(s)$ 
for all $s\in\Omega _3$. This set is denoted by 
$\mathcal{D}^\prime(X,\mathcal{O}_{s_0}(\C^p))$.
\end{defi}

\begin{ex}
The family of distributions
$t(s):\varphi\in C^\infty_c(\mathbb{R}) \mapsto \int_{\mathbb{R}}e^{sx}\varphi(x)dx  $
defines a holomorphic germ of distributions at $s=0$ with real coefficients.
\end{ex}

\paragraph{Meromorphic germs of distributions.}

Once we have a proper definition for holomorphic families of distributions, we can give a very natural definition
of meromorphic families of distributions as follows~:
\begin{defi} [meromorphic family of distributions with linear poles]
For a complex domain $\Omega \subset \C ^p$, a meromorphic family of distributions on $\Omega $ is a holomorphic family $(t(s))_ {s\in\Omega \setminus\{s;L_1=\dots =L_k(s)=0\}}$ of distributions, where $L_1,\dots,L_k$ are linear functions on $\C ^p$, such that
\begin{equation}
t(s)=(L_1(s)\dots L_k(s))^{-1}h(s), \forall s\in \Omega\setminus\{s;L_1=\dots =L_k(s)=0\}
\end{equation}
where $(h(s))_ {s\in\Omega }\in \calD^\prime(X,\mathcal{O}(\Omega))$.
\end{defi}
Now we localize the above definition to germs
at $s_0\in \C^p$~:
\begin{defi}[meromorphic germs of distributions]
A \gmd at $s _0$ with linear poles is an equivalence class of meromorphic families of distributions on some neighborhood $\Omega$ of $s _0$ with linear poles under the equivalence relation: $t(s)_{s\in\Omega _1\setminus Y_1}\sim t^\prime (s)_{s\in\Omega _2\setminus Y_2}$, where $Y_i=\{s;L_1^i=\dots =L_{k_i}^i(s)=0\}$, $i=1,2$ with  linear functions $L_j^i$ $j=1,\cdots , k_i$, $i=1,2$, if there exist a complex domain $\Omega _3\subset \Omega _1\cap \Omega _2$, and linear functions $L_1^3, \cdots L^3_{k_3}$ such that $s _0\in \Omega_3$, $s_0\in Y_1\cap Y_2$, $Y_3=\{s;L_1^i=\dots =L_{k_i}^i(s)=0\}\subset Y_1, Y_3\subset Y_2$, and $t(s)_{s\in\Omega _3\setminus Y_3} =t^\prime (s)_{s\in\Omega _3\setminus Y_3}$. 
The set of meromorphic germs of distribution with real coefficients will be denoted by $\mathcal{D}^\prime(X,\calm_{s_0}(\C^p))$.
\end{defi}
It is simple to show that~:
\begin{prop} The set $\calD^\prime(X,\mathcal{O}_{s_0}(\C^p))$ is a vector subspace of $\mathcal{D}^\prime(X,\calm_{s_0}(\C^p))$.
\end{prop}

\subsection{Power expansions of holomorphic germs}

Let us state a convenient proposition about 
power series expansion of holomorphic families of distributions whose proof is given in the appendix. 

\begin{prop}\label{weakstrongholo}
Let $X$ be some smooth manifold, $\Omega\subset \mathbb{C}^p$ and
$(t(s))_{s\in\Omega}\in \mathcal{D}^\prime(X,\mathcal{O}(\Omega))$
be some holomorphic family of distributions.
Then near every $s_0\in\Omega$, $t_s$ admits a power series expansion
$$t(s)=\sum_{\alpha} \frac{(s-s_0)^\alpha}{\alpha!} t_\alpha$$
where $\alpha=(\alpha_1,\dots,\alpha_n)\in\N^n$ and $t_\alpha$ is a distribution in $\mathcal{D}^\prime(X)$
such that
for all test function $\varphi$,
$\sum_{\alpha} \frac{(s-s_0)^\alpha}{\alpha!} t_\alpha(\varphi)$
converges as power series near $s_0$.
\end{prop}
This classical result is just a multivariable version of~\cite[Theorem 1]{grothendieckholo}
which is stated for general locally convex spaces $E$, we include a proof in the appendix to make our text self--contained.

\subsection{From Green functions to the heat kernel.}
The fundamental tool we use to investigate the
singularities of Feynman amplitudes is the heat kernel. In this section, 
we recall its main properties and explain how one can express the
regularized Green functions and the Feynman amplitudes
in terms of the heat kernel.

\subsubsection{Heat kernels}
The complex powers
of $P=-\Delta_g+V$ are related to the heat kernel $e^{-tP}$ in the following way (see also~\cite[paragraph 1.12.14 p.~112]{Gilkey}):
\begin{prop}
Let $(M,g)$ be a smooth compact, connected Riemannian manifold without boundary
and let $P=-\Delta_g+V$, $V\in C_{\geqslant 0}^\infty(M)$ 
or $M=\R^d$ with constant metric $g$ and
$P=-\Delta_g+m^2, m\in \R_{\geqslant 0}$.
Set $\Pi$ to be the spectral projector
on $\ker(P)$, and $s\in\mathbb{C}$ with $Re(s)> 0$ then
\begin{equation}
 P^{-s}=\frac{1}{\Gamma(s)}\int_0^\infty \left(e^{-tP}-\Pi\right)t^s\frac{dt}{t}
\end{equation}
in the sense of bounded operator 
from $L^2(M,\mathbb{C})\mapsto L^2(M,\mathbb{C})$ 
where $\Gamma$ is the 
Euler Gamma function.
In the sense of Schwartz kernels~:
$$\boxed{\greenf  ^s(x,y)=\frac{1}{\Gamma(s)}\int_0^\infty \left(K_t(x,y)-\Pi(x,y)\right)t^s\frac{dt}{t} }$$
where $K_t\in C^\infty((0, \infty)\times M\times M)$ is the heat kernel.
\end{prop}

Note that when $0\notin \ker(P)$ and $M$ compact or when $M=\R^d$, then
we can set $\Pi=0$.

\begin{proof}
The proposition is clear when $M=\R^d$ hence we just discuss the compact case.
As a consequence of the compactness of $M$ and the fact that $P$ is an elliptic, positive, 
self-adjoint operator, $P$ has discrete spectrum denoted by $\sigma(P)$,
the eigenfunctions $(e_\lambda)_{\lambda\in \sigma(P)}$ of $P$ 
form an orthonormal basis
of $L^2(M,\mathbb{C})$, 
so for any $u\in L^2(M,\mathbb{C})$, 
$u=\sum \langle u,e_\lambda \rangle e_\lambda$. 
By definition,
$P^{-s}u=\sum_{\lambda\in\sigma(P),\lambda\neq 0}\lambda^{-s}\langle u,e_\lambda \rangle e_\lambda $
where the sum on the right hand side converges absolutely in $L^2(M,\mathbb{C})$.
And the spectral projector $\Pi$ on $\ker(P)$ is simply
$\Pi(u)=\sum_{\lambda=0} \langle u,e_\lambda \rangle e_\lambda $.

The heat operator $e^{-tP}$ is a strongly continuous semigroup acting
on $L^2(M,\mathbb{C})$. For every
$u\in L^2(M,\mathbb{C})$~:
$$(e^{-tP}-\Pi)u=\sum_{\lambda\in \sigma(P)\setminus 0} e^{-t\lambda}\langle u,e_\lambda \rangle e_\lambda$$
where the sum on the r.h.s converges in $L^2(M,\mathbb{C})$.

Therefore for $\lambda> 0$, by a change of variable in the $\Gamma$ function
$\lambda^{-s}= \frac{1}{\Gamma(s)}\int_0^\infty e^{-t\lambda}t^s\frac{dt}{t}$,
it follows that the identity
$P^{-s}=\frac{1}{\Gamma(s)}\int_0^\infty \left(e^{-tP}-\Pi\right)t^s\frac{dt}{t}$
holds true in operator sense where the integral on the r.h.s
converges
in operator norm. Hence the same identity should
hold true for the corresponding Schwartz kernels. 
\end{proof}

\subsection{Local asymptotic expansions of heat kernels }

We will use the following property of the heat kernel
asymptotics~\cite[Thm 2.30]{Berline-04} (see also~\cite[Thm 7.15]{roe1999})~:
\begin{thm}[ Minakshisundaram–Pleijel]\label{heatasympt}
Let $(M,g)$ be a compact Riemannian manifold without boundary, $\varepsilon$ is the \textbf{injectivity radius}  
of $M$ and $P=-\Delta_g+V$, $V\in C^\infty_{\geqslant 0}(M)$. 
For any cut--off
function $\psi:\mathbb{R}_+\mapsto [0,1]$ 
such that $\psi(s)=1$ if $s\leqslant \varepsilon^2/4$ 
and $\psi(s)=0$ if $s>4\varepsilon^2/9$.
Let $K_t(x,y)\in C^\infty((0,\infty)\times M\times M)$ 
denote the heat kernel, then there exist smooth
\textbf{real valued} functions $a_k(x,y)\in C^\infty (M\times M)$, 
$k=0,1, \cdots$, with $a_0(x,y)=1$, such that for 
all $(n,p)\in\mathbb{N}^2$, and differential 
operator $P(x,D_x)$ of degree $m$,
there exists a constant $C>0$ such that
for all $t\in (0,1]$~:
\begin{eqnarray}
\sup_{(x,y)\in M^2} \vert P(x,D_x)\partial^p_t\left( K_t(x,y)-\sum_{k=0}^n \psi(\mathbf{d}^2)\frac{e^{-\frac{\mathbf{d}^2(x,y)}{4t}}}{(4\pi t)^{\frac{d}{2}}} a_k(x,y) t^k \right)  \vert\leqslant Ct^{n-\frac{d}{2}-\frac{m}{2}-p}
\end{eqnarray}
where $\mathbf{d}(.,.):M\times M\mapsto \R_{\geqslant 0}$ is the Riemannian distance function.
\end{thm}

Note that
our statement differs from the
statement in~\cite{roe1999} from
the fact that we use a cut--off function $\psi$ since
outside some neighborhood of the diagonal $\Delta\subset M\times M$,
$K_t$ vanishes at infinite order in $t$ when
$t\rightarrow 0$ (see the proof of~\cite[Thm 7.15 p.~102]{roe1999}).
In case $M=\R^d$ with constant metric $g$ and $P=-\Delta_g+m^2, m\in \R_{\geqslant 0}$,
we have the well known exact formula~:
$$ K_t(x,y)=\frac{1}{(4\pi t)^{\frac{d}{2}}}\exp\left(-\frac{\vert x-y\vert^2_g}{4t}-tm^2\right)=\frac{\exp\left(-\frac{\vert x-y\vert^2_g}{4t}\right)}{(4\pi t)^{\frac{d}{2}}} \sum_{k=0}^\infty \frac{(-1)^kt^km^{2k}}{k!}  $$
which already appeared in definition \ref{d:complexflat}. 

%

%

\section{Reduction of regularized Feynman amplitudes}
Recall our aim was to prove
analytic continuation of the regularized
amplitude 
$$t_G=\prod_{e\in E(G)} \greenf^{s_e}(x_{i(e)},x_{j(e)}) $$
in
$\calD^\prime(M^{V(G)},\calm_{s_0}(\C^{E(G)}))$
where $s_0=(s_e=1)_{e\in E(G)}$ and
$\greenf^s$ is the Schwartz kernel of
the complex power $P^{-s}$.
The main goal of this section
is to prove a technical Theorem \ref{t:technicalreduction}
which allows us to reduce our main Theorem~\ref{mainthmintro}
to the proof of an analytic continuation Theorem 
for simpler analytic objects. These are 
some kind of Feynman amplitudes
introduced in definition~\ref{d:labelledFeynmanrules} 
corresponding to
\emph{labelled Feynman graphs} defined in definition~\ref{d:labelledgraphs} which are
graphs
whose edges are decorated by some integer. Intuitively, the amplitude of the labelled graph
is obtained from the regularized amplitude $t_G(s)$ where we replace the heat kernels appearing in the
formula for the Green function
$\greenf^s=\frac{1}{\Gamma(s)}\int_0^\infty t^{s-1}\left(e^{-tP}-\Pi\right)dt$,
by the heat kernel asymptotic expansion. The integers decorating the edges exactly
correspond
to the heat coefficients in the heat kernel asymptotic expansion.

\subsection{Holomorphicity of Green's function}

The next Lemma discusses analytical properties
of the regularized Green function
of the Schr\"odinger
operator $P$ which is elliptic 
since its leading part coincides with
the Laplace operator, it is therefore automatically self-adjoint by the symmetry assumption~\cite[p.~35]{taylor2013partial}.

\begin{lemm}\label{keylemma}
Let $(M,g)$ be a smooth compact, connected Riemannian manifold without boundary
and let $P=-\Delta_g+V$, $V\in C_{\geqslant 0}^\infty(M)$ 
or $M=\R^d$ with constant metric $g$ and
$P=-\Delta_g+m^2, m\in \R_{\geqslant 0}$. 
Denote by $K_t$ the corresponding heat kernel.
Then~:
\begin{enumerate}
\item For all $k\in\mathbb{N}$, if $Re(s)>\frac{d}{2}+k$, $\greenf ^{s}$ is a $C^k$ function on $M\times M$.
\item For all $k\in\mathbb{N}$, a compact subset $K\subset M\times M\setminus {\rm Diagonal} $, the kernel $\greenf ^{s}$ is holomorphic in $s$ and valued in $C^k(K)$.
\item If we write
\begin{eqnarray*}
\greenf ^{s}(x,y)=\greenf ^{s}_{\leq}(x,y)+ \greenf ^{s}_{\geq}(x,y)
\end{eqnarray*}
where
$$\greenf ^{s}_{\geq}(x,y)=\frac{1}{\Gamma(s)}\left(\int_0^{1} (K_t-\Pi)(x,y) t^{s-1}dt\right), \greenf ^{s}_{\leq}(x,y)=\frac{1}{\Gamma(s)}\left(\int_{1}^\infty (K_t-\Pi)(x,y) t^{s-1}dt\right)
$$
then $\greenf ^{s}_{\leq}(x,y)$ is holomorphic in $s$ and valued in $C^\infty (M\times M)$ which is denoted by\\
$\greenf^s_{\leq}\in C^\infty(M\times M, \mathcal{O}(\C^p))$.
\end{enumerate}
\end{lemm}
The proof of these classical properties, when $M$ is compact, is recalled in the appendix.
For $M=\R^d$, they follow from straightforward computations.

\subsection{Reduction to local charts and localization near deepest diagonal.}

The purpose of the next two Lemmas is to localize
the proof of our main Theorem about the
analytic continuation of the distribution
$t_G(s)$ to neighborhoods
of the deepest diagonals in $M^{V(G)}$.

\begin{lemm}\label{l:gluing}
Let $X$ be a manifold without boundary, $s_0\in\C^p$ then $t(s)\in \calD^\prime(X,\calm_{s_0}(\C^p))$ iff
for every $x\in N$, there exists a neighborhood $U_x$ of $x$ such that
$t(s)|_{U_x}
\in \mathcal{D}^\prime(U_x,\calm_{s_0}(\C^p))$.

If $X$ is non compact, we require 
that there are linear functions $(L_i)_{i=1}^k$ corresponding to
a fixed polar set $Y=\{s;L_1(s)=\dots=L_k(s)=0\}\subset \mathbb{C}^p, s_0\in Y$ such that
$t(s)|_{U_x}\in\mathcal{D}^\prime(U_x,\calm_{s_0}(\C^p))$ singular along the polar set $Y$.
\end{lemm}

\begin{proof}
The direct implication is straightforward.
Assume that for every $x\in X$, there exists a neighborhood $U_x$ of $x$ such that
$t(s) |_{U_x}\in \calD^\prime(U_x,\calm_{s_0}(\C^p))$.
Then by local compactness, there is a locally finite subcover $(U_i)_i$  of $N$ such that
$t(s)|_{U_i}\in \calD^\prime(U_i,\calm_{s_0}(\C^p))$. Let $(\chi_i)_i$ be a
partition of unity where each $\chi_i$ is supported in $U_i$.
Then for every test function $\varphi\in C^\infty_c(X)$,
$\langle t(s),\varphi\rangle
=\sum_i \langle t(s), \chi_i \varphi\rangle$ is a finite sum of
meromorphic germs with linear poles at $s_0$. 
In the non compact case, 
the polar set $Y$ is fixed.
Therefore the sum is a meromorphic germ with linear poles
at $s_0$. This is a finite sum, hence $s\mapsto \langle t(s),\varphi\rangle$ is
meromorphic with linear poles at $s_0$ with linear poles.
\end{proof}

The next Lemma is inspired by
the seminal work of Popineau and Stora~\cite{Popineau}
and
it states that
it is enough to work
our analytic continuation problem for the distributions
$t_G\in \mathcal{D}^\prime(M^{V(G)})$ near
the deepest diagonals~:
\begin{lemm}[Popineau--Stora Lemma]
\label{lem:diag} If for any graph $G$, and any $x\in M$, there is
some neighborhood $U_x$ of $x$, such that 
$t_G |_{U_x^{V(G)}}\in \calD^\prime(U_x^{V(G)},\calm_{s_0}(\C^p))$ where $s_0=(s_e=1)_{e\in E(G)}$, 
then $t_G\in \calD^\prime(M^{V(G)},\calm_{s_0}(\C^p)), s_0=(s_e=1)_{e\in E(G)}$ for every graph $G$.
\end{lemm}

\begin{proof}
We prove it by induction on
the number of vertices of the graph $G$. For $\vert V(G)\vert=2$, $t_G=\prod_{e=1}^E\greenf^{s_e}(x,y)$. For a point $(x,y) \in M^2=M^{V(G)}$, if $x\not =y$,
consider neighborhoods $V_x$ of $x$ in $M$ and $V_y$ of $y$ in $M$, such that $V_x\cap V_y=\emptyset $, then $\{V_x\times V_y\}_{(x,y)\in M^2, x\not =y} \bigcup \{U_x\times U_x\}_{x\in M}$ form an open cover of $M^2$, it has a
locally finite subcover $\{W_i\}_i$ with $W_i=V_{x_i}\times V_{y_i}$ or $U_{x_i}\times U_{x_i}$.
Let $\{\chi_i\}_i$ be a partition of unity supported on $\{W_i\}_i$, then $t_G=\sum_it_G\chi_i $,
where each $t_G\chi_i$ is holomorphic at $s_0$ if the support of $\chi_i$ does not intersect
the diagonal by Lemma \ref {keylemma} or has meromorphic continuation at $s_0$ by assumption.
Now the claim follows from Lemma~\ref{l:gluing} applied to the manifold $X=M^{V(G)}$.

Now $\vert V(G)\vert=n>2$ and assume the result holds for all graphs
whose number of vertices are strictly less than $n$.
Denote by $d_n=\{x_1=\dots=x_n\}\subset M^n$, the
deepest diagonal in the configuration space $M^n$.
For $(x_1,\dots,x_n)\in M^n\setminus d_n$, let $I=\{i\ |\ x_i=x_1\}$, and $I^c=\{1,\dots,n\}\setminus I$, then $I\not =\emptyset$, $I^c \not =\emptyset$, and for any $j\in I^c$ there are neighborhoods 
$U_{j}$ of $x_1$ and $V_{j}$ of $x_j$ such that $U_{j}\cap V_j=\emptyset$. 
Let
$\mathcal{V}=(\cap_{j\in I^c} U_j)^{|I|}\times \prod _{j\in I^c}V_j,$
then $\mathcal{V}\subset M^n\setminus d_n$,
and $x_i\neq x_j,\forall (i,j)\in I\times I^c$ for all $(x_1,\dots,x_n)\in \mathcal{V}$.
Then we partition the set of edges $E(G)$ of the graph $G$ in three parts~:
$E(G)=E_I\cup E_{I^c}\cup E_{II^c}$,
where $E_I$ (resp $E_{I^c}$) is the set of edges of $G$ whose
incident vertices are in $I$ (resp $I^c$) i.e.
every edge $e\in E_I$ (resp $e\in E_{I^c}$) is bounded by vertices in $I$ (resp $I^c$).
The remaining subset of edges is denoted by
$E_{II^c}$ and is made of all edges $e\in E(G)$ which are neither in $E_{I}$
nor in $E_{I^c}$. By that, we mean the edges of $E_{II^c}$ only connect
some vertex in $I$ with another vertex of $I^c$.
So we write
$(x_1,\dots,x_n)=(x_i,x_j)_{i\in I,j\in I^c}\in M^I\times M^{I^c}$. Similarly
the complex variables $(s_e)_{e\in E(G)}$ attached to the edges of
$G$ will be divided in three groups corresponding to the edges
$E_I$, $E_{I^c}$ and $E_{II^c}$ respectively.
Then we decompose the amplitude $t_G$ as a product of three factors~:
\begin{eqnarray*}
t_G\left((s_e)_{e\in E(G)};(x_i,x_j)_{i\in I,j\in I^c}\right)&=&t_{I}\left((s_e)_{e\in E_I};(x_i)_{i\in I}\right)t_{I^c}
\left(((s_e)_{e\in E_{I^c}};(x_j)_{j\in I^c}\right)\\
&\times &t_{II^c}\left((s_e)_{e\in E_{II^c}};(x_i,x_j)_{i\in I,j\in I^c}\right)
\end{eqnarray*}
where $t_I=\prod_{e\in E_I}\greenf^{s_e}, t_{I^c}=\prod_{e\in E_{I^c}}\greenf^{s_e}, t_{II^c}=\prod_{e\in E_{II^c}}\greenf^{s_e}$.

By the induction assumption, both $t_{I}$ and $t_{I^c}$ are distributions
in $\mathcal{D}^\prime(M^I,\calm_{s_{0I}}(\mathbb{C}^{E_I})),s_{0I}=(s_e=1)_{e\in E_I}$ 
and $\mathcal{D}^\prime(M^{I^c},\calm_{s_{0I^c}}(\C^{E_{I^c}})),s_{0I^c}=(s_e=1)_{e\in E_{I^c}}$ 
respectively.
Then by Lemma \ref{l:exteriorprod}, the exterior product
$t_{I}\left((s_e)_{e\in E_I};(x_i)_{i\in I}\right)t_{I^c}
\left(((s_e)_{e\in E_{I^c}};(x_j)_{j\in I^c}\right)$ of distributions
depending on different variables is an element
in $\mathcal{D}^\prime(M^n,\calm_{s_{0II^c}}(\C^{E_I\cup E_{I^c}})), s_{0II^c}=(s_e=1)_{e\in E_I\cup E_{I^c}}$.
Now the factor $t_{II^c}$ contains only product of propagators $\greenf ^{s_e}(x_{i},x_j)$ where $x_i\neq x_j$, so in
the open subset $\mathcal{V}\subset M^n$,
$\greenf ^{s_e}(x_{i},x_j)\in C^\infty(M\times M\setminus\text{Diagonal},\mathcal{O}(\mathbb{C}))$ 
in the variables $(x_i,x_j)$ by Lemma~\ref{keylemma}.
Thus on $\mathcal{V}$, $t_{II^c}\in C^\infty(\mathcal{V},\mathcal{O}(\C^{E_{II^c}}))$ which means 
$t_{II^c}$ is holomorphic in the parameter $(s_e)_{e\in E_{II^c}}\in \C^{E_{II^c}}$.
We conclude that near any element of $M^n$, there is some open neighborhood $U\subset M^n$ 
such that $t_G\in \calD^\prime(U,\calm_{s_0}(\C^{E(G)})), s_0=(s_e=1)_{e\in E(G)}$.
Then by Lemma \ref{l:gluing}, $t_G\in \calD^\prime(M^n,\calm_{s_0}(\C^{E(G)})), s_0=(s_e=1)_{e\in E(G)}$.
\end{proof}

\subsection{Reductions to integrals on cubes.}

In the representation of the Green function as integral of the heat kernel over
$(0,+\infty)$, we would like to get rid
of the low energy part which is 
$\greenf^s_\leqslant=\int_1^\infty dt (K_t-\Pi) t^{s-1}$ 
which is smooth and holomorphic in $s$ so it 
does not contribute to the singularities of $t_G$. We thus reduce
the study of $t_G$, to the study of some formula $P_G$ which contains only integrals
over cubes which are easier to handle and contain all the singularities of $t_G$.

\begin{defi} For a graph $G$, and $E\subset E(G)$, the subgraph induced by $E$ is the subgraph $H$ of $G$, such that $E(H)=E$, $V(H)=\{v\in V(G) \ |\ v \text{ is a  vertex incident to some }e\in E  \}$.
\end{defi}

\begin{prop}\label{p:Binreduction}
If for every graph $G$, the product
$$P_G(s)=\prod_{e\in E(G)} \frac{1}{\Gamma(s_e)}\left(\int_0^1  (K_{\ell_e}-\Pi)(x_{i(e)},x_{j(e)}) \ell_e^{s_e-1}d\ell_e\right)$$
extends to 
$\mathcal{D}^\prime(M^{V(G)},\calm_{s_0}(\C^{E(G)}))$ at $s_0=(s_e=1)_{e\in E(G)}$,
then $t_G(s)$ extends to $\mathcal{D}^\prime(M^{V(G)},\calm_{s_0}(\C^{E(G)}))$ at $s_0=(s_e=1)_{e\in E(G)}$.
\end{prop}

\begin{proof}
For all
$Re(s_e)>\frac{d}{2}$, since $\greenf^s,\greenf^s_{\geq},\greenf^s_{\leq}$
all belong to $C^0(M\times M)$ by Lemma \ref{keylemma}, the
following product makes perfect sense~:
\begin{eqnarray*}
t_G(s)=\prod_{e\in E(G)} \left(\greenf^{s_e}_{\leq} + \greenf^{s_e}_{\geq}\right)
=\prod_{e\in E(G)}\greenf^{s_e}_{\geq} + \prod_{e\in E(G)}\greenf^{s_e}_{\leq} +
\sum_{E_1\cup E_2=E(G)}\left( \prod_{e\in E_1}\greenf^{s_e}_{\geq}\right) \left( \prod_{e\in E_2}\greenf^{s_e}_{\leq}\right)
\end{eqnarray*}
where the sum runs over all partitions
$E(G)=E_1\cup E_2, E_1\neq \emptyset, E_2\neq \emptyset$.
Therefore~:
\begin{eqnarray*}
t_G(s)&=&  P_G(s) + \underbrace{\prod_{e\in E(G)}\greenf^{s_e}_{\leq}} +
\sum_{E_1\cup E_2=E(G)} P_{G(E_1)}(s) \underbrace{\left( \prod_{e\in E_2}\greenf^{s_e}_{\leq}\right)}\\
\end{eqnarray*}
where $G(E_1)$ is the induced subgraph of $G$ by the subset $E_1$.
The terms underbraced are
in
$C^\infty$ functions depending holomorphically on the parameters $s\in \C^{E(G)}$ near $(s_e=1)_{e\in E(G)}$ 
since each $\greenf^s_{\leq} \in C^\infty(M\times M,\mathcal{O}(\C))$.
By assumption,
for all induced subgraph $G(E_1)\subset G$,
$P_{G(E_1)}$ extends to 
$\mathcal{D}^\prime(M^{V(G(E_1))},\calm_{s_0}(\C^{E_1}))$, $s_0=(s_e=1)_{e\in E_1}$. 
Therefore by Lemma \ref{l:factorizationgmdtimesholo},
each product $P_{G(E_1)}(s) \left( \prod_{e\in E_2}\greenf^{s_e}_{\leq}\right)$ 
has analytic continuation in $\mathcal{D}^\prime(M^{V(G)},\calm_{s_0}(\C^{E(G)}))$, $s_0=(s_e=1)_{e\in E(G)}$.
\end{proof}

Therefore it is sufficient
to study~:
\begin{equation}\label{PGs}
P_G(s)=\prod_{e\in E(G)} \frac{1}{\Gamma(s_e)}\left(\int_0^1 (K_{\ell_e}-\Pi)(x_{i(e)},x_{j(e)}) \ell_e^{s_e-1}d\ell_e\right)
\end{equation}

\begin{lemm}\label{l:absolute convergence}
Let $G$ be a graph with $E$ edges
and
\begin{eqnarray}
P_G(s)=\prod_{e\in E(G)} \frac{1}{\Gamma(s_e)}\left(\int_{[0,1]^{E(G)}}\prod_{e\in E(G)} (K_{\ell_e}-\Pi)(x_{i(e)},x_{j(e)})
\ell_e^{s_e-1}d\ell_e\right).
\end{eqnarray}
If for all $e\in\{1,\dots,E\}$, $Re(s_e)>\frac{d}{2}$, then the integral
defining $P_G(s)$ converges absolutely in $[0,1]^E$ uniformly
in $(x_1,\dots,x_{|V(G)|})\in M^{V(G)}$.
\end{lemm}
\begin{proof}
First when $M=\R^d$ or if $M$ is compact and $0\notin \ker(P)$ then $\Pi=0$.
Otherwise, if $0\in \ker(P)$, then the Schwartz kernel of
$\Pi$ must a constant function (see appendix \ref{a:proofkeylemma}). 
Therefore, it is sufficient that $Re(s)>0$
so that $\int_0^1 \Pi(x,y) \ell^{s-1}d\ell $
is Riemann integrable.
Now by Theorem \ref{heatasympt}, there exists a constant
$C_{0}>0$ such that for $\ell \in (0,1]$ and for all $(x,y)\in M^2$~:
 $$
\vert\left(K_\ell(x,y)- \frac{1}{(4\pi \ell)^{\frac{d}{2}}} \psi(\mathbf{d}^2(x,y))e^{-\frac{\mathbf{d}^2(x,y)}{4\ell}}\sum_{0\leqslant k\leqslant\frac{d}{2}+1} a_k(x,y)\ell^k\right)\vert \leqslant  C_{0}.$$
So by the triangular inequality and by positivity of the heat kernel,
we have the bound
$$0\leqslant K_\ell(x,y) \leqslant \frac{1}{(4\pi \ell)^{\frac{d}{2}}} \psi(\mathbf{d}^2(x,y))e^{-\frac{\mathbf{d}^2(x,y)}{4\ell}}\sum_{0\leqslant k\leqslant\frac{d}{2}+1} \vert a_k(x,y) \vert\ell^k + C_{0}.$$
From which we can bound the
integral~:
\begin{eqnarray*}
\int_0^1  K_\ell(x,y) \vert\ell^{s-1}\vert d\ell
&\leqslant &
\int_0^1  \frac{1}{(4\pi)^{\frac{d}{2}}} \psi(\mathbf{d}^2(x,y))e^{-\frac{\mathbf{d}^2(x,y)}{4\ell}}\sum_{0\leqslant k\leqslant\frac{d}{2}+1} \vert a_k(x,y)\ell^{k+s-\frac{d}{2}-1}\vert +C_{0}  \vert\ell^{s-1}\vert d\ell\\
&\leqslant & \int_0^1  \frac{1}{(4\pi)^{\frac{d}{2}}} \sum_{0\leqslant k\leqslant\frac{d}{2}+1} \vert a_k(x,y)\ell^{k+s-\frac{d}{2}-1}\vert+C_{0} \vert\ell^{s-1}\vert d\ell
\end{eqnarray*}
since $\psi e^{-\frac{\mathbf{d}^2}{4\ell}}\leqslant 1$ and the right hand side is absolutely integrable when $Re(s)>\frac{d}{2}$. 
Therefore
$$\left(\int_{[0,1]^{E(G)}}\prod_{e\in E(G)} (K_{\ell_e}-\Pi)(x_{i(e)},x_{j(e)})
\ell_e^{s_e-1}d\ell_e\right)=\prod_{e\in E(G)}\int_0^1(K_{\ell_e}-\Pi)(x_{i(e)},x_{j(e)})
\ell_e^{s_e-1}d\ell_e$$ is a product of convergent Riemann integrals, the above integral inversions make sense by Fubini
which yields the claim of our Lemma.
\end{proof}

Now we set~:
\begin{equation}
I_G(s)=\prod_{e\in E(G)} \frac{1}{\Gamma(s_e)}\int_0^1 K_{\ell_e}(x_{i(e)},x_{j(e)})\ell^{s_e-1}d\ell_e
\end{equation}
which is well--defined as soon as
$Re(s_e)>\frac{d}{2}, \forall e\in E(G)$ by the above arguments.
Then
\begin{equation}
P_G(s)=\sum_{E\subset E(G)} I_{G(E)}(s) \prod_{e\in E(G)\setminus E} \frac{\Pi(x_{i(e)},x_{j(e)})}{s_e}
\end{equation}
where $G(E)$ is the induced subgraph by the subset of
edges $E\subset E(G)$.
By the fact that
$\int_0^1 \Pi\ell^{s-1}d\ell=\frac{\Pi}{s}\in C^\infty(M\times M,\mathcal{O}_1(\C))$, 
which is holomorphic near $s=1$, the products of spectral projectors 
do not contribute to the poles.
So we can further reduce our study to the analytic continuation of $I_G(s)$.

\subsection{Distributional order}
In this step, we introduce a further reduction
by 
replacing each $K_\ell$
in the integral formula of $I_G(s)$
by its heat asymptotic expansion and try to control the remainders.
\begin{eqnarray*}
\frac{1}{\Gamma(s)}\int_0^1 K_{\ell}\ell^{s-1}d\ell
&=& \frac{1}{\Gamma(s)}\int_0^1 \frac{ e^{-\frac{\mathbf{d}^2(x,y)}{4\ell}}}{(4\pi\ell)^{\frac{d}{2}}} \left(\sum_{k=0}^{p}
a_{k}(x,y)\psi(\mathbf{d}^2(x,y))\ell^{k}  \right) + h_{p}(\ell,x,y) \ell^{s-1}d\ell
\end{eqnarray*}
where $h_{p}(\ell,x,y)$ is the remainder in the
heat asymptotics
which satisfies
the estimate
$\Vert  h_p \Vert_m \leqslant C \ell^{p-\frac{d}{2}-\frac{m}{2}}$
by Theorem \ref{heatasympt} and $\psi$ is the cut-off function from Theorem
\ref{heatasympt}.

We first introduce some refinement of Feynman graphs
to keep track of the information
on the heat coefficients 
for every edge.
So these are basically Feynman graphs
whose edges are decorated
by integers which correspond
to heat coefficients.
\begin{defi}[labelled graph]\label{d:labelledgraphs}
For a set $S$, an $S$-labelled
graph is a pair $(G,\overrightarrow{k})$ where 
$\overrightarrow{k}$ is a map $E(G) \to S$. 
If $S$ is $\N$, 
we call it 
shortly \textbf{labelled graph}, and for $e\in E(G)$, 
we use the short notation $k_e$ to denote the 
element $\overrightarrow{k}(e)\in \N$. If $S=\mathbb{R}_{>0}$, then
the map $E(G)\mapsto \mathbb{R}_{>0}$, called the length function, is denoted by $\ell$ and we call such 
pair $(G,\ell)$ a \textbf{metric graph}. If $\ell$ is injective, then $(G,\ell)$ is called \textbf{strict metric graph}.
\end{defi}

We next define Feynman amplitudes attached to labelled graphs.
\begin{defi}\label{d:labelledFeynmanrules}
For every labelled graph $(G,\overrightarrow{k})$,
we define the corresponding amplitude $I_{G,\overrightarrow{k}}(s)$
as follows~:
\begin{equation}
\label{defIGks}
I_{G,\overrightarrow{k}}(s)=\prod_{e\in E(G)} \frac{1}{\Gamma(s_e)} \int_{[0,1]^E} \prod_{e\in E(G)} \left(\frac{e^{-\frac{\mathbf{d}^2}{4\ell_e}}}{(4\pi)^{\frac{d}{2}}} a_{k_e}
\psi(\mathbf{d}^2)\right)(x_{i(e)},x_{j(e)})\ell_e^{k_e-\frac{d}{2}+s_e-1}d\ell_e
\end{equation}
which is well--defined and holomorphic in $(s_e)_{e\in E(G)}\in \C^{E(G)}$ on the domain $s_e>\frac{d}{2}, e\in E(G)$
by exactly the same proof as in Lemma \ref{l:absolute convergence}. 
\end{defi}
\begin{prop}\label{keypropmeroIgks}
If for every graph $G$, there is $m\in \mathbb{N}$ depending on $G$, such that for all labels 
$\overrightarrow{k}\in \N^{E(G)}$, for all $x\in M$, there is an open neighborhood $U_x\subset M$ of $x$
such that
$I_{G,\overrightarrow{k}}(s)$ has analytic continuation
in $\mathcal{D}^{\prime,m}(U_x^{V(G)},\calm_{s_0}(\C^{E(G)}))$, $s_0=(s_e=1)_{e\in E(G)}$,
then for all $G$,
$t_G(s)$ extends in
$\mathcal{D}^\prime(M^{V(G)},\calm_{s_0}(\C^{E(G)}))$, $s_0=(s_e=1)_{e\in E(G)}$.
\end{prop}
\begin{proof}
Let $n=|V(G)|$ and $U=U_x$.
By Lemma \ref {lem:diag} which allows us to localize 
our analytic continuation proof near the deepest diagonal of $M^{n}$, 
we only need
to prove that $t_G(s)$ extends as a \gmd at $(s_e=1)_{e\in E(G)}$  on $U^n$.
For a test function $\varphi$,
\begin{eqnarray*}
&&\langle I_G(s), \varphi \rangle \\
&=& \int _{U^n}\big(\prod_{e\in E(G)} \frac{1}{\Gamma(s_e)}\int_0^1 K_{\ell_e}\ell_e^{s-1} d\ell_e \big)\varphi dv(x_1)\dots dv(x_n) \\
\overline{\overline{}}&=& \int _{U^n}\prod_{e\in E(G)}\left(\sum_{k_e=0}^{p}\frac{1}{\Gamma(s_e)}\int_0^1 \frac{ e^{-\frac{\mathbf{d}^2}{4\ell_e}}}{(4\pi)^{\frac{d}{2}}}
a_{k_e}\psi(\mathbf{d}^2) \ell^{k_e+s_e-\frac{d}{2}-1}d\ell_e  +   \frac{1}{\Gamma(s_e)}\int_0^1 h_{p}\ell_e^{s_e-1}d\ell_e \right) \varphi dv(x_1)\dots dv(x_n)\\
&=&\int _{U^n} \big(\sum_{E_1\cup E_2=E(G)} \prod_{e\in E_1}\left(\sum_{k_e=0}^{p}\frac{1}{\Gamma(s_e)}\int_0^1 \frac{ e^{-\frac{\mathbf{d}^2}{4\ell_e}}}{(4\pi)^{\frac{d}{2}}}
a_{k_e}\psi(\mathbf{d}^2) \ell^{k_e+s_e-\frac{d}{2}-1}d\ell_e \right)  \prod_{e\in E_2} \left( \frac{1}{\Gamma(s_e)}\int_0^1 h_{p}\ell_e^{s_e-1}d\ell_e \right)\big)\\
&\times & \varphi dv(x_1)\dots dv(x_n)\\
\end{eqnarray*}
where the sum runs over partitions $E_1\cup E_2=E(G)$.
Therefore we obtain
\begin{eqnarray*}
\langle I_G(s), \varphi \rangle=\int _{U^n} \sum_{E_1\cup E_2=E(G)} \sum_{\overrightarrow{k}\in \{0,\dots,p\}^{E_1}}  I_{G(E_1),\overrightarrow{k}}(s)
\underbrace{\prod_{e\in E_2} \left( \frac{1}{\Gamma(s_e)}\int_0^1 h_{p}(\ell_e,x_{i(e)},x_{j(e)})\ell_e^{s_e-1}d\ell_e \right)}\varphi dv(x_1)\dots dv(x_n)\\
\end{eqnarray*}
where the summation is over all $\overrightarrow{k}\in \{0, 1, \cdots, p\}^{E(G)}$.
By Theorem \ref {heatasympt}, we have the two estimates
$$|D_x^m h_{p}(\ell_e,x,y)\ell_e^{s_e-1}|\le C \ell ^{p-\frac d2-\frac m2+s_e-1},
$$
and
$$|D_x^mh_{p}(\ell_e,x,y)\ell_e^{s_e-1}\log \ell _e|\le C \ell ^{p-\frac d2-\frac m2+s_e-1+\varepsilon},
$$
for some $\varepsilon >0$. So when $p>\frac{d+m}{2}-1$, for every $e\in E_2$, 
there is a small neighborhood
of $s_e=1$ such that the integral $\int_0^1 h_{p}(\ell_e,x_{i(e)},x_{j(e)})\ell_e^{s_e-1}d\ell_e$ is absolutely convergent
and depends holomorphically on $s_e$. Hence
the term underbraced belongs to 
$C^m(U^n,\mathcal{O}_{s_0}(\C^{E_2})), s_0=(s_e=1)_{e\in E_2}$, 
where $G(E_2)$ is the graph induced by $E_2$.
Now we conclude the proof by noticing that the 
product of $I_{G(E_1),\overrightarrow{k}}(s)\in \calD^{\prime,m}(U^n, \calm_{s_0}(\C^{E_1})), s_0=(s_e=1)_{e\in E_1}$ 
with some element in $C^m(U^n,\mathcal{O}_{s_0}(\C^{E_2})), s_0=(s_e=1)_{e\in E_2}$ 
yields an element of $ \calD^{\prime,m}(U^n, \calm_{s_0}(\C^{E_1\cup E_2})), s_0=(s_e=1)_{e\in E_1\cup E_2}$ 
by Lemma~\ref{l:factorizationgmdtimesholo} proved in the appendix.
\end{proof}

The next Theorem is the main result from the present section and
summarizes all reduction steps performed above~:
\\
\\
\fbox{
\begin{minipage}{0.94\textwidth}
\begin{thm}[Reduction Theorem]\label{t:technicalreduction}
Assume that for every graph $G$, there is an integer $m(G)$, such that for any $x\in M$, there is a chart $U_x$ of $M$ around $x$
such that for all $\overrightarrow{k}$,
$I_{G,\overrightarrow{k}}(s)|_{U_x^n}$ admits an analytic continuation in 
$\mathcal{D}^{\prime,m(G)}\left(U_x^n,\calm_{s_0}(\C^{E(G)}) \right)$, $s_0=(s_e=1)_{e\in E(G)}$. 

Then for a given graph $G$, let $m=\sup_{G^\prime\subset G} m(G^\prime) $, for any $p>\frac{d+m}{2}-1$,
we have a decomposition
\begin{eqnarray}
t_G(s)=\sum_{G^\prime\subset G}\mathbf{m}_{G^\prime}(s)h_{G\setminus G^\prime}(s) 
\end{eqnarray}
where the sum runs over induced subgraphs $G^\prime$ of $G$, 
$\mathbf{m}_{G^\prime}(s)=\sum_{\overrightarrow{k}\in \{0,\dots,p\}^{E(G^\prime)}}
 I_{G^\prime,\overrightarrow{k}}(s)\in \mathcal{D}^{\prime,m}(M^{V(G^\prime)},\calm_{s_0}(\C^{E(G^\prime)}) )$
and 
$h_{G\setminus G^\prime}(s) \in C^m(M^{V(G\setminus G^\prime)},\mathcal{O}_{s_0}(\C^{E(G)\setminus E(G)^\prime}))$. 
In particular, $t_G(s)$ extends
as an element in $\calD^\prime(M^{V(G)},\calm_{s_0}(\C^{E(G)})) $, $s_0=(s_e=1)_{e\in E(G)}$.
\end{thm}
\end{minipage}
}
\\
\\
So the above Theorem allows to reduce the proof
of Theorem \ref{mainthmintro} to the analytic continuation 
of the
simpler objects $I_{G,\overrightarrow{k}}(s)$ if we can control the distributional
order of $I_{G,\overrightarrow{k}}(s)$ independently 
of the label $\overrightarrow{k}\in \N^{E(G)}$.

\begin{proof}
We
use the following
decomposition formula which summarizes the above three reduction steps, namely 
the reduction on cubes, the elimination of the spectral projector and
the extraction of labelled graphs~:
\begin{eqnarray*}
t_G(s)|_{U_x^n}&=&\sum_{E_1\cup E_2\cup E_3\cup E_4=E(G)}\left(\sum_{\overrightarrow{k}\in \{0,\dots,p\}^{E_1}} \underbrace{I_{G(E_1),\overrightarrow{k}}(s)}\right)\prod_{e\in E_2} \left( \frac{1}{\Gamma(s_e)}\int_0^1 h_{p}(\ell_e,x,y)\ell_e^{s_e-1}d\ell_e \right)\\
&\times&\prod_{e\in E_3} \frac{\Pi(x_{i(e)},y_{j(e)})}{s_e} \prod_{e\in E_4} \greenf^{s_e}_{\leq}
\end{eqnarray*}
where the sum runs over partitions $E_1\cup E_2\cup E_3\cup E_4=E(G)$. Then
consider $m=\sup_{E_1\subset E(G)} m(G(E_1))$ which is 
the supremum of distributional orders $m(G(E_1))$ for $E_1\subset E(G)$,
$m$ is finite by assumption and bounds the distributional order of
all the terms $I_{G(E_1),\overrightarrow{k}}(s)$ underbraced. Moreover,
we saw in the proof of Proposition \ref{keypropmeroIgks} 
that if we choose
$p>\frac{d+m}{2}-1$,
then the product
$$\prod_{e\in E_2} \left( \frac{1}{\Gamma(s_e)}\int_0^1 h_{p}(\ell_e,x,y)\ell_e^{s_e-1}d\ell_e \right)
\prod_{e\in E_3} \frac{\Pi(x_{i(e)},y_{j(e)})}{s_e} \prod_{e\in E_4} \greenf^{s_e}_{\leq}$$ 
belongs to $C^m(M^{V(G(E_2\cup E_3\cup E_4))},\mathcal{O}_{s_0}(\C^{E_2\cup E_3\cup E_4}))$, $s_0=(s_e=1)_{e\in E_2\cup E_3\cup E_4}$.
Therefore the whole product $t_G(s)\in \mathcal{D}^{\prime}(M^{V(G)},\calm_{s_0}(\C^{E(G)}))$, $(s_e=1)_{e\in E(G)}$.
The above complicated formula
can be written
very concisely as
\begin{eqnarray*}
t_G(s)=\sum_{G^\prime\subset G}\mathbf{m}_{G^\prime}(s)h_{G\setminus G^\prime}(s) 
\end{eqnarray*}
where the sum runs over induced subgraphs $G^\prime$ of $G$, $\mathbf{m}_{G^\prime}(s)=\sum_{\overrightarrow{k}\in \{0,\dots,p\}^{E(G^\prime)}} I_{G^\prime,\overrightarrow{k}}(s) $ 
and $h_{G\setminus G^\prime}=\prod_{e\in E_2} \left( \frac{1}{\Gamma(s_e)}\int_0^1 h_{p}(\ell_e,x,y)\ell_e^{s_e-1}d\ell_e \right)
\prod_{e\in E_3} \frac{\Pi(x_{i(e)},x_{j(e)})}{s_e} \prod_{e\in E_4} \greenf^{s_e}_{\leq}$
where $E_2\cup E_3\cup E_4$ forms a partition of $E(G)\setminus E(G^\prime)$.
\end{proof}

\section{Desingularization of parameter space.}

Now that we reduced
the proof of Theorem \ref{mainthmintro} to the proof of Theorem
\ref{t:technicalreduction}, we start by studying  in local coordinates the
amplitudes $I_{G,\overrightarrow{k}}(s)\in \mathcal{D}^\prime(M^{V(G)})$ corresponding to labelled
graphs $(G,\overrightarrow{k})$.

\paragraph{Fixing charts.}

For any $x\in M$, take a coordinate chart $U$ around $x$ such that $U\cong \R ^d$ and  
$\bar U \subset M$ is compact and 
$\mathbf{d}(y_1, y_2)<\frac {\varepsilon}2 $ for any $y_1, y_2 \in U$. Since the volume form $dv$ on a Riemannian manifold
reads $\vert\det(g) \vert d^dx$ in local coordinate chart, we may absorb
the smooth function $\vert\det(g) \vert$ in the test function $\varphi$ and forget
about the determinant of the metric in the local coordinate system.
We number the vertices of $G$ by $\{1,\dots,n\}$,
and the edges by integers $\{1,\dots,E\}$. For every edge
$e$, we associate vertices $i(e),j(e)$ which are incident to $e$.
Then for a test function $\varphi$ with $supp(\varphi )\subset U^n$,
\begin{eqnarray*}
&&\langle I_{G,\overrightarrow{k}}(s),\varphi\rangle\\
&=&\frac{1}{(4\pi)^{\frac{dE}{2}}}\prod_{e=1}^E\frac{1}{\Gamma(s_e)}\int_{[0,1]^E}d\ell_1\dots d\ell_E\left(\int_{U^n} \prod_{e=1}^E  \exp\left(-\frac{\mathbf{d}^2}{4\ell_e}\right)\ell_e^{s_e+k_e-\frac{d}{2}-1}a_{k_e}\psi(\mathbf{d}^2)\tilde \varphi d^dx_1\dots d^dx_n\right)
\end{eqnarray*}
where $\tilde \varphi=\vert\det(g) \vert \varphi$. 
This formula is well--defined when
$Re(s_e)>\frac{d}{2}, \forall e\in \{1,\dots,E\}$
since the integration on $[0,1]^E$ is absolutely convergent, the integral
on $U^n$ converges absolutely by compactness of the support of $\varphi$, hence
we can integrate in order by 
Fubini Theorem. Furthermore, arguing as in the proof of Proposition \ref{keypropmeroIgks} 
show that $\langle I_{G,\overrightarrow{k}}(s),\varphi\rangle$ is holomorphic in $s\in \mathbb{C}^E$
when $Re(s_e)>\frac{d}{2}$.

By our choice of $U$, $\mathbf{d}^2$ is smooth on $U\times U$, it is enough to prove that~:
$$
\frac{1}{(4\pi)^{\frac{dE}{2}}}\int_{[0,1]^E}d\ell_1\dots d\ell_E\left(\int_{\R ^{dn}} \prod_{e=1}^E  \exp\left(-\frac{\mathbf{d}^2}{4\ell_e}\right)\ell_e^{s_e+k_e-\frac{d}{2}-1}a_{k_e}\psi (\mathbf{d}^2)\tilde \varphi d^dx_1\dots d^dx_n\right)
$$
extends to a meromorphic germ of distribution at $(s_e=1)$. Note that this argument 
also applies to the case where $M=\R^d$ with constant metric $g$ 
and $P=-\Delta_g+m^2, m\in \R_{\geqslant 0} $.

\subsection{Smoothness problems and the need to resolve singularities.}

Assume we work on flat space $\R^d$, then to study the analytic continuation
of $I_{G,\overrightarrow{k}}$, we need to study 
integrals of the form~:
$$ \int_{[0,1]^E}d\ell_1\dots d\ell_E\left(\int_{\R ^{dn}} \prod_{e=1}^E  \exp\left(-\frac{\vert x_{i(e)}-x_{j(e)} \vert^2}{4\ell_e}\right)\ell_e^{s_e+k_e-\frac{d}{2}-1}a_{k_e}\psi (\mathbf{d}^2)\tilde \varphi d^dx_1\dots d^dx_n\right) .$$
The analytic continuation would come from integration by parts on the
cube $[0,1]^E$
w.r.t. the variables
$(\ell_1,\dots,\ell_E)$. However, 
we see immediately that
$e^{-\vert x-y\vert^2/4\ell}$ is not
a smooth function of $(\ell,x,y)\in [0,1]\times \mathbb{R}^d\times \mathbb{R}^d$.
The problem occurs at the set
$X=\{\ell=0,x-y=0\}\subset [0,1]\times \mathbb{R}^d\times \mathbb{R}^d$.
A solution in global analysis is to 
consider the following smooth map~:
$$\pi: (t,x,h)\in [0,1]\times (\mathbb{R}^d)^2\mapsto (t,x,x+\sqrt{t}h)\in [0,1]\times (\mathbb{R}^d)^2 .$$
Note that after pull--back by $\pi$, 
we find that $e^{-\vert x-y\vert^2/4\ell}\circ \pi(x,t,h)= e^{-\vert h\vert^2/4} $
which is now a smooth function near
the preimage $\pi^{-1}\left(X \right)=\{t=0\}$
in $[0,1]\times (\mathbb{R}^d)^2$. We say that we resolved the singularities 
of $e^{-\vert x-y\vert^2/4\ell}$.
For a discussion of why one needs to use blow--ups
to study heat kernels and applications to index theory, 
the reader is referred to~\cite[p.~253]{melroseaps}.
Similarly, the 
product of exponentials 
$  \prod_{e=1}^E  \exp\left(-\frac{\vert x_{i(e)}-x_{j(e)} \vert^2}{4\ell_e}\right) $
appearing in Feynman amplitudes
is also not smooth on the whole domain
of integration $(\ell_1,\dots,\ell_E,x_1,\dots,x_n) \in [0,1]^E\times \mathbb{R}^{dn}$
and integration by parts cannot be done.
It follows that we must resolve the products
$  \prod_{e=1}^E  \exp\left(-\frac{\vert x_{i(e)}-x_{j(e)} \vert^2}{4\ell_e}\right) $
to make them smooth which is discussed in paragraph~\ref{ss:blowup}.
Such resolution of singularities 
were studied by Speer on flat space building on the work of Hepp. Also,
when $(M,g)$ is an analytic Riemannian manifold or when $M=\R^d$ 
with constant Euclidean metric,
one can use Hironaka's resolution of singularities
as in~\cite{Atiyah} or Bernstein--Sato polynomials
to regularize Feynman amplitudes~\cite{dang2015complex,Herscovich}.
However,
on a Riemannian manifold $(M,g)$, 
if for all $m\in M$,
there is an open subset $U$ containing $m$ 
and a local coordinate system $(x_1,\dots,x_d):U\subset M\mapsto \R^d$ 
such that for every $(m_1,m_2)\in U^2$,
$\textbf{d}( m_1,m_2 )=\sqrt{\sum_{i=1}^d (x_i(m_1)-x_i(m_2))^2 }$ 
\textbf{then $(M,g)$ is flat}. Otherwise for \textbf{generic} Riemannian
manifolds $(M,g)$, it is not possible to find good coordinates
to make the distance function locally quadratic because of curvature.
This makes our resolution of singularities
more difficult to handle than the one appearing in the work of Speer and the
fact that we work in the $C^\infty$ case and not in the analytic or algebraic category 
prevents us from using 
directly Hironaka's resolution of singularities 
or Bernstein--Sato polynomials.
Following the tradition in QFT~\cite{RivasseauGurau,Rivasseau}, 
our strategy is essentially combinatorial and our blow-ups are 
encoded by spanning trees of Feynman graphs whose definition
is recalled in the next paragraph.

\subsection{Spanning trees of metric graphs.}

\delete {We consider the enlarged
configuration space
$(x_1,\dots,x_n,\lambda_1,\dots,\lambda_E)\in \zb {(\R ^d) ^n}\times [0,1]^E$.}
 Let us first collect some definitions and classical results on graphs which are close to~\cite[paragraph 2.1]{krajewski2010topological}.
Recall that for all graph we consider in the present paper, since we assume the graph has no self-loop, every edge $e$
is adjacent to two different vertices.
\begin{defi} For a graph $G$,
\begin {itemize}
\item a path from vertex $u$ to $w$ in a graph $G$ is a sequence $(u=v_1, e_1, v_2, \cdots, v_{n}, e_n, v_{n+1}=w)$, where $v_i \in V(G)$, $e_j\in E(G)$ such that the vertices bounding $e_i$ are $v_{i}$ and $v_{i+1}$, $u$ is the initial vertex and $w$ is the terminal vertex, $n$ is called the length. A path is simple if all the edges are distinct. If $u=w$, it is called a 
\textbf{cycle}.
\item
The set of subgraphs is ordered as follows, we say $G_1\subset G_2$ for
two subgraphs $G_1,G_2$ of $G$ if $E(G_1)\subset E(G_2)$.
A forest $T$ is a graph
without any simple cycle and a tree is a connected forest.
\item A {\bf spanning forest} $T$ of a graph $G$ is a subgraph of $G$ which is a forest and is maximal for the inclusion relation among subgraphs which are forests. If $T$ is a tree, it is called \textbf{spanning tree}.
For any graph $G$,  we will often use
the following equivalent characterizations of spanning forests
which is a classical result in graph 
theory~\cite[p.~40-41]{lefschetzapplications}.
A subgraph $T\subset G$ is a spanning forest if and only if $T$ is a forest whose complement contains $b_1(G)$ edges~:
\begin{equation}\label{e:spanning}
b_1(G)=\vert E(G) \vert-\vert E(T)\vert 
\end{equation}
where $b_1(G)$ is the first Betti number of $G$.
\item To a metric graph $(G,\ell)$, a metric filtration of $G$ is an increasing family of subgraphs
$G_1\subset G_2 \subset \dots \subset G_E$ where $G_i$ is \textbf{induced by the shortest $i$ edges} where $E=\vert E(G)\vert$ is the number
of edges in $G$. For a strict metric graph, the metric filtration is unique.
\item For every forest $T\subset G$ and every subgraph $G_i$ of $G$, we define
$T|_{G_i}$ as the subgraph of $G_i$ \textbf{induced by} the subset of edges $E(T)\cap E(G_i)\subset E(G_i)$.
We will call $T|_{G_i}$ \textbf{the trace} of $T$ in $G_i$.
\item If $T$ is a subgraph of $G$ induced by $E(T)$, and $e\in E(G)\setminus E(T)$
then we define $T\cup e$ as the subgraph of $G$ induced by $E(T)\cup e$. For every edge $e\in E(G)$, the graph $G\setminus e$ is the subgraph induced
by $E(G)\setminus e$.
\end{itemize}

\end{defi}

\begin{defi}
For any permutation $\sigma\in S_E$ of $\{1,\dots,E\}$, the
simplex $\mathbf{\Delta}_\sigma$~:
\begin{equation}
 \Delta_\sigma=\{\ell_{\sigma(1)}< \dots < \ell_{\sigma(E)} \}.
\end{equation}
is called a \textbf{sector} of $[0,1]^E$.
\end{defi}
Before we proceed, let us remark that
for a graph $G$, an element $\ell\in [0,1]^{E(G)}$, which is a map $\ell: E(G)\to [0,1]$, naturally defines a metric graph $(G, \ell)$.
To a strict metric graph $G$, the metric induces a natural strict ordering of edges
by the length which defines an element $\ell\in [0,1]^{E(G)}$ in a unique sector.
The next Theorem, due to Kruskal, aims to show how from
a strict connected metric graph $(G,\ell)$, one can
produce some algorithm which
extracts a unique spanning tree $T$ in $G$.
\delete {For each sector $\Delta_\sigma$, it assigns to a graph $G$ a tree
$T_\sigma(G)$ which is spanning in $G$.
}
\begin{thm}[Kruskal]\label{t:inequalitiestree}
For a connected strict metric graph $(G,\ell)$,
let $G_1 \subset \dots \subset G_E=G$ be the unique metric filtration of $G$. 
We denote by $N_i$ the first Betti number $b_1(G_i)$ of $G_i$.
Then there exists a unique spanning tree $T$ of $G$ such that for all $i\in\{1,\dots,E\}$,
its trace $T|_{G_i}$
is a spanning forest of $G_i$.
\end{thm}

\begin{proof} Let $\ell: E(G)\to (0, \infty)$ be the length function.
We shall assume that the edges $E(G)$ are numbered as
$\{e_1,\dots,e_E\}$ in
such a way that $i<j\implies \ell ({e_i})< \ell ({e_j})$.
We construct the tree by the Kruskal algorithm~\cite{kruskal1956shortest} as described in~\cite[p.~107]{RivasseauGurau}.
Notice that the requirement that $T$ is a tree implies that $T$ together with all traces $T|_{G_i}$ 
contain no simple cycles. So $T|_{G_i}$ is a forest and therefore its complement in $G_i$
contains at least $N_i$ edges of $G_i$.
Also notice that for any graph $G$, for every $e\in E(G)$, we have the inequality
$0\leqslant b_1\left(G\right)-b_1(G\setminus e)\leqslant 1$. This implies that the sequence $N_1, N_2, \cdots$ is increasing.

Now we can construct the desired spanning tree $T$: start from $G_1$ which has only one edge $\{e_1\}$ hence contains no simple cycle, $N_1=0$.
Let us denote by $i_1$ the
first integer such that $b_1(G_{i_1})=1$, similarly define  $\{i_2,\dots,i_{N_E}\}\subset \{1,\dots,E\}$ such that $b_1(G_{i_2})=2,\dots,b_1(G_{i_{N_E}})=N_E=b_1(G)$ and every $i_j$
is the smallest integer so that $b_1(G_{i_j})=j$ for any $j=1, \cdots , N_E$.
Set $i_0=1$,
then we have an increasing family of subgraphs $G_{i_0}\subset G_{i_1}\subset \dots \subset G_{i_{N_E}}\subset G$.
\delete {where $G_{i_{j}}\setminus G_{i_{j-1}}$
contains exactly one loop \zb {This is not true, see theta graph} for every $j=1,\dots,N_E$.
For}
Let $G_{i}= G_{i-1}\cup e_{i}$ and we set
$T=G\setminus \cup_{j=1}^{N_E} e_{i_j}.$
We prove that the subgraph $T$ constructed
above has the property that its trace $T|_{G_i}$ to every subgraph $G_i$
is a spanning
forest in $G_i$ by induction for $j=1,\dots,E$.
First, we initialize the induction for $j=1$, $G_1$ contains just one edge hence $T|_{G_1}=G_1$ is
a spanning tree in $G_1$. Assume that $T|_{G_k}$ is a spanning forest in $G_k$, then
there are two cases:

Case 1~: $b_1(G_k)=b_1(G_{k+1})$ i.e. $N_k=N_{k+1}$
so $e_{k+1}\in T$, and
$T|_{G_{k+1}}=T|_{G_k}\cup e_{k+1}$, and let us prove that
$T|_{G_k}\cup e_{k+1}$ is a spanning forest in $G_{k+1}$.
First $T|_{G_k}\cup e_{k+1}$ contains no simple cycle. Since if it contained a simple cycle $\gamma$, 
then
$e_{k+1}$ would belong to $\gamma$ and therefore $T|_{G_k}$ would be
a spanning forest in $G_{k+1}$, so $b_1(G_{k+1})=b_1(G_{k})+1$ by equation \ref{e:spanning} which contradicts $b_1(G_k)=b_1(G_{k+1})$.
$T|_{G_k}\cup e_{k+1}$ is thus a forest.
$T|_{G_k}$ is spanning in $G_k$ hence
$T|_{G_k}\cup e_{k+1}$ meets all vertices of $G_{k+1}$ and is spanning
in $G_{k+1}$.

Case 2~: $b_1(G_k)+1=b_1(G_{k+1})$ and by definition
$T|_{G_{k+1}}=T|_{G_k}$. $T|_{G_k}$ is obviously a forest in $G_{k+1}$, its complement in $G_{k+1}$ contains
$b_1(G_k)+1=b_1(G_{k+1})$ edges by construction which implies it is spanning by equation \ref{e:spanning}.

Now we use induction to prove the uniqueness of $T$, in fact, we prove $T|_{G_k}$ is unique for any $k$. 
The initial step is trivial, and in general there are two cases. Either $b_1(G_k)=b_1(G_{k+1})$, then
$T|_{G_{k+1}}=T|_{G_k}\cup e_{k+1}$ or
$b_1(G_k)+1=b_1(G_{k+1})$ then
$T|_{G_{k+1}}=T|_{G_k}$, so our algorithm produces a unique spanning tree.
\end{proof}

\begin{coro}
\label{coro:UniqueCycle}
Let $(G,\ell)$ be a strict metric graph and $T$ be
the unique spanning forest in $T$ from Theorem \ref{t:inequalitiestree}. Then 
for every edge $e\in E(G)\setminus E(T)$, there is a unique simple cycle $\gamma_{e}$ 
in $T\cup e$, such that for any edge $e^\prime\in \gamma_e\setminus \{e\}$, $\ell (e)> \ell (e^\prime)$.
\end{coro}

\begin{proof}
Since $T$ is a spanning tree, $e\in E(G)\setminus E(T)$, so there is a unique simple cycle $\gamma_{e}$ in $T\cup e$. By our construction, if $e=e_{i_j}$, then $T|_{G_{i_j-1}}= T|_{G_{i_j}}$,
and $T|_{G_{i_j-1}}\cup e_{i_j}$ contains only one simple cycle $\gamma_{e}$, so for any edge $e^\prime\in \gamma_e$, $e^\prime\neq e$, $\ell(e)> \ell (e^\prime)$.
\end{proof}
%

\subsection{Approximation of the Riemannian distance in normal coordinates.}
For a smooth Riemannian manifold $(M,g)$, and any $x\in M$, $g(x)$ is an inner product in $T_xM$ which 
induces an isomorphism $g(x): T_xM\to  T_x^*M$ and thus an inner product $g^{-1}(x)$ on $T^*_xM$ by $g^{-1}(x)(w_1,w_2)=g(x)(g^{-1}(x)(w_1), g^{-1}(x)(w_2))$. This defines a smooth metric $g^{-1}$ on $T^*M$.

For every $x\in M$, we will use normal coordinates $(U, \phi, x^\mu )$ around $x$, without loss of generality $U$ is assumed to be geodesically
convex. The use of normal coordinates will be crucial since
it allows us to approximate the squared distance $\mathbf{d}^2(x,y)$ by
$\vert x-y\vert^2$ in local coordinates in Lemma \ref{lemmametric}. 
In some other coordinate chart, this approximation might not be as good.
On $(U, \phi, x^\mu) $, there are two metrics: the Riemannian metric $g$ and the Euclidean metric $h$:
$$h(\frac {\partial}{\partial x ^\mu}, \frac {\partial}{\partial x ^\nu})=h_{\mu\nu}=\delta_{\mu\nu}.$$
For this Euclidean metric, we will use $|x-y|$ to denote the induced distance. We recall that at the origin, we have the identity $g_{\mu\nu}(0)=h_{\mu\nu}=\delta_{\mu\nu}$. The following Lemma which dates back to Hadamard
can be found in \cite[Lemma 8.3 p.~90]{Duistermaat}, \cite[(A.3) p.~31]{nicolaescurandom}, \cite[(38) p.~171]{riesz1949}~:
\begin{lemm}[Hadamard]
\label{l:hadamard}
We denote by $\mathbf{d}$ the Riemannian distance and $\phi =\mathbf{d}^2$.
Then there exists a neighborhood $U$ of the diagonal $\subset M\times M$, such that $\phi\in C^\infty(U)$ and is symmetric, that is $\phi (x,y)=\phi (y,x)$, $\phi$ vanishes along the diagonal
at order $2$ and $\phi$ satisfies the Hamilton--Jacobi equation~:
\begin{eqnarray}\label{e:hadamarddist}
\boxed{g^{-1}\left(d_x \phi(x,y), d_x \phi(x,y)\right)=4\phi(x,y).}
\end{eqnarray}
\end{lemm}
Next we state an important Lemma which gives informations
on the jets of the function $\phi=\mathbf{d}^2$ along the
diagonal in $M\times M$.
\begin{lemm}\label{lemmametric}
For $x_0\in M$, if $(U,\phi)$ is a normal coordinate system
around $x_0$ such that $\phi (x_0)=0$ and the square of Riemannian distance $\phi$ is smooth on $U\times U$, then on $U\times U$,
\delete {let $\mathcal{I}_{d_2}$ denote the ideal of functions
vanishing on the diagonal $d_2=\{x=y\}\subset U\times U$
then~:}
\begin{equation}
\boxed{\phi(x,y)-g_{\mu\nu}(x)(x^\mu-y^\mu)(x^\nu-y^\nu),}
\end{equation}
vanishes along the diagonal
at order $3$.

\end{lemm}
The proof can be found in appendix \ref{p:lemmametric}. 

\subsection{Resolving singularities using spanning trees.}
\label{ss:blowup}
Let $(G,\ell)$ be a connected strict metric graph with edge set $E(G)$ identified with the set of integers
$\{1,\dots,E\}$ in such a way that $0\leqslant\ell_1 < \dots <\ell_E\leqslant 1$. This means that the
metric graph $(G,\ell)$ lies in a fixed sector $\mathbf{\Delta}=\{0<\ell _1< \dots < \ell _E<1 \}\subset
[0,1]^E$, $\overline{\mathbf{\Delta}}$ denotes its closure 
$\{0\leqslant\ell _1\leqslant \dots \leqslant \ell _E\leqslant 1 \}$. 
It is associated with
a unique spanning tree $T$ by Theorem~\ref{t:inequalitiestree} and the vertices of both graphs $G$ and $T$
are numbered by $\{1,\dots,n\}$.
For any $(i,j)\in \{1,\dots,n\}^2$, we denote by
$\overrightarrow{ij}$ the unique simple path in $T$ from $i$ to $j$.

The product $\prod_{e\in E(G)} e^{-\frac{\mathbf{d}^2(x_{i(e)},x_{j(e)})}{4\ell_e}} \psi(\mathbf{d}^2(x_{i(e)},x_{j(e)}))$, where $\psi$
is the cut--off function from Theorem \ref{heatasympt}, is not smooth
near the algebraic set $X=\cup_{e\in E(G)}\{\ell_e(x_{i(e)}-x_{j(e)})=0  \}\subset \mathbb{R}^{dn}\times\overline{\mathbf{\Delta}}$.
What we do next is give a recipe to resolve the singularities of such products
by some explicit map $\pi$ which is defined as follows~:

\begin{defi}
In the above notations,
the map from $\R^d \times (\mathbb{R}^{d})^{E(T)}\times [0,1]^E$ to $\mathbb{R}^{dn}\times\overline{\mathbf{\Delta}}$ is~:
\begin{eqnarray}\label{e:pi}
\boxed{\pi:\underset{\in \R^d \times (\mathbb{R}^{d})^{E(T)}\times [0,1]^E}{\underbrace{(x,(h_e)_{e\in {E(T)}},(t_k)_{k=1}^E)}} \longmapsto
\underset{\in \mathbb{R}^{dn}\times\overline{\mathbf{\Delta}}}{\underbrace{\left(
(x+\sum_{e\in \overrightarrow{1i}}
(\prod_{j\geqslant e} t_j) h_e)_{i=1}^n,
(\prod_{k\geqslant e}t_k^2)_{e=1}^E
\right)}}}
\end{eqnarray}
where the sum runs over all edges $e$ in the path $\overrightarrow{1i}$.
The map $\pi$ depends on spanning tree $T$ hence on the
strict ordering of
$E(G)$ induced by the metric $\ell$.
\end{defi}
In the sequel, we shall denote elements of the target space
$\mathbb{R}^{dn}\times\overline{\mathbf{\Delta}}$ by
$(x_i,\ell_e)_{1\leqslant i\leqslant n,1\leqslant e\leqslant E}$.
We check that the map $\pi:\R^d \times (\mathbb{R}^{d})^{E(T)}\times [0,1]^E  \mapsto \mathbb{R}^{dn}\times\overline{\mathbf{\Delta}}$ 
is a diffeomorphism outside some subset of measure zero.
\begin{prop}\label{p:blowupdiffeo} The map $\pi$ is a smooth diffeomorphism from
$\R^d \times \mathbb{R}^{d(n-1)}\times (0,1)^E$ to $\R ^{dn}\times \Delta
$.
\end{prop}

\begin{proof}
It is one to one
since we can explicitly invert
$\pi$ as
$t_E=\ell_E^{\frac{1}{2}}$,
$t_e=\left(\frac{\ell_{e+1}}{\ell_{e}}\right)^{\frac{1}{2}}\text{ for } e\leqslant E-1$
and the linear map:
$(x,(h_e)_{e\in E(T)})\in \R^d \times \mathbb{R}^{d(n-1)} \mapsto
\left(x_i=x+\sum_{e\in \overrightarrow{1i}}
\left(\prod_{j\geqslant e} t_j\right) h_e\right)_{i=1}^n\in \mathbb{R}^{dn}$
is \textbf{invertible} when $(t_j)_j\in (0,1)^E$.
Then the diffeomorphism property follows from an explicit calculation of the differential of $\pi$ whose determinant
does not vanish
when $t\in (0,1)^E$.
\end{proof}

Finally we may state the main 
Theorem from this section~:
\begin{thm}[Resolution of singularities.]
\label{thm:parabolicblowup}
Let $g$ be a Riemannian metric on $\R ^d$ and $\mathbf{d}:
\R ^d \times \R ^d\mapsto \mathbb{R}_{\geqslant 0}$ be the Riemannian distance 
whose injectivity radius is $\varepsilon$.
Let $\pi$ be the map defined by equation~(\ref{e:pi}).
For any $\psi(t)\in C_c^\infty(\mathbb{R})$
such that $\psi(t)=1$ when $t\leqslant \frac{\varepsilon^2}{4}$ and
$\psi(t)=0$ when $t>\frac{4\varepsilon^2}{9}$,
for every edge
$e\in E(G)$ bounded by the vertices $(i(e),j(e))$,
the pull--back~:
\begin{equation}
\pi^*\left( \psi(\mathbf{d}^2(x_{i(e)},x_{j(e)})) \frac{\mathbf{d}^2(x_{i(e)},x_{j(e)})}{\ell_e} \right)
\end{equation}
defines a smooth function in $\R ^d\times \mathbb{R}^{d(n-1)}\times [0,1]^E$.
\end{thm}
\begin{proof}
For $e \in E(T)$,
set $x_{j(e)}-x_{i(e)}=\pm h_e$. Recall that $h_e$ has in fact
$d$ components $(h_e^\mu)_{\mu=1}^d$.
Then by Lemma \ref{lemmametric}~:
\begin{eqnarray*}
\pi^*\left(\frac{\mathbf{d}^2(x_{i(e)},x_{j(e)})}{\ell_e} \right)&=& \frac{ g_{\mu\nu}(x_{i(e)})h^\mu_eh^\nu_e (\prod_{i\geqslant e} t_i)^2 +R(\pi^*x_{i(e)},\pi^*x_{i(e)}\pm (\prod_{i\geqslant e} t_i)h_e)}{(\prod_{i\geqslant e} t_i)^2}\\
 &=&g_{\mu\nu}(x_{i(e)})h^\mu_eh^\nu_e+r_e(t,x,h)
\end{eqnarray*}
where
$$r_e(t,x,h)=\frac{R(\pi^*x_{i(e)},\pi^*x_{i(e)}+(\prod_{i\geqslant e} t_i)h_e)}{(\prod_{i\geqslant e} t_i)^2}=\frac{O((\prod_{i\geqslant e} t_i)^3\Vert h_e\Vert^3)}{(\prod_{i\geqslant e} t_i)^2}=O((\prod_{i\geqslant e} t_i)\Vert h_e\Vert^3)$$ vanishes at order $3$ in $(h_e)_{e\in T}$ and at order $1$ in $(t_e)_{e=1}^E$ by Lemma \ref{lemmametric} and $r_e$ is smooth.

If $e\notin E(T)$, then
$$\pi^*\left( \frac{\mathbf{d}^2(x_{i(e)},x_{j(e)})}{\ell_e} \right)=\pi^*\left(\frac{g_{\mu\nu}(x_{i(e)})(x^\mu _{i(e)}-x^\mu_{j(e)})(x^\nu _{i(e)}-x^\nu_{j(e)})}{\ell_e} \right)+\pi^*\left(\frac{R(x_{i(e)},x_{j(e)})}{\ell_e}\right).$$
where
\begin{eqnarray*}
&&\pi^*\left(\frac{g_{\mu\nu}(x_{i(e)})(x^\mu _{i(e)}-x^\mu_{j(e)})(x^\nu _{i(e)}-x^\nu_{j(e)})}{\ell_e} \right)\\
&=& \frac{g_{\mu\nu}(x_{i(e)})(\sum_{e^\prime\in \gamma _e\setminus e}\varepsilon(e^\prime)\left(\prod_{j\geqslant e'} t_j\right)h^\mu _{e^\prime})(\sum_{e^\prime\in \gamma _e\setminus e}\varepsilon(e^\prime)\left(\prod_{j\geqslant e'} t_j\right)h^\nu_{e^\prime})}{(\prod_{i\geqslant e}t_i)^2}
\end{eqnarray*}
where $\varepsilon(e^\prime)=\pm 1$ and $\gamma _e$ is the unique simple cycle in $T\cup e$ in Corollary \ref {coro:UniqueCycle}.
The important fact is that
for all edge $e^\prime$ in the path $\gamma_e\setminus \{e\}$, $e^\prime <e$.
It follows that~:
\begin{eqnarray*}
&&\pi^*\left(\frac{g_{\mu\nu}(x_{i(e)})(x^\mu _{i(e)}-x^\mu_{j(e)})(x^\nu _{i(e)}-x^\nu_{j(e)})}{\ell_e} \right)\\
&=& g_{\mu\nu}(x_{i(e)})(\sum_{e^\prime\in \gamma _e\setminus e}\varepsilon(e^\prime)\left(\prod_{e^\prime\leqslant j<e} t_j\right)h^\mu _{e^\prime})(\sum_{e^\prime\in \gamma _e\setminus e}\varepsilon(e^\prime)\left(\prod_{e^\prime\leqslant j<e} t_j\right)h^\nu_{e^\prime})
\end{eqnarray*}
which is smooth 
since the product $(\pi_{i\geqslant e}t_i)^2$ 
on the denominator 
cancel out with the same 
powers appearing on the numerator.
The same argument applies to the remainder term
$\pi^*\left(\frac{R(x_{i(e)},x_{j(e)})}{\ell_e}\right)$.
\end{proof}

\subsection{Change of variables}
For a test function $\varphi$ supported in $\R ^{dn}$, since the map
$\pi$ is a smooth diffeomorphism from
$\R^d \times \mathbb{R}^{d(n-1)}\times (0,1)^E$ to $\R ^{dn}\times \Delta
$,
we can take it as a change of variables for integration:
\begin{eqnarray*}
\langle I_{G,\overrightarrow{k}}(s),\varphi\rangle
&=&\frac{1}{(4\pi)^{\frac{dE}{2}}}\left(\prod_{e=1}^E\frac{1}{\Gamma(s_e)}\right)\int_{[0,1]^E}d\ell_1\dots d\ell_E\int_{\mathbb{R}^{dn}} \prod_{e=1}^E \exp\left(-\frac{\mathbf{d}^2(x_{i(e)},x_{j(e)})}{4\ell_e}\right)\\
&&\psi(\mathbf{d}^2(x_{i(e)},x_{j(e)}))\ell_e^{s_e+k_e-\frac{d}{2}-1}a_{k_e}(x_{i(e)},x_{j(e)})\tilde \varphi d^dx_1\dots d^dx_n\\
&=&\frac{1}{(4\pi)^{\frac{dE}{2}}}\left(\prod_{e=1}^E\frac{1}{\Gamma(s_e)}\right)\sum_{\sigma\in S_E}\int_{\Delta_\sigma}d\ell_1\dots d\ell_E\int_{\mathbb{R}^{dn}} \prod_{e=1}^E \exp\left(-\frac{\mathbf{d}^2(x_{i(e)},x_{j(e)})}{4\ell_e}\right)\\
&&\psi(\mathbf{d}^2(x_{i(e)},x_{j(e)}))\ell_e^{s_e+k_e-\frac{d}{2}-1}a_{k_e}(x_{i(e)},x_{j(e)})\tilde \varphi d^dx_1\dots d^dx_n
\end{eqnarray*}
where $\sigma$ runs over the group $S_E$ of permutations of $\{1,\dots,E\}$.
The open simplices $\mathbf{\Delta}_\sigma$ do not cover the unit cube $[0,1]^E$. 
However the complement of $\cup_\sigma \mathbf{\Delta}_\sigma$
in $[0,1]^E$ has zero Lebesgue measure. Since for $Re(s_e)_{e=1,\cdots,E}$ large enough, the integral
$\int_{[0,1]^E} \prod_{e\in E(G)} \frac{e^{-\frac{\mathbf{d}^2}{4\ell_e}}}{(4\pi)^{\frac{d}{2}}} a_{k_e}
\psi(\mathbf{d}^2)\ell_e^{k_e-\frac{d}{2}+s_e-1}d\ell_e$ is absolutely convergent and depends holomorphically in $s\in \mathbb{C}^{E}$, 
we have the equality of integrals
$$ \int_{[0,1]^E} \prod_{e\in E(G)} \frac{e^{-\frac{\mathbf{d}^2}{4\ell_e}}}{(4\pi)^{\frac{d}{2}}} a_{k_e}
\psi(\mathbf{d}^2)\ell_e^{k_e-\frac{d}{2}+s_e-1}d\ell_e=\sum_{\sigma\in S(E)}\int_{\mathbf{\Delta}_\sigma } \prod_{e\in E(G)} \frac{e^{-\frac{\mathbf{d}^2}{4\ell_e}}}{(4\pi)^{\frac{d}{2}}} a_{k_e}
\psi(\mathbf{d}^2)\ell_e^{k_e-\frac{d}{2}+s_e-1}d\ell_e ,$$
where both sides depend holomorphically on $(s_e)_e$ for $Re(s_e)_{e=1,\cdots,E}$ large enough.

Now we can carry the change of variables in the fixed sector $\Delta=\{0 <\ell_1<\dots < \ell_E< 1  \}$ (the other terms will be obtained by permutation) which yields an expression of the form~:
\begin{eqnarray*}
&&\int_{\Delta} \prod_{e=1}^E\frac{d\ell_e}{\ell_e}\left(\int_{(\mathbb{R}^d)^n} \prod_{e=1}^E \exp\left(-\frac{\mathbf{d}^2(x_{i(e)},x_{j(e)})}{4\ell_e}\right)\psi(\mathbf{d}^2(x_{i(e)},x_{j(e)}))\ell_e^{s_e+k_e-\frac{d}{2}}
a_{k_e}(x_{i(e)},x_{j(e)})\tilde \varphi d^dx_1\dots d^dx_n\right) \\
&=&2^E\int_{[0,1]^E} \prod_{e=1}^E\frac{dt_e}{t_e}\int_{(\mathbb{R}^d)^n} \pi^*\left(\prod_{e=1}^E \exp\left(-\frac{\mathbf{d}^2(x_{i(e)},x_{j(e)})}{4\ell_e}
\right)\psi(\mathbf{d}^2(x_{i(e)},x_{j(e)}))\right) \\
&\times& \pi^*\left(\tilde \varphi\prod_{e\in E(G)}a_{k_e}\right)
\left(\prod_{e\in E(G)}\ell_e^{(s_e+k_e)-\frac{d}{2}}\right)
\left(\prod_{e\in E(T)}\ell_e^{\frac{d}{2}}\right) d^dx\prod_{e\in E(T)}d^dh_e.
\end{eqnarray*}
We can further simplify the product $\prod_{e\in E(G)}\ell_e(t)^{(s_e+k_e)-\frac{d}{2}}\prod_{e\in E(T)}\ell_e(t)^{\frac{d}{2}}$ as~:
\begin{eqnarray*}
&&\prod_{e\in E(G)}\left(\prod_{i\geqslant e}t_i \right)^{2(s_e+k_e)-d}\prod_{e\in E(T)}\left(\prod_{i\geqslant e}t_i \right)^d
=\prod_{e\in E(T)} \left(\prod_{i\geqslant e}t_i\right)^{2s_e+2k_e}\prod_{e\notin E(T)} \left(\prod_{i\geqslant e}t_i\right)^{2s_e-d+2k_e}\\
&=&\prod_{e\in E(G)}\left(\prod_{i\geqslant e}t_i\right)^{2s_e+2k_e}\prod_{e\notin E(T)}\left(\prod_{i\geqslant e}t_i\right)^{-d}\\
&=&t_E^{2s_E+2k_E}(t_Et_{E-1})^{2s_{E-1}+2k_{E-1}}\dots(t_E\dots t_1)^{2s_1+2k_1}(t_E\dots t_{i_k})^{-d}\dots (t_E\dots t_{i_1})^{-d}
\end{eqnarray*}
where $(i_1<\dots<i_k)\subset \{1,\dots,E\}$ are the numbers decorating the edges in the complement of $E(T)$ and $k=b_1(G)$,
hence~:
\begin{eqnarray}
\prod_{e\in E(T)} \ell_e(t)^{2s_e+2k_e}\prod_{e\notin E(T)} \ell_e(t)^{2s_e-d+2k_e}=
\prod_{e=1}^Et_e^{\sum_{i\leqslant e}2s_i+2k_i-db_1(G_e)}
\end{eqnarray}
where $G_e$ denotes the graph induced by the first $e$ edges
$\{1,\dots,e\}$.
This in turns implies that we obtain the simplified form~:
\begin{eqnarray}
&&\int_{\Delta} \prod_{e=1}^E\frac{d\ell_e}{\ell_e}\left(\int_{(\mathbb{R}^d)^n} \prod_{e=1}^E \exp\left(-\frac{\mathbf{d}^2(x_{i(e)},x_{j(e)})}{4\ell_e}\right)\psi(\mathbf{d}^2(x_{i(e)},x_{j(e)}))\ell_e^{s_e+k_e-\frac{d}{2}}
a_{k_e}(x_{i(e)},x_{j(e)})\tilde \varphi d^dx_1\dots d^dx_n\right) \notag \\
&=&\int_{[0,1]^E} \prod_{e=1}^E\frac{dt_e}{t_e}t_e^{(\sum_{i\leqslant e}2s_i+2k_i)-db_1(G_e)}\int_{(\mathbb{R}^d)^n}A(t_e,x,h_e) d^dx\prod_{e\in E(T)}d^dh_e
\label {eqn:IntSect}
\end{eqnarray}
where
\begin{eqnarray*}
A((t_e)_{e=1}^E,x,(h_e)_{e\in E(T)})&=& 2^E\pi^*\left(\left(\prod_{e=1}^E \exp\left(-\frac{\mathbf{d}^2(x_{i(e)},x_{j(e)})}{4\ell_e}
\right)\psi(\mathbf{d}^2(x_{i(e)},x_{j(e)}))a_{k_e}\right) \tilde \varphi\right).
\end{eqnarray*}

The coordinates $((t_e)_{e=1}^E,x,(h_e)_{e\in E(T)})$ on $[0,1]^E\times \R^d\times (\R^d)^{E(T)}$ will be shortly denoted
by $(t,x,h)$ for simplicity.
We now prove the smoothness in $t\in [0,1]^E$ of the partial integral
$\int_{(\mathbb{R}^d)^n}A(t,x,h) d^dx\prod_{e\in E(T)}d^dh_e$ which is needed to ensure analytic continuation.
\begin{lemm}\label{l:boundonA}
The map
$t\mapsto\int_{(\mathbb{R}^d)^n}A(t,x,h) d^dx\prod_{e\in E(T)}d^dh_e$
belongs to $C^\infty([0,1]^E)$.
\end{lemm}
\begin{proof}
The smoothness of $A$ is a direct consequence of Theorem \ref{thm:parabolicblowup}.
Now start from the definition of $A$~:
\begin{eqnarray*}
A(t,x,h)&=&2^E\pi^*\left(\prod_{e\in E(T)} \exp\left(-\frac{\mathbf{d}^2(x_{i(e)},x_{j(e)})}{4\ell_e}
\right)
\psi(\mathbf{d}^2(x_{i(e)},x_{j(e)}))
a_{k_e}\right) \\
&\times &  \pi^*\left(\prod_{e\notin E(T)}  \exp\left(-\frac{\mathbf{d}^2(x_{i(e)},x_{j(e)})}{4\ell_e}
\right)
\psi(\mathbf{d}^2(x_{i(e)},x_{j(e)}))
a_{k_e}\right) \\
&\times &\tilde \varphi(x+\sum_{e\in (r,1)}(\prod_{j\geqslant e}t_j)h_e,\dots,x+\sum_{e\in (r,n)}(\prod_{j\geqslant e}t_j)h_e).
\end{eqnarray*}

Then use the key fact that
since the open neighborhood ($\cong \R ^d$) is chosen small enough and has compact closure,
there exists
a fixed constant $\delta >0$ such that
for all $x,y\in \R ^d$, we have the following bound
on the Riemannian distance~:
\begin{eqnarray*}
\delta\vert x-y \vert\leqslant \mathbf{d}(x,y)=\vert x-y \vert+o(\vert x-y \vert) \leqslant \delta^{-1}\vert x-y \vert
\end{eqnarray*}
which follows from Lemma \ref{lemmametric}
since
$\mathbf{d}^2(x,y)-\sum g_{\mu\nu}(x)(x^\mu -y^\mu)(x^\nu-y^\nu)$ 
vanishes along the diagonal at order 3, 
and locally $\sum g_{\mu\nu}(x)(x^\mu -y^\mu)(x^\nu-y^\nu)$ 
is bounded by some multiple of $|x-y|$ by compactness of the neighborhood.
It follows that
$$\frac{ \delta^2\vert x-y \vert^2}{4\ell_e}\leqslant \frac{\bfd (x,y)^2}{4\ell_e} \leqslant \frac{\delta^{-2}\vert x-y \vert^2}{4\ell_e} $$
which implies that for all edges
$e\in E(T)$, we have the bound~:
$$  \pi^*\exp\left(-\frac{\mathbf{d}^2(x_{i(e)},x_{j(e)})}{4\ell_e}
\right)\leqslant \pi^*\exp(-\frac{ \delta^2\vert x-y \vert^2}{4\ell_e})=\exp\left(-\frac{\delta^2\vert h_e\vert^2}{4} \right). $$
This allows us to use the product
$\pi^*\left(\prod_{e\in E(T)} \exp\left(-\frac{\mathbf{d}^2(x_{i(e)},x_{j(e)})}{4\ell_e}
\right)\right)$ to control the exponential decay of $A$~:
\begin{eqnarray*}
\pi^*\left(\prod_{e\in E(T)} \exp\left(-\frac{\mathbf{d}^2(x_{i(e)},x_{j(e)})}{4\ell_e}
\right)
\right)\leqslant \prod_{e\in E(T)} \exp\left(-\frac{\delta^2\vert h_e\vert^2}{4} \right).
\end{eqnarray*}
By smoothness of the heat coefficients $a_k$, by compactness of the support of $\varphi\in C^\infty_c(U^n)$ thus of $\tilde \varphi$, and by the definition of $\pi$, for every multi-index $\alpha$, we find that there exists some constant $C_\alpha>0$
s.t.
\begin{eqnarray*}
\vert\partial_t^\alpha\pi^*\left(\tilde \varphi\prod_{e\notin E(T)}  \exp\left(-\frac{\mathbf{d}^2(x_{i(e)},x_{j(e)})}{4\ell_e}
\right)\prod_{e=1}^Ea_{k_e}\psi(\mathbf{d}^2(x_{i(e)},x_{j(e)}))\right)\vert\leqslant C_\alpha\left(1+\sum_{e\in E(T)}\vert h_e\vert \right)^{\vert\alpha \vert},
\end{eqnarray*}
the partial derivatives in $t$ contributes the powers of $h$.

Therefore we have the bound for all $(t,x,h)$~:
$$\vert \partial_t^\alpha A(t_e,x,h_e)\vert
\leqslant C_\alpha\prod_{e\in E(T)} \exp\left(-\frac{\delta^2\vert h_e\vert^2}{4} \right)
\left(1+\sum_{e\in E(T)}\vert h_e\vert \right)^{\vert\alpha \vert}$$
with $A$ compactly supported in $x$.
Then smoothness of $t\mapsto\int_{(\mathbb{R}^d)^n}A(t,x,h) d^dx\prod_{e\in E(T)}d^dh_e$ follows 
from smoothness of the integrand which is the function $A\in C^\infty( [0,1]^E\times (\mathbb{R}^d)^n)$, and any derivative $\partial_t^\alpha A$
has fast decrease in $h$ when $\vert h\vert\rightarrow +\infty$ and compact support in the variable $x\in \R^d$.
Therefore all derivatives in $\partial_t^\alpha A$ are integrable, uniformly in $t\in [0,1]^E$ and the conclusion follows
from classical results on integrals depending smoothly on parameters.
\end{proof}

\begin{lemm}\label{jetlemma}[Jet Lemma]
Fix the sector $\mathbf{\Delta}$ corresponding to the system of inequalities
$\{0\leqslant \ell_{1}\leqslant \dots\leqslant \ell_{E}\}$.
Let
\begin{eqnarray*}
A((t_e)_{e=1}^E,x,(h_e)_{e\in T})&=& \pi^*\left(\prod_{e=1}^E\left( \exp\left(-\frac{\mathbf{d}^2(x_{i(e)},x_{j(e)})}{4\ell_e}
\right)
\psi(\mathbf{d}^2(x_{i(e)},x_{j(e)}))a_{k_e}\right)\tilde \varphi\right).
\end{eqnarray*}
Then
the $k$-jet of $$\chi(t_1,\dots,t_E)=\int_{(\mathbb{R}^d)^n} A((t_e)_{e=1}^E,x,(h_e)_{e\in E(T)})d^dx\prod_{e\in E(T)}d^dh_e$$
depends continuously on the $k$-jet of $(a_{k_i}, i=1, \cdots E,\varphi,\mathbf{d}^2,g)$.
\end{lemm}
\begin{proof}
The claim follows from the formulas
defining $A$ and the change of variables $\pi$
and repeated application of the chain rule to
$\pi^*\left(\prod_{e=1}^E\left( \exp\left(-\frac{\mathbf{d}^2(x_{i(e)},x_{j(e)})}{4\ell_e}
\right)
\psi(\mathbf{d}^2(x_{i(e)},x_{j(e)}))a_{k_e}\right)\varphi\right)$.
%
\end{proof}
%


Recall that from 
Theorem 
\ref{t:technicalreduction}, the main Theorem \ref{mainthmintro}
reduces to an analytic continuation result
for the amplitudes $I_{G,\overrightarrow{k}}(s)$ corresponding
to the labelled Feynman graphs $(G,\overrightarrow{k})$.
The problem was that the integral formula for
$I_{G,\overrightarrow{k}}(s)$ involved some product of heat kernels which required
blow--ups which were performed in sectors. 
The following
Proposition shows how the integral expression $I_{G,\overrightarrow{k}}(s)$ simplifies after blow-up~:
\begin{prop}\label{p:IGksafterblowup}
Let us consider $I_{G,\overrightarrow{k}}(s)$ as in Equation (\ref{defIGks}),
for any $e\in \{1,\dots,E\}\simeq E(G)$ and any permutation $\sigma\in S_E$, $G_{\sigma(e)}$
is the subgraph of $G$ induced by the collection of edges 
$\{\sigma(1),\dots,\sigma(e)\}\subset E(G)$.
Then for every test function $\varphi\in C^\infty_c(U^n)$, there exists a family
$\chi_\sigma\in C^\infty([0,1]^E)$ indexed by $\sigma\in S_E$ such that
the identity~:
\begin{eqnarray}\label{e:reductionIGks}
\int_{\mathbb{R}^{dn}}I_{G,\overrightarrow{k}}(s)\varphi d^{nd}x=\prod_{e=1}^E\frac{1}{\Gamma(s_e)}
\sum_{\sigma\in S_E}\int_{[0,1]^E} \prod_{e=1}^E\frac{dt_e}{t_e}t_e^{\sum_{i\leqslant e}(2s_{\sigma(i)}+2k_{\sigma(i)})-db_1(G_{\sigma(e)})}\chi_\sigma(t)
\end{eqnarray}
holds true for all
$Re(s_e), e\in \{1,\dots,E\}$ large enough and both sides are holomorphic in $s\in \mathbb{C}^E$.
\end{prop}
In the next paragraph, we will proceed to the meromorphic continuation of the r.h.s of
equation (\ref{e:reductionIGks}) as meromorphic function with linear poles in $s$ and we will also bound
the distributional order of $I_{G,\overrightarrow{k}}(s)$ independently of the label $\overrightarrow{k}$.

\subsection{Integration by parts, bounding orders and pole decomposition.}
Now the proof of Theorem \ref{mainthmintro} on the analytic continuation of
$t_G(s)$ is reduced to the meromorphic continuation in
$s\in \C^E$ of the  
right hand side
of Equation (\ref{e:reductionIGks}).  
The meromorphic continuation comes from integration
by parts as shown
in the next Lemma and its corollary.
\begin{lemm}\label{l:analyticcontinuationmodelcase}
Let $E$ be a positive integer, for any smooth function $\psi$ on $[0,1]$,
$$
I_s(\psi)=\int_{[0,1]^E} t_1^{s_1}\dots t_E^{s_E}\psi(t_1,\dots,t_E)d^Et
$$ can be analytically extended to a meromorphic germ at $(s_e=p_e)_{e=1}^E \in\mathbb{Z}^E$, more precisely, let $I=\{i\ | \  p_i< 0\}$
\begin{equation}
\left (\prod _{i\in I}(s_i-p_i)\right )I_s(\psi)
\end{equation}
extends to a holomorphic germ at $(s_e=p_e)_e$ and
$I_s \in \calD^{\prime,m}([0,1]^E,\calm_{s_0}(\C^E) )$, $s_0=(p_1,\dots,p_E)$,
$m=\sum_{i\in I} \vert p_i\vert $.
\end{lemm}

The proof of this Lemma, given in the appendix, follows from integration by parts.
One consequence of this Lemma is
\begin{coro}\label{c:integrationbypartscorollary}
Denote by $1_{[0,1]^E}$ the indicator function
of the unit cube $[0,1]^E\subset\R^E$.
Let $(L_1, \cdots , L_E)$ be linear functions in $s\in \C^E$
with real coefficients $L_i\in (\R^E)^*, 1\leqslant i\leqslant E$, and $(a_1,\dots,a_E)\in \Z^E$. Set $I=\{i\ | \  a_i<0\}$, then
$$\left(\prod _{i\in I} L_i(s) \right )\int_{[0,1]^E} t_1^{L_1(s)+a_1}\dots t_E^{L_E(s)+a_E}\psi(t_1,\dots,t_E)d^Et
$$
is a holomorphic germ at $s=0\in \C^E$, and
$$1_{[0,1]^E} t_1^{L_1(s)+a_1}\dots t_E^{L_E(s)+a_E}
$$
extends to an element in
$\calD^{\prime,m}(\R^E ,\calm_0(\C^E))$
where $m=\sum_{i\in I} \vert a_i\vert $ and the polar set
is contained in $\{\prod _{i\in I}{L_i}=0\}$.
\end{coro}

Applying Corollary \ref{c:integrationbypartscorollary} to the r.h.s. of
Equation (\ref{e:reductionIGks}) shows that~:
\begin{lemm}
\label{l:polesinsector}
Let $S_\sigma$ be the set of all subgraphs $H\in \{G_{\sigma(1)}\subset\dots\subset G_{\sigma(E)}=G\}$ such that $b_1(H)\geqslant 1$, then
\begin{eqnarray*}
\prod_{H\in S_\sigma}\left(\sum_{e\in E(H)}s_e-\vert E(H)\vert\right)
\left(\int_{[0,1]^E} \prod_{e=1}^E\frac{dt_e}{t_e}t_e^{\sum_{i\leqslant e}(2s_{\sigma(i)}+2k_{\sigma(i)})-db_1(G_{\sigma(e)})}
\chi_\sigma(t)\right)
\end{eqnarray*}
is a \textbf{holomorphic germ} at $(s_e=1)_e$.
\end{lemm}
\begin{proof}
By Corollary \ref {c:integrationbypartscorollary} and a change of variables $s_e'=s_e-1, e\in\{1, \cdots, E\}$, 
we need to consider the following set of indices~:
$$I=\{e\ | \ e+\sum_{i\leqslant e}k_{\sigma(i)}-\frac d2 b_1(G_{\sigma(e)})\le 0 \}\subset \{1,\dots,E \}$$
which is contained in
$\{e \ | \ b_1(G_{\sigma(e)})\ge 1\},
$
which yields the conclusion.
\end{proof}

Let us comment on the above bound on the location of the pole.
First the bound seems suboptimal since the set of indices
$I=\{e\ | \ e+\sum_{i\leqslant e}k_{\sigma(i)}-\frac d2 b_1(G_{\sigma(e)})\le 0 \}\subset \{1,\dots,E \}$
is only a subset of $\{e \ | \ b_1(G_{\sigma(e)})\ge 1\}\subset \{1,\dots,E \}
$. However, it is important for us that
we can give a bound on the location of the poles which does not depend on the multi--index
$\overrightarrow{k}$ since poles from the original Feynman amplitude
$t_G(s)$ do not depend on $\overrightarrow{k}$.
The formula
of Theorem~\ref{t:technicalreduction} expresses $t_G$ as a sum
of $I_{G,\overrightarrow{k}}$ for some $\overrightarrow{k}$. Hence the poles of
$t_G(s)$ come from contributions from 
the poles of $I_{G,\overrightarrow{k}}$. Therefore it is convenient 
to have some $\overrightarrow{k}$ independent bound for poles of $I_{G,\overrightarrow{k}}$. 
Finally, we bound the distributional order
of $I_{G,\overrightarrow{k}}$ and also give precise location on the affine planes supporting the
poles of
$I_{G,\overrightarrow{k}}$ in the following~:
\begin{prop}[Poles of $I_{G,\overrightarrow{k}}$ and distributional order]
\label{p:keyproppolesanddistriborder}
Let $G$ be a graph whose set of edges is in bijection with $\{1,\dots,E\}$.
For any $e\in \{1,\dots,E\}\simeq E(G)$ and any permutation $\sigma\in S_E$, $G_{\sigma(e)}$
is the subgraph of $G$ induced by the collection of edges $\{\sigma(1),\dots,\sigma(e)\}\subset E(G)$.
For every permutation $\sigma\in S_E$, we associate the filtration
$\{G_{\sigma(1)}\subset\dots\subset G_{\sigma(E)}=G\}$,
and we consider the set $S_\sigma$ of all subgraphs $H\in \{G_{\sigma(1)}\subset\dots\subset G_{\sigma(E)}=G\}$ such that $b_1(H)\geqslant 1$.
For every $\overrightarrow{k}\in \mathbb{N}^{E(G)}$, the distribution $I_{G,\overrightarrow{k}}(s)$ defined
in definition \ref{defIGks} can be analytically continued
to $\calD^{\prime,m}(U^n,\calm_{s_0}(\C^E))$, $s_0=(s_e=1)_{e=1}^E$
where
\begin{equation}
m=\sum_{H\subset G; 2\vert E(H)\vert-db_1(H)-1<0} db_1(H)-2\vert E(H)\vert+1.
\end{equation}

For every
test function $\varphi$,
\begin{eqnarray*}
\int_{\mathbb{R}^{dn}}I_{G,\overrightarrow{k}}(s)\varphi d^{nd}x&=&
\sum_{\sigma\in S_E} \prod_{H\in S_\sigma} \frac{1}{\left(\sum_{i\in H}s_i-E(H) \right)} f_\sigma(s)
\end{eqnarray*}
where $f_\sigma$ is holomorphic germ at $(s_e=1)_e$.
\end{prop}

\begin{remark} The bound on the distributional order depends only on the topology of the graph $G$ and the dimension $d$ and not on
the element $\overrightarrow{k}\in \mathbb{N}^{E}$. The bound on the distributional order is also not sharp since we should only sum over subgraphs
$G^\prime\in \{G_{\sigma(1)},\dots,G_{\sigma(E)}\}$ such that $2\vert E(G^\prime)\vert-db_1(G^\prime)-1<0$ then
take the supremum over all permutations $\sigma$.
\end{remark}

\begin{proof}
We proved in Proposition \ref{p:IGksafterblowup} that
for every labelled graph $(G, \overrightarrow{k})$,  and every test function $\varphi\in C^\infty_c(U^{n}), n=\vert V(G)\vert$, there exists
a family $\chi_\sigma(t)$ of smooth functions
on the cube $[0,1]^E$ indexed by permutations $\sigma\in S_E$ such that:
\begin{eqnarray*}
\int_{\mathbb{R}^{dn}}I_{G,\overrightarrow{k}}(s)\varphi d^{nd}x&=&\prod_{e=1}^E\frac{1}{\Gamma(s_e)}
\sum_{\sigma\in S_E}\int_{[0,1]^E} \prod_{e=1}^E\frac{dt_e}{t_e}t_e^{\sum_{i\leqslant e}(2s_{\sigma(i)}+2k_{\sigma(i)})-db_1(G_{\sigma(e)})}\chi_\sigma(t).
\end{eqnarray*}

By applying Lemma \ref{l:polesinsector} to
$$\sum_{\sigma\in S_E}\int_{[0,1]^E} \prod_{e=1}^E\frac{dt_e}{t_e}t_e^{\sum_{i\leqslant e}(2s_{\sigma(i)}+2k_{\sigma(i)})-db_1(G_{\sigma(e)})}\chi_\sigma(t),$$
we obtain the meromorphic continuation of $s\mapsto\int_{\mathbb{R}^{dn}}I_{G,\overrightarrow{k}}(s)\varphi d^{nd}x$
with the bound on the location of poles.
To show that $I_{G,\overrightarrow{k}}(s)$ is actually an element
of $\calD^{\prime,m}(U^n,\calm_{s_0}(\C^E))$, $s_0=(s_e=1)_{e=1}^E$, we need to show that
$$\varphi\in C^\infty_c(U^n) \mapsto  \int_{\mathbb{R}^{dn}}I_{G,\overrightarrow{k}}(.)\varphi d^{nd}x \in \calm_{s_0}(\C^E) $$ depends linearly on the $m$-jet of $\varphi$ for some $m$. By Corollary \ref {c:integrationbypartscorollary},
the integral $$\int_{[0,1]^E} \prod_{e=1}^E\frac{dt_e}{t_e}t_e^{\sum_{i\leqslant e}(2s_{\sigma(i)}+2k_{\sigma(i)})-db_1(G_{\sigma(e)})}\chi_\sigma(t)$$
depends linearly on the $m$-jet of $\chi_\sigma$ 
for $$m=\sum_{G^\prime\subset G; 2\vert E(G^\prime)\vert-db_1(G^\prime)-1<0} db_1(G^\prime)-2\vert E(G^\prime)\vert+1.$$
Then by Lemma \ref{jetlemma}, the $m$-jet of $\chi_\sigma$ depends continuously
on the $m$-jet of $\varphi$ which yields the result.
\end{proof}

Now let us restate the first
main Theorem from our paper
and conclude its proof~:
\\
\\
\fbox{
\begin{minipage}{0.94\textwidth} 
\begin{thm}
\label{mainthm}
Let $(M,g)$ be a smooth, compact, connected Riemannian manifold without boundary of dimension $d$, $dv(x)$ the Riemannian volume and $P=-\Delta_g+V$, $V\in C^\infty_{\geqslant 0}(M)$ 
or $M=\mathbb{R}^d$ with a constant metric $g$ and $P=-\Delta_g+\lambda^2, \lambda\in \R_{\geqslant 0}$.
Then for every
graph $G$,
\begin{equation}
t_G(s)=\prod_{e\in E(G)}\greenf^{s_e}(x_{i(e)},x_{j(e)})
\end{equation}
can be analytically continued 
as an element
of $\calD^{\prime}(M^{V(G)},\calm_{s_0}(\C^{E(G)}))$
where $s_0=(s_{e}=1)_{e \in E(G)}$, 
with linear poles supported
on the union of affine hyperplanes
$$\bigcup_{G^\prime}\{\sum_{e\in G^\prime}s_e-\vert E(G^\prime)\vert=0\}$$
where the union runs over subgraphs $G^\prime$ of $G$ such that
$2\vert E(G^\prime)\vert-b_1(G^\prime)d\leqslant 0 $ and $\vert E(G^\prime)\vert$ is the number of edges
in $G^\prime$.
\end{thm}
\end{minipage}
}
\\
\\
\begin{proof}
From Theorem
\ref{t:technicalreduction},
one has a decomposition~:
\begin{eqnarray*}
t_G(s)|_{U_x^n}&=&\sum_{G^\prime\subset G}
\left(\sum_{\overrightarrow{k}\in \{0,\dots,p\}^{E_1}} \underbrace{I_{G^\prime,\overrightarrow{k}}(s)}\right)
\times h_{G\setminus G^\prime}(s)
\end{eqnarray*}
where the sum runs over subgraphs  $G^\prime\subset G$ and 
$h_{G\setminus G^\prime}(s)\in C^m(M^{V(G)},\mathcal{O}_{s_0}(\C^{E(G)\setminus E(G^\prime)}))$, 
$s_0=(s_e=1)_{e\in E(G)\setminus E(G^\prime)}$ as soon as 
$p>\frac{d+m}{2}-1$ for $m=\sum_{G^\prime\subset G; 2\vert E(G^\prime)\vert-db_1(G^\prime)-1<0} db_1(G^\prime)-2\vert E(G^\prime)\vert+1$. Therefore,
the analytic continuation
of $t_G(s)$ should follow
from the analytic continuation of
$I_{G^\prime,\overrightarrow{k}}$ for all subgraphs $G^\prime$ of $G$
and the fact that the distributional order of
$I_{G^\prime,\overrightarrow{k}}$ is bounded from above by some
integer $m(G^\prime)$ independent of $\overrightarrow{k}\in \mathbb{N}^{E(G^\prime)}$.
But Proposition \ref{p:keyproppolesanddistriborder} precisely gives us
that the distributional order of every $I_{G^\prime,\overrightarrow{k}}$ is bounded from above by some
integer which depends only on the topology of $G^\prime$. 
Now following the notations from Proposition \ref{p:keyproppolesanddistriborder}, 
for every subgraph $G^\prime\subset G$, 
we denote
by
$S_{E(G^\prime)}$ the permutations of the edges $E(G^\prime)=\{1,\dots,E^\prime\}$. For every permutation
$\sigma\in S_{E(G^\prime)}$ corresponds a canonical filtration
$\{G^\prime_{\sigma(1)}\subset\dots\subset G^\prime_{\sigma(E^\prime)} \}$ of $G^\prime$ and
$S_\sigma$ denotes
the set of all subgraphs
$H\in \{G^\prime_{\sigma(1)}\subset\dots\subset G^\prime_{\sigma(E^\prime)} \}$ such that
$b_1(H)\geqslant 1$.
Finally doing all the bookkeeping, we find that
\begin{eqnarray*}
t_G(s)=\sum_{G^\prime\subset G}\sum_{\sigma\in S_{E(G^\prime)}}\left(\prod_{H\in S_\sigma} \frac{1}{(\sum_{i\in H}s_i)-E(H)} \right)h_\sigma(s)
\end{eqnarray*}
where $h_\sigma(s)\in \calD^\prime(M^{V(G)},\mathcal{O}_{s_0}(\C^{E(G)}))$, $s_0=(s_e=1)_{e\in E(G)}$.
\end{proof}

\section{Renormalization of Feynman amplitudes}
\label{s:renormalization}

In this second part of our paper, we shall apply the analytic continuation results
derived in the previous part to the renormalization of Feynman amplitudes
on Riemannian manifolds.

\subsection{Renormalization maps.}
\label{p:notationsdefirenorm}
For a smooth manifold $(M,g)$ and for every finite subset $I\subset \N$, we denote by $M^I$ the configuration space of points labelled by $I$. 
For $J\subset I$ with $|J|\ge 2$, $D_{J}$ is the subset $\{(x_i)_{i\in I} \text{ s.t. }x_j=x_k\text{ for }j,k\in J\}$ of $M^I$ called $J$-diagonal. Let $\Delta _I =\cup _{J\subset I, |J|\ge 2}D_{J}$ be the maximal diagonal.

\begin{defi}[Labelling vertices] Let 
$I\subset \mathbb{N}, \vert I\vert<+\infty$ be
some finite subset of integers. A \textbf{graph with vertices labelled by $I$} 
is a pair $(G, \iota)$, where $G$ is a graph, $\iota$ is an injective map from $V(G)$ to $I$.
\end{defi}

For a graph with vertices labelled by $I$, $(G, \iota)$, defines an element
$$t_{G}=\underset{e\in E(G)}{\prod} \greenf (x_{i(e)},x_{j(e)}),$$
where $(i(e),j(e))\in I^2$ and $t_{G}$
is a smooth function on $M^I\setminus \Delta _I$. 
For a finite subset $I$ of $\N $, let us denote 
by $\mathcal{F}(M^I)$ the linear span of $t_{G}$ of all 
$(G, \iota)$ with $\iota (V(G))\subset I$ as smooth 
functions on $M^I\setminus \Delta _I$.

For a linear map
$\mathcal{R}:E\mapsto \mathcal{D}^\prime(M)$
where $E$ is a vector space and $M$ is a smooth manifold, 
and any open subset $U\subset M$, let $i_U:U\hookrightarrow M$ denote the inclusion map. 
Then $\mathcal{R}|_U=i_U^*\mathcal{R}:E\mapsto \mathcal{D}^\prime(U)$ is the pull--back of $\mathcal{R}$ by $i_U$.
Following recent work by Nikolov--Stora--Todorov~\cite{NST}, we can give a definition
of renormalization as follows~:
\begin{defi}\label{d:Functionalequations}
A renormalization is a sequence
of (not necessarily continuous) linear maps $\mathcal{R}_{I}:\mathcal{F}(M^I)\mapsto \mathcal{D}^\prime(M^I)$ indexed by finite
subsets $I$ of $\mathbb{N}$, which satisfies the following
system of \textbf{functional equations}~:

\begin{itemize}
\item For $I\subset J$, and $ t\in \mathcal{F}(M^I)$,
\begin {equation}
\boxed{\calr_{J}(t)=\calr _I (t).}
\end{equation}
This is the \textbf {compatibility} condition for the family of linear maps.
\item $\forall t\in \mathcal{F}(M^I), \varphi\in C^\infty_c(M^I\setminus \Delta _I),$
\begin{equation}
\boxed{\langle\mathcal{R}_{I}(t),\varphi\rangle=\langle t,\varphi\rangle.}
\end{equation}
This means that
$\mathcal{R}_{I}(t)$ is a \textbf{distributional extension} of $t\in C^\infty(M^I\setminus \Delta_I)$.
\item For a graph $(G, \iota )$ with vertices labelled by $J\subset \N$, and $ I\subset J=\iota (V(G))$, set $I^c=J\setminus I$, 
let $E_I=\{e\in E(G) ; i(e),j(e)\in I^2\},E_{I^c}=\{e\in E(G) ; i(e),j(e)\in I^{c2}\},E_{II^c}=E(G)\setminus \left(E_I\cup E_{I^c}\right)$, and we denote by
$G_I,G_{I^c},G_{II^c}$ the corresponding induced subgraphs of $G$.
For open subsets $U,V$ of $M$ with $\text{dist}(U,V)>0$, denote by 
$U^I\times V^{I^c}$ the subset
$\{ (x_j)_{j\in J}\in M^J \text{ s.t. }x_i\in U, \forall i\in I, x_i\in V, \forall i\in I^c \}\subset M^J$.
Then~:
\begin{eqnarray*}
\boxed{\mathcal{R}_{J}|_{U^I\times V^{I^c}}(t_{G})=\left(\mathcal{R}_{J}|_{U^I}(t_{G_I} )\boxtimes \mathcal{R}_{J}|_{V^{I^c}}(t_{G_{I^c}})\right)  t_{G_{II^c}}  }
\end{eqnarray*}
as distributions in $\mathcal{D}^\prime(U^I\times V^{I^c})$. This means that renormalization must preserve \textbf{locality}.

\item 
Let $\Phi:M\mapsto M$ be an orientation preserving diffeomorphism and 
we denote by $\Phi_I:M^I\mapsto M^I $ the induced
diffeomorphism on configuration space $M^I$. We assume 
that the renormalization maps
depend on the Riemannian metric $g$ and denote it by 
$\mathcal{R}[g]=( \mathcal{R}[g]_I)_I$ to stress this dependence.
Then the covariance equation on renormalization maps
reads for all graph $(G, \iota)$ with vertices labelled by $I$~:
\begin{equation}
\boxed{\mathcal{R}[\Phi^*g]_I\left( \Phi_I^*t_G\right)= \Phi_I^*\left(\mathcal{R}[g]_I\left(t_G\right)\right).}
\end{equation}
This axiom of functorial nature ensures the renormalization is \textbf{covariant}.
\delete {\item For every diffeomorphism
$\Phi:N\mapsto M $ of $N$ onto an open submanifold of $M$ and a Riemannian structure $g$ on $M$, then
\begin{equation}
\forall t\in \mathcal{F}(M^I),\,\ \mathcal{R}_{N^I}[\Phi^*g](\Phi^*t)=\Phi^*\left( \mathcal{R}_{M^I}[g](t)\right).
\end{equation}
This ensures the renormalization is \textbf{covariant}. \zb {It seems this is redundant }}
\end{itemize}
\end{defi}

The following property follows from the locality condition~: for a graph
$(G, \iota )$ with vertices labelled by $I$, if $G$ is the disjoint union of $G_1 $ and $G_2$, $\iota (V(G_1))\subset I_1$, $\iota (V(G_2))\subset I_2$, $I_1\cap I_2=\emptyset$
then~:
\begin{eqnarray*}
\boxed{\mathcal{R}_{I_1\cup I_2}(t_{G})=\calr _{I_1}(t_{G_1}) \boxtimes \calr _{I_2}(t_{G_2}) }
\end{eqnarray*}
as distributions in $\mathcal{D}^\prime(M^{I_1\cup I_2})$.

\subsection{Decompositions of \gmds}
\label{ss:decompositiongermsdistrib}
Our goal in this
paragraph is to extend
the decomposition
of \cite{guopaychazhang2015} (see also \cite[Appendix]{berline2016local})
of the space $\calm_{s_0}$ of meromorphic germs with linear poles
at $s_0\in \R^p\subset \C^p$ to their distributional counterpart
$\calD^\prime(.,\calm_{s_0})$ defined in 
paragraph~\ref{ss:meromgermsdistrib}.
This decomposition plays an essential role
in our definition of renormalization maps by projections.
Recall we
denoted by $\calo_{s_0}$ the space of holomorphic germs at $s_0$.

Let us fix a
nondegenerate bilinear form
$$Q(\cdot, \cdot): \R ^p\times \R ^p\to \R, $$
 which induces a nondegenerate bilinear form
 $$Q^*(\cdot, \cdot): (\R ^p)^*\times (\R ^p)^*\to \R.
 $$
We can now define the concept of polar germ~\cite[definition 2.3]{guopaychazhang2015}, a polar germ at $s_0$ is a meromorphic germ of the form $\frac 1{L_1^{n_1}(s-s_0)\cdots L_k^{n_k}(s-s_0)}h(\ell (s-s_0))$, where $L_1, \cdots, L_k $ are linearly independent linear functions in $(\R ^p)^*$, $(n_1, \cdots n_k) \in \N_{>0}^k$, $\ell =(\ell _1, \cdots , \ell _n): \R ^p \to \R ^n$ defined by linear functions $\ell _1, \cdots , \ell _n$, and $h$ is a holomorphic germ at $0\in \C^n$, such that $Q^*(L_i, \ell _j)=0$, $1\leqslant i\leqslant k$, $1 \leqslant j \leqslant  n$. 
Let $\mathcal{P}_{s_0}$ be the linear span of polar germs at $s_0$ in $\calm _{s_0}$.

Notice the polar set is defined by real linear functions, by similar proof as in \cite{guopaychazhang2015} using
some geometry of cones, we have the following~:
\begin{prop}
 \label{prop:germd}
 There is a decomposition:
 $$\calm _{s_0}=\calo_{s_0}\oplus \mathcal{P}_{s_0}.
 $$
\end{prop}

 Now we extend the
concept of polar germs to distributions valued in polar germs.
\begin{defi} 
Let $M$ be a smooth manifold and $p\in \N$.
A polar germ of distributions at $s_0\in \R^p\subset \C^p$ is some element of $\calD^\prime(M,\calm_{s_0}(\C^p))$ 
of the form $\frac 1{L_1^{n_1}(s-s_0)\cdots L_k^{n_k}(s-s_0)}h(\ell (s-s_0))$, where $L_1, \cdots, L_k $ are linearly independent linear functions in $(\R ^p)^*$, $(n_1, \cdots n_k) \in \N_{>0}^k$, $\ell =(\ell _1, \cdots , \ell _n): \R ^p \to \R ^n$ defined by linear functions $\ell _1, \cdots , \ell _n$, and $h\in \calD^\prime(M,\calo_{0}(\C^p))$ such that $Q^*(L_i, \ell _j)=0$, $1\leqslant i\leqslant k$, $1 \leqslant j \leqslant  n$. 
We denote by $\mathcal{D}^\prime(M,\mathcal{P}_{s_0}(\C^p))$ the linear span of polar germs of distributions at $s_0$ in $\calD^\prime(M,\calm_{s_0}(\C^p))$.
\end{defi}

\begin{lemm} If $\frac 1{L_1^{n_1}(s-s_0)\cdots L_k^{n_k}(s-s_0)}h(\ell (s-s_0))$ and $\frac 1{M_1^{m_1}(s-s_0)\cdots M_p^{m_p}(s-s_0)}g(\ell ^\prime (s-s_0))$ represent the same nonzero meromorphic germ of distributions, then $k=mp$,  and   $M_1, \cdots, M_m, L_1,\dots, L_k$ can be rearranged in such a way that $L_i$ is a multiple of $M_i$ and $n_i=m_i$ for $1\leq i\leq k$.
\end{lemm}

\begin{proof} Since this meromorphic germ is not zero, we can take a test function $\varphi $ such that $h(\ell (s-s_0))(\varphi)$ is not identically zero, then $\frac 1{L_1^{n_1}(s-s_0)\cdots L_k^{n_k}(s-s_0)}h(\ell (s-s_0))(\varphi)$ and $\frac 1{M_1^{m_1}(s-s_0)\cdots M_p^{m_p}(s-s_0)}g(\ell ^\prime (s-s_0))(\varphi)$ represent the same polar germ, by the same proof as in \cite[Lemma 2.8]{guopaychazhang2015}, we have the conclusion.
\end{proof}
We then prove the promised decomposition Theorem for $\calD^\prime(M,\calm_{s_0})$ which 
generalizes the result in \cite{guopaychazhang2015}.
\begin{thm}\label{t:decompgmd}
Let $M$ be some smooth manifold and $s_0\in\R^p\subset \C^p$. 
We have the direct sum decomposition~:
$$\calD^\prime(M,\calm_{s_0}(\C^p))=\calD^\prime(M,\calo_{s_0}(\C^p))\oplus \calD^\prime(M,\mathcal{P}_{s_0}(\C^p)).$$
\end{thm}

\begin{proof} Without loss of generality, 
we can assume that $s_0=0$. For $t \in \calD^\prime(M,\calm_{0}(\C^p))$, by definition, there exist $L_1, \cdots ,L_k\in (\R ^p)^*$ such that $L_1\cdots L_k t \in \calD^\prime(M,\calo_{0}(\C^p))$.
By partial fractions decompositions as in the proof of \cite[Lemma 2.9 property a)]{guopaychazhang2015}, we may
assume there is
$(n_1,\dots,n_k)\in \N_{>0}^k$
such that ${L_1^{n_1}\cdots L_k^{n_k}}t\in \calD^\prime(M,\mathcal{O}_0(\C^p))$ with $L_1, \cdots, L_k $ linearly independent and $(n_1,\cdots,n_k) \in \N^k_{>0}$.

Now let us expand $L_1, \cdots, L_k $
to a basis $\left(e_1, \dots,e_p\right)$ of $(\R ^p)^*$ with $e_i=L_i, 1\leqslant i \leqslant k$ and $Q(e_i, e_j)=0$ for 
$1\leqslant i \leqslant k$,  
$k+1\leqslant j \leqslant p$. Then by
Proposition \ref{weakstrongholo}, we have the power series expansion for ~:
$$z_1^{n_1}\cdots z_k^{n_k}t=\sum_{\alpha\in \mathbb{N}^p} \frac {z^\alpha}{\alpha !} t_\alpha,
$$ where $z=\sum z_ie_i^* \in (\C ^p)^*$.
So when we apply $z_1^{n_1}\cdots z_k^{n_k}t$ against the test function $\varphi$, we
obtain
$z_1^{n_1}\cdots z_k^{n_k}t(\varphi)=\sum_{\alpha\in \mathbb{N}^p} \frac {z^\alpha}{\alpha !} t_\alpha(\varphi),
$
which is absolutely convergent in a small neighborhood of $0\in \C^p$.

Let $S=\{d=(d_1, \cdots , d_p)\in \N ^p\ |\ d\neq (0,\dots,0), 0\leqslant d_i\leqslant n_i \}$. For $d\in S$, let
$I(d)=\{i\ | \ d_i\not =0\}\subset \{1,\dots,p\}$, and set
$$N_d=\{\alpha \in \N ^p\ |\ \alpha _i=n_i-d_i \text{ if } i\in I(d), \ \alpha _i\ge n_i \text{ if } i \in \{1,\dots,k\}\setminus I(d), \alpha _i\in \N \text{ if } i\in\{k+1,\dots,p\} \}.
$$
Then we note that $N_{d}\cap N_{e}=\emptyset $ if $d\not =e\in S$ and
most importantly, we have the partition $\N ^p=\bigcup _{d\in S}N_d$.

Now for $z_i\not =0, 1 \leqslant i \leqslant k$,
$$t(\varphi )=\sum_{\alpha\in \mathbb{N}^p} \frac {z^{\alpha -\alpha_0}}{\alpha !}t_\alpha (\varphi)=\sum _{d\in S}\sum _{\alpha \in N_d}\frac {z^{\alpha -\alpha_0}}{\alpha !}t_\alpha (\varphi),
$$
where $\alpha _0=(n_1, \cdots, n_k, 0,\cdots,0)$. And we have
$$\sum _{\alpha \in N_d}\frac {z^{\alpha -\alpha_0}}{\alpha !}t_\alpha (\varphi)=\frac 1{z_{I(d)}^d}\sum _{\alpha \in N_d}\frac {z_{[p]\setminus I(d)}^{\alpha -\alpha_0}}{\alpha !}t_\alpha (\varphi),
$$
where $z_{I(d)}^d=\prod _{i\in I(d)}z_i^{d_i}$, $z_{[p]\setminus I(d)}^{\alpha -\alpha_0}=\prod _{i\in \{1,\dots,k\}\setminus I(d)}z_i^{\alpha _i-s_i}\prod _{i\in \{k+1,\dots,p\}}z_i^{\alpha _i}$.

Let
$$h_d=\sum _{\alpha \in N_d}\frac {z_{[p]\setminus I}^{\alpha -\alpha_0}}{\alpha !}t_\alpha .$$ By Proposition \ref{weakstrongholo}, for every compact
$K\subset U$ there exists $C>0$ and some continuous seminorm $P$ for the Fr\'echet topology
of $C^\infty_K(U)$ such that~:
$\forall \varphi\in C^\infty_K(U), \vert t_\alpha(\varphi)\vert\leqslant \frac{\alpha!}{r^{\vert\alpha\vert}}  CP(\varphi) $
for all multi-indices $\alpha\in \mathbb{N}^p$.
This implies that for every $0<R<r$,
for all
$\vert z\vert<R$,
\begin{eqnarray*}
\vert \sum_{\alpha\in \mathbb{N}^p} \frac {z^\alpha}{\alpha !} t_\alpha (\varphi) \vert\leqslant
\sum_{\alpha\in \mathbb{N}^p} \vert \frac {z^\alpha}{\alpha !} t_\alpha(\varphi)\vert
\leqslant
\sum_{\alpha\in \mathbb{N}^p}  \frac {R^{\vert\alpha\vert}}{\alpha !} \frac{\alpha!}{r^{\vert\alpha\vert}}  CP(\varphi)=
\left(1-\frac{R}{r}\right)^{-p}
CP(\varphi).
\end{eqnarray*} Therefore,
$
\vert (\prod _{i\in [k]\setminus I(d)}z_i^{s_i}) h_d(\varphi) \vert\leqslant
\sum_{\alpha\in \mathbb{N}^p} \vert \frac {z^\alpha}{\alpha !} t_\alpha(\varphi)\vert
\leqslant
\left(1-\frac{R}{r}\right)^{-p}
CP(\varphi),
$ so $(\prod _{i\in [k]\setminus I(d)}z_i^{s_i}) h_d\in \calD^\prime(M,\calo_0(\C^p))$, so $h_d\in \calD^\prime(M,\calo_0(\C^p))$.

   Now by definition, $\frac 1{z_{I(d)}^d} h_d$ is a polar germ of distributions if $d\not =(0, \cdots , 0)$, $h_{(0,\cdots , 0)}\in\calD^\prime(M,\calo_0(\C^p))$, and
 \begin{equation}\label{e:decompgmddetailed}
 \boxed{  t=\sum _{d\in S}\frac 1{z_{I(d)}^d} h_d \in  \calD^\prime(M,\mathcal{P}_{0}(\C^p)) +\calD^\prime(M,\mathcal{O}_{0}(\C^p))}
 \end{equation}
where the singular part reads as a finite sum of polar germs as a corollary of
the above argument.
So we have
$\calD^\prime(M,\calm_{s_0}(\C^p))=\calD^\prime(M,\calo_{s_0}(\C^p)) + \calD^\prime(M,\mathcal{P}_{s_0}(\C^p)).$
To show it is a direct sum, if $t\in \calD^\prime(M,\calo_{s_0}(\C^p)) \cap \calD^\prime(M,\mathcal{P}_{s_0}(\C^p))$, then for any test function $\phi $, $t(\phi ) \in \mathcal{P}_{s_0}\cap \calo_{s_0}$, so $t(\phi)=0$ by Proposition \ref{prop:germd}, which implies $t=0$.
\end{proof}

A consequence of the decomposition Theorem is the following 
\begin{prop}\label{p:proj}
Let $M$ be a smooth manifold, $p\in \N$ and $s_0\in \R^p\subset \C^p$.
There exists
a 
projection
$$\boxed{\pi_p:\calD^\prime(M,\calm_{s_0}(\C^p))\mapsto \calD^\prime(M,\calo_{s_0}(\C^p))}$$
which sends distribution valued in meromorphic germs at $s_0$
to distribution valued in holomorphic germs at $s_0$ such that
$\ker(\pi_p)=\calD^\prime(M,\mathcal{P}_{s_0}(\C^p))$. 
\end{prop}
\begin{remark}
Note that $\pi_p$ is uniquely
determined by the vector subspace of polar germs which are in turn 
uniquely determined by the choice of
the canonical quadratic form
$Q:\R^p\times \R^p\to \R$ that we fixed at the beginning of the present section.
\end{remark}
 
In appendix \ref{p:productsofgmd}, we show some useful Lemmas
on the functorial properties of the projection $\pi_p$ for $p\in \mathbb{N}$.
As a consequence of~\cite{guopaychazhang2015}, we have a similar projection, still denoted
a bit abusively by $\pi_p$, at the germ level~: $\pi_p:\calm_{s_0}\mapsto \calo_{s_0}$. It follows that
the two projectors are related by the following 
equation~:
\begin{coro} 
\label{coro:pi}
Let $X$ be a smooth manifold and $s_0\in \R^p\subset \C^p$.
For all $t(s) \in \calD^\prime(X,\calm_{s_0}(\C^p))$ and for all test function $\varphi\in C^\infty_c(X)$,
$$\boxed{(\pi_p t(s))(\varphi)=\pi_p (t(s)(\varphi)).}
$$
\end{coro}

\subsection{A renormalization map by projections}
\label{ss:renormproj}

From now on, for any integer $p$, we fix
the canonical quadratic form $Q$ on $\R^p$~: $Q(x)=\sum_{i=1}^p \vert x_i\vert^2$ 
and we study germs at $s_0=(1,\dots,1)\in \R^p$. 
We denote by $\mathbf{ev}|_{s_0}$ the evaluation of
some holomorphic germ at $s_0$.
The properties of the family of projections $(\pi_p)_{p\in \N}$
allow to us to give a definition of
renormalization maps as follows~:
\begin{defi}[Renormalization maps by projections]
\label{d:renormmapszeta}
For $I \subset \N$, we define the renormalization map $\calr _I$ as follows: for a graph $(G, \iota)$ with vertices labelled by 
$I$,
\begin{eqnarray*}
\boxed{\mathcal{R}_{I}(t_{G})=\mathbf{ev}|_{s_0}\left(\pi_{|E(G)|}\left( t_G(s) \right)\right)  }
\end{eqnarray*}
where $ \greenf^{s}$ is the Schwartz kernel of $(-\Delta)^{-s}$.
\end{defi}

\begin{thm}[Renormalization Theorem]
\label{t:renormtheorem}
Let $(M,g)$ be a smooth, compact, connected Riemannian manifold without boundary of dimension $d$, $dv(x)$ the Riemannian volume and $P=-\Delta_g+V$, $V\in C^\infty_{\geqslant 0}(M)$ 
or $M=\mathbb{R}^d$ with a constant metric $g$ and $P=-\Delta_g+m^2, m\in \R_{\geqslant 0}$.
For every finite subset $I\subset \N$, for every
graph $(G,\iota)$ with vertices labelled by $I\subset \N$, we define
$\mathcal{R}_I(t_G)\in \calD^\prime(M^I)$ as in definition \ref{d:renormmapszeta} and extend it by linearity 
to the vector space $\mathcal{F}(M^I)$.

Then the collection of renormalization maps 
$\left(\mathcal{R}_{I}\right)_{I\subset \mathbb{N}, \vert I\vert <+\infty}$
satisfies the functional equations of definition \ref{d:Functionalequations}.
\end{thm}

\begin{proof} The compatibility condition is encoded in the family of projections.
For simplicity of notations, we choose to drop the subindex $s_0$ for the space of germs so we will write
$\calm,\calo$ instead of $\calm_{s_0},\calo_{s_0}$ and it will always be understood from the context that
we consider holomorphic and meromorphic germs localized at $s_0=(1,\dots,1)\in \C^p$ for some $p\in \N$.
Furthermore, we also
write $\pi$ for the projections instead of $\pi_{\vert E(G)\vert}$ where it will be understood
that for every graph $G$, $\pi(t_G(s))$ means $\pi_{\vert E(G)\vert}(t_G(s))$. 

 We now prove that
$\mathcal{R}_I(t_{G})$ is a distributional
extension of $t_{G}$. By Lemma \ref{keylemma}, on $M^I\setminus \Delta_I$, for every $e\in E(G)$, 
every
Green function $\greenf^{s_e}\in C^\infty(M^I\setminus \Delta_I,\calo) $ are in fact smooth and depend holomorphically on $s_e$. 
We also have the convergence
$\greenf^{s_e}\rightarrow \greenf$ in $C^\infty(M^2\setminus \Delta _2)$ when $s_e\rightarrow 1$.
Therefore,
for any test function
$\varphi\in C_c^\infty(M^I\setminus \Delta_I)$, by Corollary \ref{coro:pi}, 
$\pi(t_G(s))(\varphi)=\pi(t_G(s)(\varphi))=t_G(s)(\varphi)$ since $t_G(s)(\varphi)$ is holomorphic at $s_0=(s_e=1)_{e\in E(G)}\in \C^{E(G)}$, and 
$$\mathcal{R}_I(t_G)(\varphi)=\mathbf{ev}|_{s_0}\pi(t_G(s))(\varphi)=\mathbf{ev}|_{s_0}t_G(s)(\varphi)=t_G(\varphi).$$

Now let us prove the locality. 
For a graph $(G, \iota )$ with vertices labelled by $J\subset \N$, and $ I\subset J=\iota (V(G))$, set $I^c=J\setminus I$, 
let $E_I=\{e\in E(G) ; i(e),j(e)\in I^2\},E_{I^c}=\{e\in E(G) ; i(e),j(e)\in I^{c2}\},E_{II^c}=E(G)\setminus \left(E_I\cup E_{I^c}\right)$, and we denote by
$(G_I,G_{I^c},G_{II^c})$ the corresponding induced subgraphs of $G$.
Start from $t_G(s)=t_{G_I}(s_I)t_{G_{I^c}}(s_{I^c})t_{G_{II^c}}(s_{II^c})$ where
$s=(s_e)_{e\in E(G)}, s_I=(s_e)_{e\in E_I}, s_{I^c}=(s_e)_{e\in E_{I^c}}, s_{II^c}=(s_e)_{e\in E_{II^c}}$.
For a pair of disjoint open subsets $(U,V)$ such that $\text{dist}(U,V)>0$, 
denote by
$\{ (x_j)_{j\in J}\in M^J \text{ s.t. }x_i\in U, \forall i\in I, x_i\in V, \forall i\in I^c \}$
the open subset of configuration space $M^J$. Then note that
the product
$t_{G_{II^c}}(s_{II^c})=\prod_{e\in E_{II^c}} \greenf^{s_e}(x_{i(e)},x_{j(e)})\in C^\infty(U^I\times V^{I^c},\calo(\C^{E_{II^c}}))$. 
It follows by Lemma \ref{l:factorizationgmdtimesholo} proved in the appendix
that
$$\pi\left( t_G(s) \right)=\pi\left(t_{G_I}(s_I)t_{G_{I^c}}(s_{I^c})t_{G_{II^c}}(s_{II^c}) \right)=\pi\left(t_{G_I}(s_I)t_{G_{I^c}}(s_{I^c})\right)t_{G_{II^c}}(s_{II^c}).$$

Now the distributions $t_{G_I}(s_I)\in \mathcal{D}^\prime(U^I,\calm(\C^{E_I}))$ and $t_{G_{I^c}}(s_{I^c})\in \mathcal{D}^\prime(V^{I^c},\calm(\C^{E_{I^c}}))$ depend on different variables, therefore
by Lemma \ref{l:factorizationforgmd}, 
$\pi\left(t_{G_I}(s_I)\boxtimes t_{G_{I^c}}(s_{I^c})\right)=\pi\left(t_{G_I}(s_I)\right)\boxtimes 
\pi\left( t_{G_{I^c}}(s_{I^c})\right)$. Then as 
distributions on $U^I\times V^{I^c}$~:
\begin{eqnarray*}
\mathcal{R}_J(t_G)
&=&\mathbf{ev}|_{(s_e=1)_{e\in E(G)}}
\left(\pi\left(t_{G_I}(s_I)\right)
\boxtimes
\pi\left( t_{G_{I^c}}(s_{I^c})\right)
\times t_{G_{II^c}}(s_{II^c})\right)\\
&=&\mathbf{ev}|_{(s_e=1)_{e\in E_I}} \pi\left(t_{G_I}(s_I)\right)
\mathbf{ev}|_{(s_e=1)_{e\in E_{I^c}} } 
\pi\left( t_{G_{I^c}}(s_{I^c})\right) 
\mathbf{ev}|_{(s_e=1)_{e\in E_{II^c}}}t_{G_{II^c}}(s_{II^c})\\
&=&\mathcal{R}_{I}\left(t_{G_I} \right)\mathcal{R}_{I^{c}}\left(t_{G_{I^c}} \right)t_{G_{II^c}}|_{U^I\times V^{I^c}}
\end{eqnarray*}
where $t_{G_{II^c}}$ is smooth on $U^I\times V^{I^c}$ which yields the desired equation.
\end{proof}

\section{Appendix: technical details.}

\subsection{Proof of Proposition \ref{weakstrongholo}.}
\label{weakstrongholoproof}

\begin{proof}
Without loss of generality
assume that $z=0$. By definition and the multidimensional
Cauchy's formula~\cite[p.~3]{gunningrossi}, for any polydisc
$D_1\times \dots\times D_p$ around
$z=0$, for any test function $\varphi\in C ^\infty _c(M)$  and any $\lambda$ in the polydisc~:
\begin{eqnarray*}
t(\lambda)(\varphi)=\frac{1}{(2\pi i)^p}\int_{\partial D_1}\dots \int_{\partial D_p} \frac{t(z) (\varphi)dz_1 \dots dz_p}{(z_1-\lambda_1)\dots (z_p-\lambda_p)}
=\frac{1}{(2\pi i)^p}\int_{\partial D_1}\dots \int_{\partial D_p} \sum _\alpha \lambda ^\alpha \frac{t(z) (\varphi)dz_1 \dots dz_p}{z_1^{\alpha_1+1}\dots z_p^{\alpha_p+1}}.
\end{eqnarray*}
So for any multi-index $\alpha$ and any test function $\varphi$, we define the functional $t_\alpha $ by
$$t_\alpha(\varphi) = \frac{\alpha !}{(2\pi i)^p}\int_{\partial D_1}\dots\int_{\partial D_p} \frac{t(z) (\varphi)dz_1 \dots dz_p}{z_1^{\alpha_1+1}\dots z_p^{\alpha_p+1}},$$
then the series
$\sum_{\vert\alpha\vert\geqslant 0} \frac{s^\alpha}{\alpha !} t_\alpha(\varphi)$
converges absolutely to $t(s)(\varphi)$.

First note that the functional $t_\alpha $ is linear,
it remains to prove that $t_\alpha$ is continuous.
For that, it suffices to show that for every compact $K\subset M$, the restriction of $t_\alpha $ to
the Fr\'echet space $C^\infty_K(M)$ of test functions supported in $K$ is continuous.
For fixed $s$, $t(s)$ is linear continuous on
$C^\infty_K(M)$ therefore
there exists some constant $C(s)$ and continuous
seminorm $P$ of $C^\infty_K(M)$ such that
$\vert t(s)(\varphi) \vert\leqslant C(s)P(\varphi),\forall \varphi\in C^\infty_K(M) $.
Conversely for fixed $\varphi$,
$s\in D_1\times \dots\times D_p \mapsto t(s)(\varphi)$ is bounded
by holomorphicity.
By an application of the uniform boundedness principle since $C^\infty_K(M)$ is Fr\'echet,
for every compact
$K\subset M$, there exists $C>0$ and some continuous seminorm $P$ for the Fr\'echet topology
of $C^\infty_K(M)$ such that~:
\begin{eqnarray*}
\forall \varphi\in C^\infty_K(M), \sup_{s\in \partial D_1 \times\dots\times \partial D_p}\vert t(s)(\varphi)\vert\leqslant CP(\varphi).
\end{eqnarray*}

Assuming that all discs $\partial D_i$ have radius $r$, it immediately follows that
$t_\alpha$ satisfies a distributional version of Cauchy's bound:
\begin{equation}\label{Cauchybound}
\forall \varphi\in C^\infty_K(M), \vert t_\alpha(\varphi)\vert\leqslant \frac{\alpha!}{r^{\vert\alpha\vert}}  CP(\varphi).
\end{equation}
This also implies that for all $\varphi\in C^\infty_K(M)$, the power series
$\sum_{\alpha} \frac{s^\alpha}{\alpha !} t_\alpha(\varphi)$
converges for $\vert\lambda\vert <r$ i.e. the convergence radius equals $r$.
\end{proof}

\subsection{Products of \gmd in different variables.}
\label{p:productsofgmd}

In this part, we prove some useful Lemmas on products of \gmd in different variables.

\begin{lemm}\label{l:exteriorprod}
Let $(X_1,X_2)$ be smooth manifolds, $\mu_1 \in \R ^{p_1}\subset \mathbb{C}^{p_1}$ and $\mu _2 \in \R ^{p_2}\subset \mathbb{C}^{p_2}$.
If $t_1(s_1)\in \calD^\prime(X_1,\calm_{\mu_1}(\C^{p_1}))$ and $t(s_2)\in \calD^\prime(X_2,\calm_{\mu_2}(\C^{p_2}))$ 
then
the external tensor product
$t_1(s_1)\boxtimes t_2(s_2)$ is a well--defined element in 
$\mathcal{D}^\prime(X_1\times X_2,\calm_{(\mu_1,\mu_2)}(\C^{p_1+p_2}))$.
\end{lemm}

\begin{proof}
Denote by $dv_1,dv_2$ some smooth densities on $X_1,X_2$ respectively.
Since every compact subset $K\subset X_1\times X_2$,
can be
covered by some finite number of products of compacts of the form
$K_1\times K_2$, 
by Lemma \ref{l:gluing} it suffices to show
that for all compacts $K_1\subset X_1, K_2\subset X_2$,
the element $t(s_1;x)t(s_2;y)|_{K_1\times K_2}$ 
is a well--defined meromorphic family of distribution in $\mathcal{D}^\prime(K_1\times K_2)$
at $(\mu_1,\mu_2)\in \mathbb{C}^{p_1}\times \mathbb{C}^{p_2}$ with linear poles. 
Hence we can assume, without loss of generality, that we work
over some product $K_1\times K_2\subset X_1\times X_2$
of compact subsets and we
assume without loss of 
generality that we work around 
$(\mu _1,\mu _2)=(0,0)$.
There exists mononomials $P(s_1)=L_1(s_1)\dots L_k(s_1)$ 
and $Q(s_2)=M_1(s_2)\dots M_l(s_2)$, 
where $(L_i)_{i=1}^k,(M_i)_{i=1}^l$ are linear functions, such that
$P(s_1)t_1(s_1)$ and $Q(s_2)t_2(s_2)$ are holomorphic germs of distributions at $s_1=\mu_1,s_2=\mu_2$ 
respectively. 
Therefore by Proposition~\ref{weakstrongholo}, we know that $P(s_1)t_1(s_1)$ and $Q(s_2)t_2(s_2)$ both admit Laurent series expansions
$P(s_1)t_1(s_1)=\sum_{\alpha_1} s_1^{\alpha_1} u_{\alpha_1}$,
$Q(s_2)t_2(s_2)=\sum_{\alpha_2}s_2^{\alpha_2} v_{\alpha_2}$
where there exists two integers $(m_1,m_2)$ corresponding
to the distributional orders of $(t_1|_{K_1},t_2|_{K_2})$ and  two positive real numbers
$(r_1,r_2)$ such that
for all multi--index $(\alpha_1,\alpha_2)\in \N^{p_1+p_2}$, we have the bounds~:
\begin{eqnarray}
\Vert u_{\alpha_1} \Vert_{(C^{m_1})^\prime}\leqslant C_1r_1^{\vert\alpha_1\vert}, \ \,\
\Vert v_{\alpha_2} \Vert_{(C^{m_2})^\prime}\leqslant C_2r_2^{\vert\alpha_2\vert}.
\end{eqnarray}

We define the series
$P(s_1)t_1(s_1)\boxtimes Q(s_2)t_2(s_2)=\sum_{\alpha_1,\alpha_2} s_1^{\alpha_1}s_2^{\alpha_2} u_{\alpha_1}\boxtimes v_{\alpha_2} $ and we shall prove that the above series
converges for $\vert s_1\vert+\vert s_2\vert$ small enough
in the sense that
for every test function
$\varphi(x_1,x_2)$ supported in $K_1\times K_2$, the series
$$ \sum_{\alpha_1,\alpha_2} s_1^{\alpha_1}s_2^{\alpha_2} u_{\alpha_1}\boxtimes v_{\alpha_2}(\varphi)=\sum_{\alpha_1,\alpha_2} s_1^{\alpha_1}s_2^{\alpha_2} \int_{X_1\times X_2} u_{\alpha_1}(x_1) v_{\alpha_2}(x_2)\varphi(x_1,x_2)
dv_1(x_1)dv_2(x_2) $$
converges absolutely.
We first prove it for some element $\varphi=\varphi_1\boxtimes \varphi_2\in C^\infty_{K_1}(X_1)\boxtimes C^\infty_{K_2}(X_2)\subset 
C^\infty_{K_1\times K_2}(X_1\times X_2)$
which is a tensor product of two elements.
For $s\in \C^{p}$, we shall use the notation $\Vert s\Vert=\sup_{j\in \{1,\dots,p\}} \vert s_{j} \vert $
and for some multi--index $\alpha\in \N^{p}$, $\vert \alpha\vert=\sum_{j=1}^{p}\alpha_j$. 
Then the series converges thanks to the bound~:
\begin{eqnarray*}
&&\vert\sum_{\alpha_1,\alpha_2} s_1^{\alpha_1}s_2^{\alpha_2} \int_{X_1\times X_2} u_{\alpha_1}(x_1) v_{\alpha_2}(x_2)\varphi(x_1,x_2)
dv_1(x_1)dv_2(x_2)\vert\\
&\leqslant &
\sum_{\alpha_1,\alpha_2}\Vert s_1^{\alpha_1}s_2^{\alpha_2}\Vert \vert\int_{X_1} u_{\alpha_1}(x_1)\varphi_1(x_1)dv_1(x_1)\int_{X_2} v_{\alpha_2}(x_2)\varphi_2(x_2)dv_2(x_2)\vert\\
&\leqslant &
\sum_{\alpha_1,\alpha_2} \Vert s_1\Vert^{\vert\alpha_1\vert} C_1r_1^{\vert \alpha\vert} \Vert \varphi_1\Vert_{C^{m_1}(X_1)} 
\Vert s_2\Vert^{\vert\alpha_2\vert}C_2r_2^{\vert\alpha\vert} \Vert \varphi_2\Vert_{C^{m_2}(X_2)}\\
&\leqslant & 
\sum_{\alpha_1,\alpha_2} C_1(\Vert s_1\Vert r_1)^{\vert\alpha_1\vert}  C_2(\Vert s_2\Vert r_2)^{\vert\alpha_2\vert} \Vert \varphi\Vert_{C^m(X_1\times X_2)}
\end{eqnarray*}
for any $m\geqslant \sup(m_1,m_2)$ where the r.h.s is absolutely convergent for $s_1,s_2$ small enough.
Then we conclude by using the fact that the completed
tensor product
$C^\infty_{K_1}(X_1)\widehat{\boxtimes} C^\infty_{K_2}(X_2)$ 
coincides with $C^\infty_{K_1\times K_2}(X_1\times X_2) $ 
where the topology for which we do the completion does not matter
since $C^\infty_{K_i}(X_i)$ are Fr\'echet nuclear spaces. Therefore 
the algebraic tensor product $C^\infty_{K_1}(X_1)\boxtimes C^\infty_{K_2}(X_2)$ is dense in
$C^\infty_{K_1\times K_2}(X_1\times X_2) $ and the inequality
$$\vert\sum_{\alpha_1,\alpha_2} s_1^{\alpha_1}s_2^{\alpha_2} \int_{X_1\times X_2} u_{\alpha_1}(x_1) v_{\alpha_2}(x_2)\varphi(x_1,x_2)dv_1(x_1)dv_2(x_2)\vert
\leqslant \sum_{\alpha_1,\alpha_2} C_1(\Vert s_1\Vert r_1)^{\vert\alpha_1\vert}  C_2(\Vert s_2\Vert r_2)^{\vert\alpha_2\vert} \Vert \varphi\Vert_{C^m(X_1\times X_2)}
 $$ holds true for all $\varphi\in C^\infty_{K_1\times K_2}(X_1\times X_2)$.
\end{proof}

For every $p\in \C^p, s_0\in \R^p\subset \C^p$, let
$\pi_p:\calD^\prime(M,\calm_{s_0}(\C^p))\mapsto \calD^\prime(M,\calo_{s_0}(\C^p))$ be the projection from Proposition \ref{p:proj}.
Then~:
\begin{lemm}\label{l:factorizationforgmd}
Under the assumptions of the previous Lemma, 
the following equation holds true~:
\begin{eqnarray}\label{e:factorizationgmd}
\boxed{\pi_{p_1+p_2}\left(t_1\boxtimes t_2\right)=\pi_{p_1}(t_1) \boxtimes \pi_{p_2}(t_2).}
\end{eqnarray}
\end{lemm}
\begin{proof}
The proof of equation \ref{e:factorizationgmd} goes as follows,
we decompose $t_1$ and $t_2$ as $t_1=\pi_{p_1}(t_1)+(1-\pi_1)(t_1)$ and 
$t_2=\pi_{p_2}(t_2)+(1-\pi_2)(t_2)$ where $(\pi_{p_1}(t_1),\pi_{p_2}(t_2))\in 
\calD^\prime(X_1,\calo_{\mu_1}(\C^{p_1}))\times \calD^\prime(X_2,\calo_{\mu_2}(\C^{p_2}))$ and
$((1-\pi_1)(t_1),(1-\pi_2)(t_2))\in \calD^\prime(X_1,\mathcal{P}_{\mu_1}(\C^{p_1}))\times 
\calD^\prime(X_2,\mathcal{P}_{\mu_2}(\C^{p_2}))$.
Then note that
\begin{eqnarray*}
t_1\boxtimes t_2&=&
\underset{\in \calD^\prime(X_1\times X_2,\calo_{(\mu_1,\mu_2)}(\C^{p_1+p_2}))}{\underbrace{\pi_{p_1}(t_1)\boxtimes \pi_{p_2}(t_2)}} \\ 
&+& \underset{\in 
\calD^\prime(X_1\times X_2,\mathcal{P}_{(\mu_1,\mu_2)}(\C^{p_1+p_2}))}{\underbrace{
(1-\pi_1)(t_1)\boxtimes\pi_2(t_2) + \pi_1(t_1)\boxtimes(1-\pi_2)(t_2)     +(1-\pi_1)(t_1)\boxtimes(1-\pi_2)(t_2) }}
\end{eqnarray*}
where the term $
(1-\pi_1)(t_1)\boxtimes\pi_2(t_2) + \pi_1(t_1)\boxtimes(1-\pi_2)(t_2)     +(1-\pi_1)(t_1)\boxtimes(1-\pi_2)(t_2) $
is a finite sum of polar germs by equation
(\ref{e:decompgmddetailed}). It follows that
$\pi_{p_1+p_2}\left(t_1\boxtimes t_2\right)=\pi_{p_1}(t_1) \boxtimes \pi_{p_2}(t_2)$ by the uniqueness of the decomposition
which follows from Theorem \ref{t:decompgmd}.
\end{proof}

By a similar proof as in the above Lemma, we also have~:
\begin{lemm}\label{l:factorizationgmdtimesholo}
Let $X$ be a smooth manifold, $U\subset X$ an open subset and $m\in \N$.
Set $(p_1,p_2)\in \mathbb{N}^2$ to be an arbitrary pair
of integers and $(\mu_1,\mu_2)\in \R^{p_1}\times \R^{p_2}\subset \C^{p_1+p_2}$. 
Let $t(s_1)\in  \mathcal{D}^{\prime,m}(U,\calm_{\mu_1}(\C^{p_1}))$ and $h(s_2)\in C^m(U,\calo_{\mu_2}(\C^{p_2}))$ .
Then the product $t(s_1)h(s_2)$ is an element
of $\mathcal{D}^{\prime,m}(U,\calm_{(\mu_1,\mu_2)}(\C^{p_1+p_2})) $
which
satisfies the equation~:
\begin{eqnarray}
\boxed{ \pi_{p_1+p_2}(t(s_1)h(s_2))=\pi_{p_1}(t(s_1))h(s_2) .}
\end{eqnarray}
\end{lemm}

\begin{proof}
Without loss of generality, we can work locally since
all local results can be glued together by partition of unity
thanks to Lemma \ref{l:gluing}.
For every 
test function $\varphi\in C^\infty_c(U)$,
$\langle t(s_1) h(s_2),\varphi \rangle=\langle t(s_1), \underset{\in C^m_c(U)}{\underbrace{ h(s_2)\varphi}} \rangle$
hence the product $t(s_1) h(s_2)$ is well defined in
$\calD^{\prime,m}(U)$
as soon as both $t(s_1),h(s_2)$ exist.
We now explain the meromorphicity of $(s_1,s_2)\mapsto \langle t(s_1) h(s_2),\varphi \rangle$
at $(\mu_1,\mu_2)\in \C^{p_1+p_2}$.
Since $t(s)\in \calD^{\prime,m}(U,\calm_{\mu_1}(\C^{p_1}))$, there exists
$u(s)\in \calD^{\prime,m}(U,\calo_{\mu_1}(\C^{p_1}))$ and linear functions
$(L_1,\dots,L_k)$
such that $\left( L_1(s) \cdots L_k(s)\right)t(s) = u(s)$.
Therefore the product 
$t(s_1) h(s_2)$ also reads
$ \frac{1}{ L_1(s_1) \cdots L_k(s_1)} u(s_1)h(s_2) $.
Then using power expansions in $s_1-\mu_1$ for $u(s_1)$ as in
Theorem \ref{weakstrongholo} and expanding $h(s_2)$ in powers of $s_2-\mu_2$
where coefficients are in $C^m(U)$, 
we easily show that
$u(s_1)h(s_2) \in  \calD^{\prime,m}(U,\calo_{(\mu_1,\mu_2)}(\C^{p_1+p_2}))$
for $(s_1,s_2)\in \C^{p_1+p_2}$ close enough to
$(\mu_1,\mu_2)\in \C^{p_1+p_2}$ which proves
$t(s_1)h(s_2)\in\mathcal{D}^{\prime,m}(U,\calm_{(\mu_1,\mu_2)}(\C^{p_1+p_2})) $.

The equation $ \pi_{p_1+p_2}(th)=h\pi_{p_1}(t)$ immediately follows
from the fact that $\pi_{p_2}(h(s_2))=h(s_2)$ since $h$ is holomorphic 
and $h(s_2)(1-\pi_{p_1})(t(s_1))$ is valued in polar germs.
\end{proof}

\subsection{Proof of Lemma \ref{keylemma}.}
\label{a:proofkeylemma}

\begin{proof}
Since our Riemannian manifold $(M,g)$ is connected, $\ker(P)$
contains only constant functions.
Indeed $P u=0$ implies that $u\in C^\infty$ by elliptic regularity and $ 0 =\langle u, -\Delta_g u\rangle + \langle u,Vu\rangle \implies\langle \nabla u,\nabla u \rangle =0\implies \nabla u=0$
thus $u$ is constant on connected components.
Let us determine explicitly the spectral projector
$\Pi$, it should satisfy for all $u$~:
$$0= \langle 1,u-\Pi(u) \rangle=\int_M (u-\Pi(u))=\int_M udx -\Pi(u)\text{Vol}(M)\implies \Pi(u)= \frac{\int_M udx}{\text{Vol}(M)}.$$
The Schwartz kernel of the spectral projector $\Pi$ is therefore the constant function $\Pi(x,y)=\text{Vol}(M)^{-1}$.

The first two claims about the Schwartz kernel $\greenf ^{s}(x,y)$ follow from~\cite[Theorem 4 p.~302]{Seeley}
in the celebrated work of Seeley, by applying his Theorem to $A=P-\Pi$ which is a well defined
elliptic pseudodifferential operator of order $2$.

For the third claim, we start from the formula
$\greenf ^{s}=\int_0^\infty\left(e^{-tP}-\Pi\right)(x,y)t^{s-1}dt$
and our proof exactly follows the proof of~\cite[Proposition 1]{BarMoroianu}
where we replace the heat semigroup $e^{t\Delta_g}$ in their proof by the semigroup
$(e^{-tP}-\Pi)$ whose Schwartz kernel is  $K_t-\Pi$ and is denoted by $p_t$.
Start
from the formula
$p_t(x,y)=\langle \delta_x,(e^{-tP}-\Pi)\delta_y \rangle_{L^2(M)}=\langle (e^{-\frac{t}{2}P}-\Pi)\delta_x,(e^{-\frac{t}{2}P}-\Pi)\delta_y \rangle_{L^2(M)}$.
For any integers $(k,l,m)$,
$\vert \partial^m_t P_{x}^kP_{y}^l p_t(x,y)\vert=\vert  P_{x}^{k+m}P_{y}^l p_t(x,y)\vert$
since $\partial_t^m\left(e^{-tP}-\Pi\right)=P^m\left(e^{-tP}-\Pi\right)$. Hence,
$$\vert \partial^m_t P_{x}^kP_{y}^l p_t(x,y)\vert\leqslant \Vert (e^{-(t-\varepsilon)P}-\Pi)\Vert_{B(L^2(M))} \Vert P_{x}^{k+m}
(e^{-\frac{\varepsilon}{2}P}-\Pi)\delta_x\Vert_{L^2(M)}
\Vert P_{y}^l\left(e^{-\frac{\varepsilon}{2}P}-\Pi\right)\delta_y\Vert_{L^2(M)}.$$
Therefore taking the supremum over $(x,y)\in M\times M$ yields~:
\begin{eqnarray*}
\Vert \partial^m_t P_{x}^kP_y^l p_t\Vert_{C^0(M\times M)}\leqslant \Vert (e^{-(t-\varepsilon)P}-\Pi)\Vert_{B(L^2(M))} \Vert P_x^{k+m}
(e^{-\frac{\varepsilon}{2}P}-\Pi)\delta_x\Vert_{L^2(M)}
\Vert P_{y}^l\left(e^{-\frac{\varepsilon}{2}P}-\Pi\right)\delta_y\Vert_{L^2(M)}
\end{eqnarray*}
where both $\Vert P_{x}^{k+m}
(e^{-\frac{\varepsilon}{2}P}-\Pi)\delta_x\Vert_{L^2(M)}$ and 
$\Vert P_{y}^l\left(e^{-\frac{\varepsilon}{2}P}-\Pi\right)\delta_y\Vert_{L^2(M)}$ are finite
since both $(e^{-\frac{\varepsilon}{2}P}-\Pi)\delta_x$ and $(e^{-\frac{\varepsilon}{2}P}-\Pi)\delta_y$
are smooth functions because 
the semigroup $(e^{tP}-\Pi)_{t\in\R_{\geqslant 0}}$ is smoothing.
Furthermore, the term $\Vert (e^{-(t-\varepsilon)P}-\Pi)\Vert_{B(L^2(M))}$
has
exponential decay when $t\rightarrow +\infty$
since $\left(e^{-(t-\varepsilon)P}-\Pi\right)$ is a smoothing operator which has
a gap in the spectrum, indeed by spectral theory
$e^{-tP}u=\sum_{\lambda\in \sigma(P)}e^{-t\lambda}\Pi_\lambda(u)$ where $\Pi_\lambda$ is the spectral projector
on the eigenspace of eigenvalue $\lambda$ and the
r.h.s.
converges absolutely in all Sobolev spaces $H^s(M),s\geq 0$ when $t>0$.
More generally, we obtain decay estimates of the form
\begin{eqnarray*}
\Vert \partial^m_tp_t \Vert_{C^k(M\times M)}&\leqslant &\sum_{l_1,l_2\leqslant \frac{k}{2}+1}
\Vert\partial^m_t P_{x}^{k_1}P_{y}^{k_2} p_t \Vert_{C^0(M\times M)}\\
&\leqslant & C_{k,m} \Vert (e^{-(t-\varepsilon)P}-\Pi)\Vert_{B(L^2(M))}\leqslant C_{k,m} e^{-(t-\varepsilon)\lambda_1}
\end{eqnarray*}
where $\lambda_1>0$ is the smallest non zero eigenvalue of $P$ which exists since
$\sigma(P)$ is a discrete subset of $[0,+\infty)$. It follows that
the integral $\int_1^\infty t^{s-1}p_tdt$ converges absolutely for all $s\in \mathbb{C}$
and is valued in all Banach spaces $C^k(M\times M),k\in\mathbb{N}$ since~:
\begin{eqnarray*}
\Vert \int_1^\infty t^{s-1}p_tdt \Vert_{C^k(M\times M)}\leqslant
\int_1^\infty t^{Re(s)-1}\Vert p_t \Vert_{C^k(M\times M)}dt\leqslant C_k\int_1^{\infty} t^{Re(s)-1}e^{-(t-\varepsilon)\lambda_1}dt.
\end{eqnarray*}
The integral $\int_1^\infty t^{s-1}p_tdt  $ depends holomorphically in $s$
since
\begin{eqnarray}
\Vert \int_1^\infty \left(\frac{d}{ds}\right)^l t^{s-1}p_tdt \Vert_{C^k(M\times M)}\leqslant
C_k\int_1^{\infty} t^{Re(s)-1}\log(t)^le^{-(t-\varepsilon)\lambda_1}dt
\end{eqnarray}
where the r.h.s. is absolutely convergent and we can conclude by dominated convergence arguments.
\end{proof}

\subsection{Proof of Lemma \ref{l:analyticcontinuationmodelcase}.}

\begin{proof} First notice that when $Re(s_i)> -1$, $i=1, \cdots , E$, this integral is absolutely convergent and holomorphic in $s$.

Now If $E=1$, then by integration by parts, for $Re(s)> -1$,
$$\int_{[0,1]} t^{s}\psi(t)dt=\sum _{i=0}^{k-1} (-1)^{i}\frac 1{l_i(s)}\psi^{(i)} (1)+(-1)^k\frac 1{l_{k-1}(s)}\int_{[0,1]} t^{s+k}\psi ^{(k)}(t)dt,
$$
where $l_i(s)=(s+1)\cdots (s+i+1)$, the l.h.s is a meromorphic function when $ Re(s)> -k-1$ with possible poles at $s=-1, \cdots, -k$, so it extends to a meromorphic function on $ Re(s)> -k-1$. 

In general,
for $Re(s_i)> -1$, $i=1, \cdots, E$, and $k_1, \cdots , k_E\in \Z _{>0}$,
\begin {eqnarray}
I_s(\psi)&=&\int_{[0,1]^E} t_1^{s_1}\dots t_E^{s_E}\psi(t_1,\dots,t_E)d^Et \notag \\
&=&\sum _{\{j_1, \cdots, j_m\}\subset \{1, \cdots , E\}} \sum _{\overset {j\not =j_1, \cdots j_m}{i_j=0, \cdots , k_j-1}}\frac {(-1)^{i_j}}{ l_{i_j}(s_j)}\prod ^m\frac {(-1)^{k_{j_i}}}{ l_{k_{j_i}-1}(s_{j_i})} \notag\\
&&\int _{[0,1]^m}\prod _{j=j_1, \cdots j_m}t^{s_j+k_j}\big (\prod _{j\not =j_1, \cdots j_m}\partial_{t_{j}}^{i_j}\big)\partial_{t_{j_1}}^{k_{j_1}}\cdots \partial _{ t_{{j_m}}}^{k_{j_m}}\psi |_{t_j=1,j\not =j_1, \cdots j_m}dt_{j_1}\cdots dt_{j_s}
\label {eqn:IntByParts}
\end{eqnarray}
the r.h.s is a meromorphic function when $ Re(s_i)> -k_i-1$. So $I_s(\psi)$ extends to a meromorphic germ at any point in $\Z ^E$.

Now at a given point $(p_e)_e\in \Z ^E$, $\frac 1{s_e-a_e}$ is holomorphic except $a_e=p_e$, therefore
$$
\left (\prod _{i\in I}(s_i-p_i)\right )I_s(\psi)
$$
is a holomorphic germ at $(p_e)_e$. The distribution order of $I_s(\psi)$ at the point $(p_e)$ can be read from Equation (\ref {eqn:IntByParts}) easily.
\end{proof}

\subsection{Proof of Lemma \ref{lemmametric}.}
\label{p:lemmametric}

\begin{proof}
In the chart $(U\times U, (x^\mu,y^\nu))$, let us consider the Taylor expansion of $\phi (x,y)$,
$\phi(x,y)=\sum_{k\geqslant 0}\phi_{[k]}(x,y),$ where
$\phi_{[k]}(x,y)=\sum_{\vert\alpha\vert+\vert\beta\vert=k}\frac{x^\alpha y^\beta}{\alpha!\beta!}\partial_x^\alpha\partial_y^\beta \phi(0,0)$.
\\
Obviously $\phi_{[0]}(x,y)=0$. \delete {Let us take $x=y$ in Equation (\ref {e:hadamarddist}),
$$g^{-1}(x)(\frac {\partial \phi}{\partial x^\mu}(x,x)dx^\mu , \frac {\partial \phi}{\partial x^\nu}(x,x)dx^\nu)=0$$
which means
$$\frac {\partial \phi}{\partial x^\mu}(x,x)dx^\mu\equiv 0\Rightarrow \frac {\partial \phi}{\partial x^\mu}(x,x)\equiv 0 .
$$
By symmetry
$$\frac {\partial \phi}{\partial y^\mu}(x,x)\equiv 0,$$
So} 
By symmetry and $\phi(x,x)=0$, we know that
$$\phi_{[1]}(x,y)=0.
$$
By symmetry and $\phi(x,x)\equiv 0$, we know that
$\phi_{[2]}(x,y)=\sum _{\mu } a_\mu (x^\mu-y^\mu)^2,$
now by the fact $\phi (0,y)=\| y\|^2$, we find that~:
\begin{equation}\label{e:phi2}
\phi_{[2]}(x,y)=\sum (x^\mu -y^\mu)^2
.\end{equation}
\\
In fact, let us take $x=y$ in Equation (\ref {e:hadamarddist}),
$$g^{-1}(x)(\frac {\partial \phi}{\partial x^\mu}(x,x)dx^\mu , \frac {\partial \phi}{\partial x^\nu}(x,x)dx^\nu)=0$$
which means
$$\frac {\partial \phi}{\partial x^\mu}(x,x)dx^\mu\equiv 0\Rightarrow \frac {\partial \phi}{\partial x^\mu}(x,x)\equiv 0 .
$$
By symmetry
\begin{equation}
\label {eqn:pphiy}\frac {\partial \phi}{\partial y^\mu}(x,x)\equiv 0.
\end{equation}
\\
Now let us make a change of variables on $V\times W \to U\times U$ given by
$$(v,h)\mapsto (v,v+h), $$
we can take $V, W$ small enough such that $V\times W$ is a coordinate chart around $(x_0,x_0)$.
\\
Let $\tilde{\phi}(v,h)=\phi(v,v+h)$.
Take a partial Taylor expansion in $h$ for $\tilde \phi$,
$$\tilde{\phi}(v,h)=\tilde \phi (v,0)+ \frac {\partial \tilde \phi}{\partial h^\mu}(v,0)h^\mu +\frac 12 \frac {\partial ^2 \tilde{\phi}}{\partial h^\mu \partial h^\nu } (v,0)h^\mu h^\nu +\varepsilon_3$$
\\
 where $\varepsilon_3$ vanishes at order $3$ in $h$.
\\
We know
$$\tilde \phi (v,0)=\phi (v,v)=0,
$$
by Equation (\ref {eqn:pphiy}),
$$\frac {\partial \tilde \phi}{\partial h^\mu}(v,0)=\frac {\partial \phi}{\partial y^\mu}(v,v)=0.
$$
By chain rule,
$$\frac {\partial ^2 \tilde{\phi}}{\partial h^\mu \partial h^\nu } (v,0)=\frac {\partial ^2 {\phi}}{\partial y^\mu \partial y^\nu } (v,v)
$$
Equation~(\ref{e:hadamarddist}) shows~:
$$\frac {\partial \phi}{\partial x^\mu}(x,y)g^{\mu \nu}(x)\frac {\partial \phi}{\partial x^\nu}(x,y)=4\phi (x,y).
$$
Taking $\frac {\partial^2 }{\partial x^{\mu _1}\partial x^{\nu _1}}$ on both sides and let $x=y=v$, we have
$$\frac {\partial^2 \phi}{\partial x^\mu\partial x^{\mu _1} }(v,v)g^{\mu \nu}(v)\frac {\partial^2 \phi}{\partial x^\nu \partial x^{\nu _1}}(v,v)+\frac {\partial^2 \phi}{\partial x^\mu\partial x^{\nu _1} }(v,v)g^{\mu \nu}(v)\frac {\partial^2 \phi}{\partial x^\nu \partial x^{\mu _1}}(v,v)=4\frac {\partial^2 \phi}{\partial x^{\mu _1}\partial x^{\nu _1} }(v,v).
$$
that is
$$\frac {\partial^2 \phi}{\partial x^\mu\partial x^{\mu _1} }(v,v)g^{\mu \nu}(v)\frac {\partial^2 \phi}{\partial x^\nu \partial x^{\nu _1}}(v,v)=2\frac {\partial^2 \phi}{\partial x^{\mu _1}\partial x^{\nu _1} }(v,v).
$$
Notice that
$\frac {\partial^2 \phi}{\partial x^{\mu}\partial x^{\nu } }(v,v)$ is invertible since $\frac {\partial^2 \phi}{\partial x^{\mu }\partial x^{\nu } }(0,0)=\delta_{\mu\nu} $ by (\ref{e:phi2})
and if $U$ is chosen small enough.
Then we get that
$$\frac {\partial^2 \phi}{\partial x^\mu \partial x^{\nu}}(v,v)=2g_{\mu \nu}(v).
$$
Since $\phi $ is symmetry, we know
$$\frac {\partial^2 \phi}{\partial y^\mu \partial y^{\nu}}(v,v)=\frac {\partial^2 \phi}{\partial x^\mu \partial x^{\nu}}(v,v)=2g_{\mu \nu}(v).
$$
So
$$\tilde{\phi}(v,h)=g_{\mu \nu}(v)h^\mu h^\nu +\varepsilon_3$$
which concludes the proof.
\end{proof}

\end{document}